\documentclass[11pt,letter]{article}
\usepackage{HLT23-Packages}

\title{Functional Lower Bounds in Algebraic Proofs: \\
Symmetry, Lifting, and Barriers\footnotemark[1]}

\author{
Tuomas Hakoniemi\footnotemark[2]
\smallskip
\\
\small University of Helsinki 
\and 
Nutan Limaye%
\footnotemark[3]
\smallskip
\\
\small IT University of Copenhagen 
\and
Iddo Tzameret\footnotemark[5]
\\ 
 \small Imperial College London
}

\date{}

\usepackage{amssymb}
\usepackage{color}
\usepackage{amsmath}
\usepackage[normalem]{ulem}

\renewcommand{\thefootnote}{\fnsymbol{footnote}}

\begin{document}


\clearpage\maketitle\thispagestyle{empty}

\maketitle
 
\begin{abstract}
Strong algebraic proof systems such as IPS (Ideal Proof System; 
Grochow-Pitassi~\cite{GP18}) offer a general model for 
deriving polynomials in an ideal and refuting unsatisfiable 
propositional formulas, subsuming most standard propositional 
proof systems. A major approach for lower bounding the size of 
IPS refutations is the Functional Lower Bound Method (Forbes, 
Shpilka, Tzameret and Wigderson~\cite{FSTW21}), which reduces 
the hardness of refuting a polynomial equation $f(\vx)=0$ with 
no Boolean solutions to the hardness of computing the function 
$\nicefrac{1}{f(\vx)}$ over the Boolean cube with an algebraic circuit. 
Using symmetry, we provide a general way to obtain many new hard 
instances against fragments of IPS via the functional lower 
bound method. This includes hardness over finite fields and  hard 
instances different from Subset Sum variants, both of 
which were unknown before, and stronger constant-depth lower bounds.
Conversely, 
we expose the limitation of this method by showing it 
\mbox{cannot} lead to proof complexity lower bounds for any 
hard \textit{Boolean} instance (e.g., CNFs) for any
sufficiently strong proof systems. Specifically, we show the 
following: 

\begin{description}[leftmargin=6pt]
\item \textbf{Nullstellensatz degree lower bounds using 
    symmetry}: Extending \cite{FSTW21} we show that every unsatisfiable symmetric polynomial with $n$ variables 
    requires degree $>n$ refutations
    (over sufficiently large characteristic).
    Using symmetry again, by characterising the
    $\nicefrac{n}{2}$-homogeneous slice appearing in refutations, we give examples of unsatisfiable \emph{invariant} polynomials of degree $\nicefrac{n}{2}$ that require degree $\ge n$ refutations.   

\item \textbf{Lifting to size lower bounds}: Lifting our Nullstellensatz  
    degree bounds to  IPS-size lower bounds, we obtain exponential lower 
    bounds for any  poly-logarithmic degree symmetric instance 
    against IPS refutations
    written as oblivious read-once algebraic programs (roABP-IPS).   
    For invariant polynomials, we show lower bounds against roABP-IPS
    and refutations written as multilinear formulas in the \emph{placeholder} IPS regime
    (studied by Andrews-Forbes \cite{AF22}), where the hard instances
    do not necessarily have small roABPs themselves, including over
    \emph{positive characteristic} fields. This provides the first explicit example of a hard instance against 
    IPS fragments over finite fields. 

    By an adaptation of the work of Amireddy, Garg, Kayal, Saha and Thankey~\cite{AGK0T23}, we extend and strengthen the constant-depth IPS lower bounds obtained recently in Govindasamy, Hakoniemi and Tzameret~\cite{GHT22} which held only for multilinear proofs, to \emph{$\poly(\log \log n)$ individual degree} proofs. 
    This is a natural and stronger constant depth proof system than in \cite{GHT22}, which we show admits small refutations for standard hard instances like the pigeonhole principle and Tseitin formulas.  

\item \textbf{Barriers for  Boolean instances}: 
    While lower bounds against strong propositional proof
    systems were the original motivation for studying algebraic     
    proof systems in the 1990s (Beame \textit{et al.}~\cite{BeameIKPP96} and Buss \textit{et al.}~\cite{BussIKPRS96}),  
    we show that the functional lower bound method \textit{alone} cannot establish any size lower bound for \textit{Boolean} instances for any sufficiently strong proof systems, and in particular, cannot lead to lower bounds against \ACZ$[p]$-Frege and \TCZ-Frege.  
\end{description}
\end{abstract}

\footnotetext[1]{A preliminary abbreviated version of this work appears in STOC 2024. The preliminary version included a quantitative strengthening of the lower bound against multilinear constant-depth IPS refutations from \cite{GHT22}. The current version extends and strengthens the lower bound in  \cite{GHT22} to work against $O(\log \log n)$ individual degree constant-depth IPS refutations.}
\footnotetext[2]{Part of this work was done at Simons Institute for the Theory of Computing. This work was partly funded by Helsinki Institute of Information Technology (HIIT) and by the European Research Council (ERC) under the European Union's Horizon 2020 research and innovation programme (grant agreement No.~101002742).}
\footnotetext[3]{Part of this
work was done at Simons Institute for the Theory of Computing, UC Berkeley. Part of this project has received funding from the Independent Research Fund Denmark (grant agreement No. 10.46540/3103-00116B).
Part of this work was supported by Basic Algorithms Research Copenhagen (BARC) which is funded by VILLUM Foundation Grant 16582.}
\footnotetext[5]{Department of Computing. Part of this
work was done at Simons Institute for the Theory of Computing, UC Berkeley. Part of this project has received funding from the European Research Council (ERC) under the European Union's Horizon 2020 research and innovation programme (grant agreement No.~101002742). Email: \text{iddo.tzameret@gmail.com}}
\tableofcontents 
 
\renewcommand*{\thefootnote}{\arabic{footnote}}
\newpage
\section{Introduction}


This work studies lower bounds against strong propositional proof systems. We focus   on strong algebraic proof systems that extend the 
Nullstellensatz system, and explore the capabilities and limitations of the \emph{functional lower bound method} which is arguably the most successful technique to date for achieving such lower bounds. 
Our goal is to expand the range of applications of this technique to
yield new hard instances, as well as to provide new variants of this 
technique showing how to get lower bounds in new settings like finite 
fields, hard instances qualitatively different from previously known ones, and 
improved constant depth lower bounds. Finally, we ask whether 
this technique alone can lead to the solution of long-standing open problems 
in proof complexity, showing essentially it cannot reach this goal.

\subsection{Proof Complexity and Strong Algebraic Proof Systems}
\label{sec:first-section-in-intro}

%
%
%
%
%
%
%

Algebraic proof systems model efficient derivation within polynomial ideals. Starting from a set of polynomials over a field (the set of ``axioms''), one uses addition and multiplication to form new polynomials in the ideal generated by the axioms. When the axioms are unsatisfiable, namely do not have a common root over the field, by Hilbert's Nullstellensatz one can derive the polynomial $1$. This way, one has in fact \emph{refuted}  the set of axioms, since a common root of the axioms must nullify every polynomial in their ideal. 

Due to its natural setup, algebraic proof systems attracted a lot of attention in complexity and specifically propositional proof complexity, in which one studies the efficiency with which  different proof systems prove propositional tautologies or refute unsatisfiable propositional formulas. For that purpose, one can consider algebraic proof systems as propositional proof systems by adding as a  default to the set of axioms \emph{Boolean axioms} like $x^2_i-x_i$ for all variables $x_i$. These Boolean axioms force any common root of the set of axioms to be Boolean in itself. 

The starting point of algebraic proof systems is the work of Beame, Impagliazzo, Kraj\'{i}\v{c}ek, Pitassi and Pudl\'ak  
\cite{BeameIKPP96} (cf.~Buss, Impagliazzo, 
Kraj{\'{\i}}{\v{c}}ek, Pudl{\'{a}}k and Razborov~\cite{BussIKPRS96}), which was motivated by the long-standing open problem of $\ACZ[p]$-Frege lower bounds. The $\ACZ[p]$-Frege proof system operates with constant-depth propositional formulas together with counting modulo $p$ gates, for $p$ a prime. Initial results on algebraic proof systems established connections between these proof systems and algebraic proof systems.

The algebraic proof system introduced by \cite{BeameIKPP96}
is the \emph{Nullstellensatz} system. In Nullstellensatz a refutation witnessing the unsatisfiability of a set of axioms given as polynomial equations  $\{f_i(\vx)=0\}_i$ over a field, is a polynomial combination of the axioms that equals 1 as a formal polynomial, namely:  
\begin{equation}
\label{eq:NS-first}
\sum_i g_i(\vx)\cd f_i(\vx) = 1\,,
\end{equation}
for some polynomials $\{g_i(\vx)\}_i$. 
The degree of this  Nullstellensatz refutation is the maximal degree of $g_i(\vx)\cd f_i(\vx)$. The size of this refutation is the sparsity, that is the total number of monomials in all the polynomials $g_i(\vx)\cd f_i(\vx)$.
The sparsity measure is what makes these proof systems weak (e.g., even a simple polynomial like $(x_1-1)\cdots(x_n-1)=0$ accounts for an exponential size because the number of monomials in it is $2^n$).

Although the initial motivation behind the introduction of 
the Nullstellensatz system was to achieve progress on $\ACZ[p]$-Frege lower bounds, it seems that Nullstellensatz lower bounds techniques that focus on the number of monomials or degree are not enough to reach lower bounds for relatively strong proof systems (including  $\ACZ[p]$-Frege). For that purpose, one can consider stronger algebraic proof systems for which size is measured by \emph{algebraic complexity}, namely by the minimal size of an \emph{algebraic circuit} computing the polynomials $g_i(\vx)$ in \Cref{eq:NS-first}. That way, one can hope to employ ideas from algebraic circuit complexity, in a similar way that proof techniques from Boolean circuit complexity like random restrictions and switching lemmas serve to study  \ACZ-Frege lower bounds \cite{Ajt88,PBI93,KPW95}. 

The idea of considering algebraic proof systems operating with algebraic circuits was investigated initially by Pitassi \cite{Pit97,Pit98}. And subsequently was studied in Grigoriev and Hirsch \cite{GH03}, Raz and Tzameret \cite{RT06,RT07,Tza11-I&C}, and finally in the introduction of the Ideal Proof System (IPS) by Grochow and Pitassi \cite{GP18} which loosely speaking is the  \emph{Nullstellensatz proof system where the polynomials  $g_i(\vx)$ are written as algebraic circuits}; indeed, Forbes, Shpilka, Tzameret and Wigderson~\cite{FSTW21} showed that IPS is equivalent to Nullstellensatz in which the polynomials $g_i$ in \Cref{eq:NS-first} are written as algebraic circuits. In other words, an IPS refutation of the set of axioms $\{f_i(\vx)=0\}_i$ can be defined similarly to \Cref{eq:NS-first} (here we display explicitly the Boolean axioms): 
\begin{equation}
\label{eq:intro:NS-second}
\sum_i g_i(\vx)\cd f_i(\vx) + \sum_j h_j(\vx)\cd (x_j^2-x_j)
= 1\,,
\end{equation}
for some polynomials $\{g_i(\vx)\}_i$, where we think of the polynomials $g_i, h_j$ written as algebraic circuits (instead of e.g., counting the number of monomials they have towards the size of the refutation). Thus, the size of the IPS refutation in \Cref{eq:intro:NS-second} is $\sum_i \textsf{size}(g_i(\vx))+\sum_j \textsf{size}(h_j(\vx))$, where $\textsf{size}(g)$ stands for the (minimal) size of an algebraic circuit computing the polynomial $g$. 

It turns out that IPS is very strong, and simulates most known concrete propositional proof systems (such as Frege, Extended Frege, and more); see \cite{GP18,PT16}.
It is thus natural to consider proof systems that sit  between the weak Nullstellensatz on the one end and the strong IPS on the other end. This is done roughly by writing the polynomials $g_i, h_j$ in \Cref{eq:intro:NS-second} as restricted kinds of algebraic circuits, such as constant-depth circuits \cite{GP18,AF22,GHT22}, noncommutative formulas \cite{LTW18}, algebraic branching programs \cite{FSTW21,Kno17} and multilinear formulas \cite{FSTW21}, to give some examples.

When considering algebraic circuit classes weaker than general algebraic circuits, one has to be a bit careful with the definition of IPS. For technical reasons the formalization in \Cref{eq:intro:NS-second} does not capture the precise definition of IPS restricted to the relevant circuit class, rather the fragment which is denoted by $\cC$-\lIPS\ (``\begin{small}LIN\end{small}'' here stands for the linearity of the axioms $f_i$ and the Boolean axioms; that is, they appear with power 1). In this work, we focus on $\cC$-\lIPS\ and a similar stronger variant denoted $\cC$-\lbIPS. 
\emph{Henceforth, throughout the introduction, refutations in the system $\cC$-\lIPS\ are  defined  as in \Cref{eq:intro:NS-second} where the polynomials $g_i, h_j$ are written as circuits in the circuit class $\cC$.} 

\noindent\textbf{Technical comment}: All our lower bounds are   proved by lower bounding the algebraic circuit size of the $g_i$'s  in \Cref{eq:intro:NS-second}, namely the products of the axioms $f_i$, and \emph{not} the products of the Boolean axioms (that is, we ignore the circuit size of the $h_i$'s).
For this reason, our  lower bounds are slightly stronger than lower bounds  on $\cC$-\lIPS, rather they are lower bounds on the system denoted  $\cC$-\lbIPS\ (see \Cref{def:IPS} for a precise formulation).
\medskip 

When considering the $\cC$-\lIPS\ refutation in \Cref{eq:intro:NS-second}, we can see that its size as defined above does \emph{not} depend on the size of the axioms $f_i$ and $\baxioms$. Therefore, it is possible to think of lower bounds on $\cC$-\lIPS\ refutations for axioms $f_i$ that are large (e.g., super-polynomial in the number of variables) or outside the circuit class $\cC$: although the axioms do not have small representation (or are even non-explicit) their $\cC$-\lIPS\ refutations may be small, and we want to rule out this possibility.
We call this regime of lower bounds a \emph{placeholder IPS lower bound}. Here, ``placeholder'' stands for the fact that in the refutation \Cref{eq:intro:NS-second} we can replace the axioms $f_i$ and \baxioms\ by variables (which take the role of axiom placeholders), measure the size of the resulting refutation, and then replace the axioms in their place again to get \Cref{eq:intro:NS-second}.
Forbes \emph{et al.}~\cite{FSTW21} considered placeholder IPS
lower bounds using the approach of hard multiples, and studied
its relation to the randomness versus hardness paradigm. 

While it is reasonable to assume that we can establish 
placeholder IPS lower bounds using some (possibly non-explicit) 
polynomial equations that require large circuit size, 
it is interesting to consider \emph{explicit} such hard
instances in the IPS placeholder regime. Such explicit hard instances
in the form of determinant identities were recently established by
 by Andrews and Forbes \cite{AF22} in the placeholder IPS regime.


\para{Different IPS Lower Bound Methods.}
The first to obtain proof  complexity lower bounds by reductions to algebraic circuit complexity lower bounds were
\cite{FSTW21}.
They introduced two methods: the \emph{functional lower bound} method and the \emph{lower bounds for multiples} method. 

\textbf{Functional lower bounds}: At the moment, the functional lower bound method yields stronger lower bounds than the latter one. It produced several concrete proof complexity lower bounds for variants of subset-sum instances against fragments of IPS. These fragments include IPS refutations written as read once (oblivious) algebraic branching programs (roABPs), depth-3 powering formulas (introduced as ``diagonal depth-3 circuits'' in Saxena~\cite{Saxena08}), and multilinear formulas.  
Alekseev, Hirsch, Grigoriev, and Tzameret \cite{AGHT20} used a method very similar to the functional lower bound method to establish a conditional lower bound on (general) IPS (leading to \cite{Ale21}). And \cite{GHT22} used this method together with Limaye, Srinivasan, and Tavenas constant-depth algebraic circuit lower bounds \cite{LST21} to establish multilinear constant depth IPS lower bounds. 

\textbf{Lower bound for multiples}: On the other hand, the lower bound for multiples method was used in \cite{FSTW21,AF22} to 
establish lower bounds against the weaker model of \emph{placeholder} $\cC$-IPS proofs, in which the hard instances do not necessarily have small circuits themselves in $\cC$.
\smallskip

It is worth mentioning three other (fairly strong, IPS-style) algebraic proof lower bounds approaches in the literature, as follows. 
\smallskip

\textbf{Meta-complexity}: Another method for IPS lower bounds is the \emph{meta-complexity} technique introduced by Santhanam and Tzameret \cite{ST21} who proved an IPS size lower bound on a self-referential statement. However, this lower bound is currently only conditional.

\textbf{Noncommutative approach}: Li, Tzameret and Wang \cite{LTW18} (following \cite{Tza11-I&C}) considered IPS refutations over noncommutative formulas. Using this, they showed that size lower bounds against Frege proofs can be reduced to the task of proving that the rank of certain families of matrices is high. This however, does not yet yield any concrete lower bound since it is unclear at the moment how to characterise  precisely families of matrices that correspond to noncommutative IPS proofs (namely, given an unsatisfiable  CNF formula we would like to characterise and identify certain properties of matrices that correspond to noncommutative IPS refutations of the CNF formula; and using those properties establish rank lower bounds).

{\textbf{PC with extension variables}:} Finally, Impagliazzo, Mauli and Pitassi \cite{IMP23} (following Sokolov \cite{Sok20}; see improvements in \cite{DMM23}) proved a polynomial calculus with extension variable lower bounds for a CNF formula, over a finite field. Although this is a lower bound against the number of monomials appearing in refutations, due to the (restricted amount of) extension variables, this can be considered as a proof system between depth-2 IPS (i.e., Nullstellensatz) to depth-3 IPS---denoted informally as ``depth 2.5-IPS''. This lower bound is over finite fields. Apparently this method cannot go beyond this ``depth 2.5-IPS'' lower bounds.
\medskip 

It is important to notice also that all unconditional IPS  size lower bounds (apart from those that are lower bounds against number of monomials only, and hence stated in the setting of Nullstellensatz or Polynomial Calculus, including the result of \cite{IMP23}), are all for hard instances that are \emph{non-Boolean} (e.g., not CNF formulas; namely, polynomials that do not take on 0-1 values over the Boolean cube). We discuss this matter in \Cref{sec:intro:barriers}.

\subsection{The Functional Lower Bound Method}

We start by recalling the functional lower bound method
which is a reduction from algebraic circuit lower bounds to proof complexity lower bounds introduced  in \cite{FSTW21}, which holds some resemblance to the well-known feasible interpolation lower bound technique in proof complexity 
\cite{BPR97,Kra97-Interpolation}.

We denote by \baxioms\ the set of Boolean axioms $\{x_i^2-x_i :\; i\in[|\vx|]\}$.\medskip

\fbox{
\begin{minipage}[c]{0.95\textwidth}
\begin{theorem}[Functional Lower Bound Method;~Lemma 5.2  in~\cite{FSTW21}]
\label{def:single-functional-lower-bound-method}
Let $\cC\subseteq\pring$ be a circuit class closed under 
(partial) field-element  assignments
(which stands for the class of ``polynomials with 
small circuits"). Let $f(\vx)\in \cC$ be a polynomial, where the collection of polynomials $f(\vx)$ 
and $\baxioms$ is unsatisfiable (i.e., does not have a
common root). A \demph{functional lower bound against
$\cC$-{\rm \lbIPS} for $f(\vx)$ and \baxioms} is a lower
bound  argument using the following circuit lower bound
 for $\frac{1}{f(\vx)}$: Suppose that $g\not\in\cC$
for all $g\in\pring$ with 
\begin{equation}\label{eq:intro:functional-lower-bound-mthd}
g(\vx)=\frac{1}{f(\vx)}, \quad \forall\vx\in\bits^n\,.
\end{equation}
Then, $f(\vx)$ and $\baxioms$
do not have $\mathcal{C}$-\lbIPS\ refutations. 
Moreover, if $\cC$ is a set of multilinear polynomials,
then, $f(\vx)$ and $\baxioms$ do not have  $\mathcal{C}$-\IPS\
 refutations.
\end{theorem}
\end{minipage}
}
\medskip 

The idea behind this theorem is  simple. Let 
$$
g(\vx)\cd f(\vx) + \sum_i h_i(\vx)\cd (x_i^2-x_i) = 1 
$$
be a $\cC$-\lbIPS\ refutation of $f(\vx)$ 
and $\baxioms$, for some $g, h_i$'s. Then, since the Boolean axioms nullify over the Boolean cube, \Cref{eq:intro:functional-lower-bound-mthd} holds for $g$ (note that $1/f$ is defined over the Boolean cube because $f$ is unsatisfiable, meaning it does not have a 0-1 root). Hence, since $g\not\in\cC$
for all $g$ for which  \Cref{eq:intro:functional-lower-bound-mthd} holds, we get the  $\mathcal{C}$-\lbIPS\ lower bound. 

The method is called ``functional" lower bound, since the algebraic circuit lower bounds are functional, namely apply to the \emph{family of all polynomials} that compute the function $1/f$ over the Boolean cube (in contrast with the usual ``syntactic'' view of algebraic lower bounds which hold for a specific single polynomial defined as a vector of monomials).
We refer the reader to the work of Forbes, Kumar, and Saptharishi \cite{FKS16} which investigated functional lower bounds in algebraic circuit complexity.

\subsection{Our Results and Techniques}

\paragraph{Summary and Organisation}

\begin{itemize}
\item Past results on algebraic proofs established degree lower bounds for (fully) symmetric instances of  \emph{degree only one}, namely subset sum instances of the form $\sum_i x_i -\beta$ (cf.~\cite{IPS99,FSTW21}). By making the use of symmetry explicit in these lower bounds, in \Cref{sec:symmetric-deg-lower-bounds}
we generalise these Nullstellensatz degree lower bounds to symmetric instances of any degree.

\item  We further show that full symmetry is not necessary to obtain degree lower bounds. In particular, in \Cref{sec:invariant-axioms-deg-lower-bounds} we 
show degree lower bounds against Nullstellensatz refutations for (``partially symmetric'') vector invariant polynomials. Since these instances are different from subset sum variants, they are unsatisfiable over both positive and zero characteristics (unlike subset sum, which is unsatisfiable only over large or zero characteristics); and moreover, since we analyse precisely the $n/2$-degree slice of Nullstellensatz refutations of these instances, we are able to demonstrate degree lower bounds for these instances over every field, including over finite fields.  

\item We show how to lift our new degree lower bounds to \emph{size} lower bounds against different fragments of IPS in \Cref{Size Bounds for General Symmetric Hard Instance} and \Cref{sec:Invariant-size-lower-bounds}.
We use a more involved lifting on subset sum degree bounds to establish new size lower bounds against $O(\log\log n)$ individual degree IPS refutations of constant-depth in \Cref{sec:indiv-degree-size-lower-bounds}

\item Finally, in \Cref{sec:barriers} we show that the Functional Lower Bound approach alone cannot lead to lower bounds against Frege proof systems such as $\ACZ[p]$-Frege and $\TCZ$-Frege. 
\end{itemize}

\para{Notation.} Let  $\N':=\N\cup\{0\}$. Given $n$ variables $\vx:=\{x_1,\ldots,x_n\} $ and a vector 
$\valpha\in \N'^n$, we denote by $\vx^{\valpha}$ the monomial $\prod_{i=1}^n x_i^{\alpha_i}$. Using this notation we have 
$\deg(\vx^{\valpha}) =\sum_{i\in[n]}{\alpha_i} $, denoting the \textit{total degree} of $\vx^{\valpha}$. Given a filed \F, a polynomial in $\F[\vx]$ is a linear combination of monomials. The \emph{degree} of a polynomial (also called \emph{the total degree}) is the maximal total degree of its (nonzero) monomials. The \emph{individual degree} of a variable in a polynomial is the maximal power of the 
variable across all monomials in the polynomial.
The individual degree of a polynomial is the maximal individual degree of a variable across all variables in the polynomial.

Given a polynomial $f(\vx)\in\F[\vx]$ we say that 
\demph{$f(\vx)$ is symmetric} if $f(\vx)$ is invariant 
(i.e., stays the same) under permutation of all the  variables:
\begin{equation*}
f(\vx) = f(x_{\sigma(1)},\dots,x_{\sigma(n)}),
\end{equation*}
for every $\sigma\in S_n$, where $S_n$ is the group of permutations over $n$ element. 
The \demph{elementary symmetric polynomial of degree 
$0\le d\le n $} over the $n$ 
variables \vx\ is defined as 
$ \ednx:=\sum_{\valpha\in{n\choose d}}{\vx^{\valpha}}\;,$
where for the sake of convenience we set 
$\el_{0,n}(\vx):=1$.  

%
%

\subsubsection{Nullstellensatz Degree Lower Bounds}\label{sec:intro:NS-degree}
We show two new Nullstellensatz degree lower bounds, which are later lifted to IPS size lower bounds.

\para{Symmetric Instances.}

First we establish hardness for  any symmetric polynomial. 
\begin{corollary*}[see Cor.~\ref{cor:NS-full-symmetric-deg-lower bound}; Single unsatisfiable symmetric polynomials require high degree refutations]
\label{cor:intro:NS-full-symmetric-deg-lower bound}
Assume that $n \ge 1$ and  $1 \le d \le n$, \F\ 
is a field of characteristic 
strictly greater than $\max(2^n,n^d)$, and 
$f(\vx)$ is a symmetric polynomial of degree  $d$ such that $f(\vx)-\beta$ has no 0-1 solution, for $\beta \in \mathbb{F}$. Suppose that $g$ is a multilinear polynomial such that 
\[
g(\xbar) \cdot (f(\vx) - \beta) = 1 \mod{\xbar^2 - \xbar}.
\]
Then, the degree of $g(\vx)$ is at least $n-d+1$.
Accordingly, the degree of every Nullstellensatz refutation
of $f(\vx) - \beta$ is at least $n+1$. 
\end{corollary*}

Note indeed that the statement of the corollary gives   a Nullstellensatz degree lower bound because a  Nullstellensatz refutation of $f(\vx)-\beta$ looks like $g(\vx)\cd(f(\vx)-\beta) + \sum_i h_i\cd(x_i^2-x_i) = 1 $.
Recall also that taking a polynomial modulo the Boolean axioms \baxioms, is the same as multilinearizing the polynomial.

The following claim is the main technical observation  behind the degree lower bound and is a generalisation of the \cite{FSTW21}
lower bounds for the case of linear symmetric polynomials,
i.e., when $d=1$. Intuitively, the idea is that the
multilinearization of the product $\left(\ednx-\beta\right) \cd \eknx $ contains nonzero monomials of degree $\ge 1$ given that $d+k\le n$, hence it cannot equal the polynomial $1$. 
(Note that when $d+k>n$ we cannot make sure that this product when we multiply out terms,
does not yield cancellations of monomials 
resulting potentially in the 1 polynomial.)

\begin{claim*}[see Claim~\ref{cla:action-of-Sn-x-e-d}; Multilinearizing the product of elementary symmetric polynomials yields high degree]
Let $n \ge 1$ and $1 \le d \le n$. If $k$ is such that $k\le n-d$, then
$$
    \ednx \cd\eknx = 
    2^{d+k}\cd\elpx{d+k}{n}+ [\text{\small degree $\le d+k-1$ terms}] \mod \vx^2-\vx
                \;.
$$
\end{claim*}
Using this claim and the fact that every symmetric polynomial 
can be written as a polynomial in the elementary symmetric polynomials, we conclude the Nullstellensatz lower bounds for symmetric polynomials.

\para{Vector Invariant Polynomials.}
We provide degree lower bounds for instances that are different from subset sum variant (this means formally that they are not a substitution instance of $\sum_{i=1}^n x_i$, for $n=\omega(1)$). Later we show how to lift those to size lower bounds. These instances are interesting because they are different from previous IPS hard instances which were all based on the subset sum, and they hold  over finite fields as well. 

 Our hard instances are inspired by  ideas from  \emph{invariant theory} and specifically vector invariants. 
Roughly speaking, an invariant polynomial in the variables \vx\ is a polynomial that stays the same when each variable $x_i$ is replaced by the $i$th element of the vector $A\vx$, for $A$ a matrix taken from a matrix group. We provide some general information from invariant theory in \Cref{sec:invariant-axioms-deg-lower-bounds}.
We consider a class of invariant polynomials known as \emph{vector invariants}, which is  a well-studied class of invariant polynomials \cite{Richman,CampbellHughes,DerksenKemper} 


Let $\xbar := \{x_1,x_2,\ldots, x_{2n}\}$ and $\ybar := \{y_1, y_2, \ldots, y_{2n}\}$ 
be commuting variables over a field $\mathbb{F}$ of characteristic greater than 3 (for size lower bounds we would need char$(\F)\ge 5$). 
Let 
$$\widetilde{Q}(\xbar, \ybar) := \left(\prod_{i \in [2n],~i: \text{ odd}} (x_iy_{i+1} - y_ix_{i+1})\right).
$$
The hard instance is defined as 
\begin{equation}\label{eq:intro:Qxy}
Q(\xbar, \ybar) := \widetilde{Q}(\xbar, \ybar) - \beta
\end{equation}
where $\beta \in \mathbb{F}$ and $\beta \notin \{-1,0,1\}$. 

We use several interesting properties of this polynomial to prove the lower bound (see \Cref{fac:prop-Q}). Specifically, the polynomial is invariant under the following action: for every odd $i\in[2n]$ 
(it is sufficient for the present work to think of
actions, denoted $\hookrightarrow$, as substitutions
of variables by polynomials)
\begin{equation}\label{eq:433:30}
x_{i}\hookrightarrow x_{i}~~~~~~~~~~~~ 
x_{i + 1} \hookrightarrow x_{i+1}~~~~~~~~~~
y_{i}\hookrightarrow x_i + y_i~~~~~~~~~  
y_{i+1}\hookrightarrow x_{i+1} + y_{i+1}\,.
\end{equation}

 For $i \in [2n]$ and $i$ odd, let $a_i := (x_i y_{i+1} - y_ix_{i+1})$. Then, $a_i$ is the determinant of the matrix $M_i = \left(\begin{array}{cc}
             x_i & x_{i+1} \\
             y_i & y_{i+1}
        \end{array}\right).$ Moreover, $a_i \in \{-1,0,1\}$ over the Boolean cube. Note that $\widetilde{Q}(\vx,\vy)-\beta =0$ is unsatisfiable when $\beta \notin \{-1, 0, 1\}$.

\begin{theorem*}[see Thm.~\ref{thm:ns-degreeLB-Q}; Nullstellensatz degree lower bounds for invariant instances]
Let $\mathbb{F}$ be any field of characteristic at least  $3$ and  let $\beta \notin \{-1, 0, 1\}$. Then, $Q(\xbar, \ybar)=0$, $\{x_i^2-x_i=0\}_i$, and $\{y_i^2-y_i=0\}_i$ are unsatisfiable and  any polynomial $f(\xbar, \ybar)$ such that $f(\xbar, \ybar) = 1/Q(\xbar, \ybar)$, for $\xbar \in \{0,1\}^{2n}$ and $\ybar \in \{0,1\}^{2n}$, has degree  at least $2n$.  
\end{theorem*}

Note indeed that the statement of the theorem gives a Nullstellensatz degree lower bound because a  Nullstellensatz refutation of $\widetilde{Q}(\vx,\vy)-\beta$ looks like $f(\vx,\vz)\cd(\widetilde{Q}(\vx,\vz)-\beta) + \sum_i h_i\cd(x_i^2-x_i)
+ \sum_ih_i'\cd(y_i^2-y_i) = 1 $, meaning that $f(\vx,\vz) = 1/Q(\vx,\vz)$ over the Boolean cube. 

We show in \Cref{sec:roabp-sizeLB-Q} and \Cref{sec:invrnt-any-ordr} that the invariant theoretic properties specified in \Cref{eq:433:30}  facilitate a complete understanding of the coefficient space pertaining to  the homogeneous degree $2n$-slice of any refutation of $Q(\vx,\vy)$.  


\subsubsection{Lifting Degree to Size Lower Bounds}

The idea of ``lifting'' usually refers to the notion of taking a hard instance against a specific computational (or proof) model and altering the instance to make it hard even for a \emph{stronger} computational (or proof) model 
(see the recent survey by de Rezende, 
G{\"{o}}{\"{o}}s and Robere~\cite{RGR22} for a discussion  
on the use of lifting in proof complexity and references therein).
Standard lifting usually proceeds by taking a substitution instance of the original hard instance. In this case the  substitution is defined according to a ``gadget'', which is a specific function or polynomial, applied separately on each of the variables. For example, $f(z_1,\dots,z_n)\mapsto f(x_1y_1,\dots, x_ny_n)$, in which the gadget is defined as $z_i\mapsto x_i y_i$, for all $i\in[n]$. 

The idea to use lifting  to turn a hard instance against Nullstellensatz \emph{degree} to a hard instance against $\cC$-IPS \emph{size} was introduced in \cite{FSTW21}. The gadgets used in \cite{FSTW21} were quite simple. Here we show that our new hard instances against Nullstellensatz degree can be lifted with similar gadgets to our IPS fragments of interest. In \cite{GHT22} the gadget  is more involved, and we show how to carry it over to the new setting of Amireddy \textit{et al}.~\cite{AGK0T23} to gain lower bounds against a stronger fragment of constant depth IPS refutations. 

\para{Using Lifting against IPS.} Let us consider how lifting is used to yield size lower bounds. Recall that in the functional lower bound method, we reduced the task of lower bounding the size of a $\cC$-IPS refutation of $f(\vx)=0$ into the task of lower bounding the size of an algebraic circuit from the class $\cC$ that computes the function $g(\vx)=\nicefrac{1}{f(\vx)}$ over the Boolean cube. 
To establish an algebraic circuit lower bound we usually need to lower bound the rank of a  certain matrix denoted $\coeffs{\vu|\vv}(g)$ corresponding to the coefficient matrix of the polynomial $g$ under a partition of the variables $\vx=(\vu,\vv)$. In such a coefficient matrix, the $(M,N)$ entry is the coefficient in $g$ of the monomial $M\cd N$, with $M$ a monomial in the $\vu$-variables and $N$ a monomial in the $\vv$-variables.  

Note however that in our case $g$ is not a polynomial, rather a \emph{family} of many polynomials all of which compute the \emph{function} $\nicefrac{1}{f(\vx)}$
over the Boolean cube. 
To prove such a lower bound on a family of polynomials,  \cite{FSTW21} used an alternative rank argument: the \emph{evaluation dimension} (as suggested by Saptharishi \cite{Saptharishi12}; cf.~\cite{FKS16}), which for us will be defined as the dimension of the following space of polynomials under partial assignments $\{g(\vu,\valpha):\; \valpha\in\bits^{|\vv|}\}$. 
For the most part, our arguments  also use evaluation dimension (except for the vector invariant polynomials in which our analysis is tighter). It is not hard to show (\Cref{res:evals_eq-coeffs}) that  evaluation dimension is a lower bound on  the rank of $\coeffs{\vu|\vv}(g)$. 
We are thus left with the task of lower bounding ~$\dim \{g(\vu,\valpha):\; \valpha\in\bits^{|\vv|}\}$. 
This is where we need lifting. Specifically,  we need to maintain two main properties:
\vspace{-5pt} 

\begin{enumerate}
\item
 We need to \emph{substitute} the original variables \vx\ of our polynomial $g$ (and accordingly $f$) with  gadgets, resulting in a new polynomial equipped with a natural partition of variables $\vu,\vv$ that would provide high dimension to $\{g(\vu,\valpha):\; \valpha\in\bits^{|\vv|}\}$.\vspace{-4pt}  

\item 
$g(\vu,\valpha)$, for $\valpha\in\bits^{|\vv|}$,  reduces to our  \emph{original} instance $g(\vx)$ (possibly at a smaller input length). 
\label{it:intro:lifting-idea}

\end{enumerate}

The gist behind the two properties is this: to lower bound $\dim \{g(\vu,\valpha):\; \valpha\in\bits^{|\vv|}\}$ we use 
\Cref{it:intro:lifting-idea}. This allows us to use our original Nullstellensatz refutation degree lower bound on each element of the set $\{g(\vu,\valpha):\; \valpha\in\bits^{|\vv|}\}$. Using such a degree lower bound on $g(\vu,\valpha)$ for distinct assignments  $\valpha\in\bits^{|\vv|}$ we can ``isolate'' enough distinct leading monomials (per some global fixed monomial ordering). By a known result, the number of distinct leading monomials in a space of polynomials is a lower bound on its dimension. Hence, we conclude the evaluation dimension lower bound (and the circuit lower bound).


\para{Symmetric Instances under Lifting.}


We show lower bounds against $\cC$-\lbIPS\ of any symmetric polynomial under lifting, where $\cC$ is the class of read once (oblivious) algebraic branching programs (roABPs). Accordingly, we denote this system by \roAlbIPS\ (see \Cref{def:roABP} for the definition of roABP).
The lifting is defined by replacing the elementary symmetric polynomials with elementary symmetric polynomials with a bigger number of variables and then applying the gadget to each variable.

More precisely, given a symmetric polynomial $f(\vq)$ with $n$ variables $q_1,\dots, q_n$, by the fundamental theorem of symmetric polynomials (\Cref{prop:fund-thm-sym-polynomials}), it can be written as a polynomial in the elementary symmetric polynomials: 
$f(\vq):= h(y_1/\el_{1,n}(\vq),\dots,y_n/\el_{n,n}(\vq))$ for some polynomial $h(\vy)$ (where $h(y/g(\vq))$ denotes the polynomial $h$ in which the variable $y$ is substituted by the polynomial $g(\vq)$).
Consider the polynomial
$
f'(\vw):=h(y_1/\el_{1,m}(\vw),\dots,y_n/\el_{n,m}(\vw))
$
for $m= {2n\choose 2}$ and $\vw = \{w_{i,j}\}_{i<j\in[2n]}$.
We now apply a similar gadget to \cite{FSTW21},
defined by the mapping 
$$
w_{i,j}\mapsto z_{i,j}x_ix_j\,,
$$
which substitutes the $m$ variable $w_{i,j}$ by $m+2n$ variables $\{z_{i,j}\}_{i<j\in[2n]}$, $x_1,\dots,x_{2n}$:
\begin{equation}
\label{eq:intro:1588:15}
f^\star(\vz,\vx):=h(y_1/(\el_{1,m}(\vw))_{w_{i,j}\mapsto z_{i,j}x_ix_j},\dots,y_n/(\el_{n,m}(\vw))_{w_{i,j}\mapsto z_{i,j}x_ix_j})\,,
\end{equation}
where $(\el_{j,m}(\vw))_{w_{i,j}\mapsto z_{i,j}x_ix_j}$ means that we apply the lifting $w_{i,j}\mapsto z_{i,j}x_ix_j$ to the \vw\ variables. 

Note that when we use lifting to change the variables in
an instance, we also \emph{add the Boolean axioms for each of the
new variables}.

\begin{corollary*}[Symmetric instances are hard for roABP-\lbIPS; see~Cor.~\ref{res:lbs-fn:lbs-ips:vary-order}]
Let $n\ge 1$, $m={n \choose 2}$, and $\F$ be a field with $\chara(\F)>\max(2^{4n+2m},n^d)$.  
Let $f\in\F[\vq]$ be a symmetric polynomial with $n$ variables
of degree $d=O(\log n)$, and $f^\star(\vz,\vx)$  be as in \Cref{eq:intro:1588:15}. 
Let $\beta\in\F$ be such that $f^\star(\vz,\vx)-\beta=0$ and $f(\vq)-\beta=0$ are each unsatisfiable over Boolean values. 
Then, any \roAlbIPS\ refutation (in any variable order) of
 $f^\star(\vz,\vx)-\beta=0$ (together with the Boolean axioms to the \vx- and \vz-variables) requires $2^{\Omega(n)}$-size. 
\end{corollary*}

The new idea we employ in this result is showing how to maintain \Cref{it:intro:lifting-idea} of the lifting scheme above, even in the \emph{absence} \emph{of} \emph{maximal degree lower bounds} for the original hard polynomial $g(\vx)$ in that scheme. While in \cite[see remark after Proposition 5.8]{FSTW21} a maximal $n$ lower bound on $g(\vx)$ was thought to be necessary  to the dimension lower bound, we show that we can lower 
bound $\dim\{g(\vu,\valpha):\; \valpha\in\bits^{|\vv|}\}$
with only an $\Omega(n)$ Nullstellensatz degree lower bound for $g(\vx)$ ($n$ is the number of variables in \vx). Recall that general symmetric polynomials have only $\Omega(n)$ Nullstellensatz degree lower bounds according to \Cref{cor:NS-full-symmetric-deg-lower bound}
(for low enough symmetric polynomials). (This, on the other hand, will mean that we do not get size lower bounds for formula multilinear IPS refutations (as in \cite{FSTW21}), since the results of Raz-Yehudayoff (\Cref{thm:full-rank-lb}) require full rank $2^n$ lower bounds, while the dimension lower bound we get is only $2^{\Omega(n)}$). 

%
%

\para{Invariant Instances under Lifting.}

We show how to lift the vector invariant polynomials hard against Nullstellensatz degree $Q(\xbar, \ybar)$ from \Cref{sec:intro:NS-degree} (\Cref{eq:intro:Qxy}), to hard instances against IPS refutation-size, where refutations are written as roABPs and multilinear formulas, respectively.

Let $\ubar = \{u_1, u_2, \ldots, u_{4n}\}$, let $m = {{4n}\choose{4}}$, and $\zbar = \{z_1, z_2, \ldots, z_m\}$. The hard instance  $P(\ubar, \zbar) \in \mathbb{F}[\ubar, \zbar]$ is defined by:\[P(\ubar, \zbar) := \left(\prod_{i<j<k<\ell \in [4n] } 1 - z_{i,j,k,\ell} + z_{i,j,k,\ell} (u_iu_\ell - u_ju_k)\right) - \beta\,.
\]

We show the following \roAlbIPS\ and multilinear-formula-\lbIPS\ lower bounds. Here, multilinear-formula-\lbIPS\  denotes $\cC$-IPS where $\cC$ is the class of multilinear formulas, and we note that multilinear-formula-\lbIPS\   is equivalent to full multilinear-formulas-IPS (\cite{FSTW21}).

\begin{theorem*}[see Thm.~\ref{thm:any-order}]
Let $\mathbb{F}$ be a field of characteristic $\geq 5$ and let $P(\ubar, \zbar)$ be as defined above. Then $P(\ubar, \zbar), \{u_i^2-u_i\}_i, \{z_i^2-z_i\}_i$ is unsatisfiable as long as $\beta \notin \{-1, 0, 1\}$. 
%
%
And any \roAlbIPS\  refutation of $P(\ubar, \zbar), \{u_i^2-u_i\}_i, \{z_i^2-z_i\}_i$ requires $\exp(\Omega(n))$ size. Moreover, any multilinear-formula-IPS refutation requires $n^{\Omega(\log n)}$ size and any product-depth-$\Delta$ multilinear-formula-IPS requires size $n^{\Omega((n/\log n)^{1/\Delta}/\Delta^2)}$. 
\end{theorem*}

This is the first IPS fragment lower bound over finite fields, solving an open problem in \cite{PT16,GHT22}. \nutan{Also mention functional lower bounds for multilinear circuits here?} Previous IPS fragment lower bounds were all using variants of the subset sum $\sum_{i=1}^n x_i-\beta =0$ as hard instances. They crucially used the fact that the field has large characteristic so that the instance first is indeed unsatisfiable (note that over characteristic at most $n$ it is satisfiable over Boolean values), and second the degree lower bound can be carried through (see the proof of \Cref{lem:sym-poly-ref-degree-high}, paragraph before \Cref{cla:action-of-Sn-x-e-d} for an explanation).


Note that these lower bounds are in the placeholder IPS regime because the hard instances themselves do not have small roABPs and multilinear formulas.
\iddolater{Do we have these lower bound for the instances themselves or do we simply don't know?}


\paragraph{Proof Overview.} The starting point for this lower bound proof is the Nullstellensatz degree lower bound for $Q(\vx,\vy)$. The degree lower bound establishes that $f(\vx,\vy)$, namely the polynomial that agrees with $1/Q(\vx,\vy)$ over the Boolean cube, has a degree at least $2n$ (as stated in \Cref{thm:ns-degreeLB-Q}; see \Cref{sec:intro:NS-degree}). Here, we further refine this statement. 
We completely characterise the homogeneous degree-$2n$ slice of the polynomial $f(\vx,\vy)$. Specifically, we show that any monomial of degree $2n$ in $f(\vx,\vy)$ has coefficient either $\frac{1}{\beta(1-\beta)}$ or $\frac{1}{\beta(1+\beta)}$. This allows us to give a lower bound on the coefficient space of the polynomial $f(\vx, \vy)$ under any order in which $\xbar < \ybar$.  Here, we use the invariant properties of the polynomial $Q(\xbar, \ybar)$ crucially. We observe that if $Q(\xbar, \ybar)$ is invariant under a certain action, then so is the refutation (i.e., the polynomial $f(\xbar, \ybar)$).

To obtain a lower bound in any order (not just when $\xbar < \ybar$), we build further on the above. For this, we use a lifted version of $Q(\xbar, \ybar)$ (this is a different lifting than for the symmetric instances size lower bounds), namely the polynomial $P(\ubar, \zbar)$ mentioned above. Let $g(\ubar, \zbar)$ be the polynomial that agrees with $1/P(\ubar, \zbar)$ over the Boolean cube. We interpret $g(\ubar, \zbar)$ over $\mathbb{F}[\zbar][\ubar]$. And observe that for every partition of the variables $\ubar$ into equal parts, say $\vbar, \wbar$,  there exists a  $0$-$1$ assignment to the $\zbar$ variables, such that it recovers an instance of $f(\vbar, \wbar)$. 
This corresponds to \Cref{it:intro:lifting-idea} in the conditions
used in the lifting scheme above.
Accordingly, using this condition allows us to prove lower 
bounds on the coefficient dimensions for any partition of variables. \medskip

Interestingly, the fact that our lower bound works for \emph{all orders} of $\ubar$ variables, allows us to prove a functional lower bound for multilinear \emph{formulas}. 
Formally, we get the following corollary as a byproduct of the above theorem.
\begin{corollary}[New functional \emph{formula} lower bound]
    \label{cor:functionLB-multi}
    Let $P(\ubar, \zbar)$ be as defined above. Let $g(\ubar, \zbar)$ be a polynomial that agrees with $1/P(\ubar, \zbar)$ over the Boolean cube. Then, any multilinear formula that agrees with $g(\ubar, \zbar)$ over the Boolean cube must have size $\exp(\Omega(n))$. 
\end{corollary}

Previously, functional lower bounds were known for the roABP model due to~\cite{FKS16}. In their case, the hard polynomial was in \textsf{VNP}.  Above, our hard polynomial is $g(\ubar, \zbar)$. We suspect that the same functional lower bound as stated in \Cref{cor:functionLB-multi} can also be proved for $P(\ubar, \zbar)$. If this is true, then we  get a polynomial computable by a product-depth-$2$ circuits for which we have an exponential functional lower bound. This is likely to be of independent interest. 






\para{Extended Constant-Depth Lower Bounds.}

Using Limaye, Srinivasan, and Tavenas  \cite{LST21} constant-depth circuit lower bounds, Govindasamy, Hakoniemi, and Tzameret \cite{GHT22} established constant-depth IPS lower bounds against a lifted subset sum. The lifting in \cite{GHT22} was rather involved to fit the ``lopsided'' rank measure introduced by \cite{LST21}. In~\cite{GHT22}, the lower bound was not tight and it applied only to IPS refutations computable
by \emph{multilinear} polynomials. 


We show how to use recent progress on constant-depth lower bounds by Amireddy, Garg, Kayal, Saha, and Thankey~\cite{AGK0T23} to achieve tighter lower bounds for constant-depth IPS refutations, 
while also extending \cite{GHT22} lower bounds to constant-depth IPS refutations computing polynomials of $O(\log\log n)$-individual degrees.

\begin{theorem*}[Constant-depth IPS lower bounds; 
See Thm.~\ref{thm:constant-depth-lower-bounds}]
Let $n,\Delta$ and $\delta$ be positive integers, and assume 
that $\F$ is a field with $\mathrm{char}(\F) = 0$, and $\beta\in\F$.\footnote{The lower bound also 
holds for fields of large enough characteristic that depends
on the number of variables.}
Let $g$ be a polynomial of 
individual degree at most $\delta$ such that it agrees with 
    \[
    \frac{1}{\sum_{i,j,k,\ell\in[n]}z_{ijk\ell}x_ix_jx_kx_\ell - \beta}\quad\text{ over Boolean values.}
    \]
    Then, any circuit of product-depth at most $\Delta$ computing $g$ has size at least
    \[
    n^{\scalebox{1.23}{$\Omega$}\left(\frac{(\log n)^{2^{1 - 2\Delta}}}{\delta^2\cdot\Delta} \right)}.
    \]
\end{theorem*}

Note that by the functional lower bound scheme above, this result gives immediately a constant-depth individual degree-$\delta$ IPS size lower bounds for 
$\sum_{i,j,k,\ell\in[n]}z_{ijk\ell}x_ix_jx_kx_\ell - \beta$ (whenever
$\beta\in\F$ makes this polynomial unsatisfiable, namely, nonzero, over the Boolean cube).


Note that while \cite{GHT22} established a similar size lower 
bound for the \emph{unique multilinear} polynomial computing
$\nicefrac{1}{\sum_{i,j,k,\ell\in[n]}z_{ijk\ell}x_ix_jx_kx_\ell
 - \beta}$ over the Boolean cube, 
Theorem~\ref{thm:constant-depth-lower-bounds}
establish this lower bound for the \emph{family} of polynomials 
whose individual degree is $\delta$,
computing this polynomial over the Boolean cube.

This  solves the  question raised in \cite{AGK0T23}
about the ability to use their framework to obtain functional lower bounds for low-depth algebraic formulas in a direct manner without hardness-escalation via set-multilinear formulas.%
\footnote{\cite{AGK0T23} reads:~``\textit{Furthermore, it is conceivable that a direct argument can also be used to obtain functional lower bounds for low-depth formulas which might be useful in proof complexity.''}} Specifically, in~\cite{LST21}, the constant-depth circuit lower bound is first proved for a restricted setting of \emph{set-multilinear} constant-depth circuits. This is then \emph{escalated} to first obtain homogeneous constant-depth circuit lower bounds and in turn this is escalated further to obtain general constant-depth circuit lower bounds. On the other hand,~\cite{AGK0T23} show a lower bound for the homogeneous constant-depth circuits directly, thereby bypassing the need for proving set-multilinear lower bounds. A natural question arose from this development: can the techniques used in~\cite{GHT22} be extended using the ideas in~\cite{AGK0T23} to get stronger proof complexity lower bounds for constant-depth IPS proofs? We answer this question affirmatively.


We briefly discuss the proof of this lower bound and what is new in it. \cite{AGK0T23} showed how to circumvent the use of set-multilinear polynomials in \cite{LST21}. 
%
In doing so it also puts leaner requirements for the hard-instances, when the work of  \cite{LST21} required their hard instances to be set-multilinear.
For \cite{GHT22}, this meant that in order to use constant-depth algebraic circuit lower bounds they needed to show that any constant-depth IPS proof embeds in some sense a hard set-multilinear 
polynomial. 
%
To embed set-multilinear polynomials in any proof, \cite{GHT22} needed to stick to IPS refutations that compute multilinear polynomials only. Otherwise, if the IPS proofs themselves are not multilinear their set-multilinear projection (namely, the set-multilinear polynomials they ``embed'') could be zero, regardless of their functional behaviour over the Boolean cube (one can use multilinearization on the proofs, but there is no guarantee that this operation  increases the evaluation dimension, meaning that a lower bound on multilinear proofs is achieved that may not hold for non-multilinear proofs).

In the circuit lower bound,~\cite{AGK0T23} use a measure, introduced in \cite{GKS20}, called \textit{Affine Projections of Partials} (\textsf{APP})\nutan{Add ref.}. They analyse  this measure and prove that lower bounds for this measure are enough to prove a superpolynomial lower bound for constant-depth homogeneous circuits. We benefit from their analysis to be able to extend the lower bounds from~\cite{GHT22}. Specifically,  we 
are able to lower bound the $\APP$ measure for arbitrary IPS refutations (perhaps even highly non-multilinear) of our hard instance. From these lower bounds on the measure we are able to infer constant-depth lower bounds for refutations of bounded individual degree.

At the technical level, this gives two kinds of improvements. First, we are able to analyse more general class of proofs and prove lower bounds for them. And second, we are able to modify the framework of \cite{AGK0T23} in order to use it for functional lower bounds.

\subsubsection{Barriers for Boolean Instances Lower bounds}\label{sec:intro:barriers}

Lower bounds against $\ACZ[p]$-Frege proofs stand as one of the most elusive lower bound questions in proof complexity, open for more than three decades, while still considered within reach using current techniques (especially, due to the known $\ACZ[p]$ circuit lower bounds).
In light of the simulation of $\ACZ[p]$-Frege by constant-depth IPS refutations over $\F_p$ shown in \cite{GP18}, it is promising to think of using algebraic circuit lower bounds in the framework of IPS to solve this open problem. We show that at least when it comes to the functional lower bound method (as defined precisely in general in \Cref{def:general-functional-lower-bound-method}), this goal is impossible to achieve.

When we attempt to prove a lower bound against a propositional proof system operating with Boolean formulas (such as $\ACZ[p]$-Frege) by way of algebraic proofs lower bounds, we need to focus on hard instances against the algebraic proof systems that are nevertheless Boolean. 
We say that an instance consisting of a set of polynomials $\{f_i(\vx)=0\}_i$, for $f_i(\vx)\in\F[\vx]$, is \emph{Boolean} whenever $f_i(\vx)\in\bits$ for $\vx\in\bits^{|\vx|}$.
For example, a CNF written as a set of (polynomials representing) clauses is a Boolean instance. Similarly, the standard arithmetization of propositional formulas is Boolean instances. Note that up to this point, we discussed only \emph{non}-Boolean hard instances. For example, the subset sum $\sum_i x_i - \beta$ is highly non-Boolean due to its image under \bits-assignments being $\{-\beta,1-\beta,\dots,n-\beta\}$. The vector invariant polynomial instances are also non-Boolean because their image under Boolean assignments is $\{-1,0,1\}$. Similarly, previous hard instances against fragments of IPS are all non-Boolean.

\begin{theorem*}[Main barrier; see Thm.~\ref{thm:barrier}]
The functional lower bound method cannot establish
lower bounds for any \textit{\emph{Boolean}} instance against sufficiently strong proof 
systems. 
In particular, it cannot establish any lower bounds against \ACZ$[p]$-Frege, \TCZ-Frege (and constant-depth
\lbIPS\ when the hard instances are Boolean).
\end{theorem*}

 
Here, a \emph{sufficiently strong proof system}  (see \Cref{def:suf-strong-ps} for the precise definition) means a proof system that basically has the AND introduction rule,  in the sense that from $\phi_1,\dots,\phi_n$ one can efficiently derive $\bigwedge_i \phi_i$. Most reasonably strong proof systems such as \ACZ$[p]$-Frege and \TCZ-Frege clearly have this property.

To understand this result, first, we need to understand how to potentially use the functional lower bound method for Boolean instances. For the sake of simplicity, let us discuss the case where the proof system we try to prove lower bounds against is  $\cC$-$\lbIPS$ for some algebraic circuit class $\cC$ (in \Cref{sec:barriers} we discuss the general case of any proof system including strictly propositional ones like \ACZ$[p]$-Frege).

Let $\cF:=\{f_i(\vx)=0\}_i$ be a collection of \emph{Boolean} polynomial equations 
in $\cC$, in the above sense, and suppose 
we wish to establish a lower bound for
$\cF$ against  $\cC$-$\lbIPS$ using the functional lower bound method. For this purpose, we prove a lower bound for $f(\vx)=0$ against
$\cC$-\lbIPS, using the functional lower bound method. We then need to show that there is a short $\cC$-{\rm \lbIPS}-proof of $\cF$ from $f(\vx)=0$ (and $\baxioms$) (it is possible that $\cF$ is equal to $\{f(\vx)=0\}$).
Then, we can conclude there is  no 
$\cC$-{\rm \lbIPS}-refutations of $\cF$ 
(otherwise, starting from $f(\vx)=0$, we can efficiently derive $\cF$ and refute $f(\vx)=0$ in contradiction to the assumption that $f(\vx)=0$ is hard for $\cC$-{\rm \lbIPS}).

The idea behind the barrier is as follows: if $\cC$-{\rm \lbIPS} is sufficiently strong we can efficiently derive the arithmetization of $\bigwedge_i f_i(\vx)$ in $\cC$-{\rm \lbIPS} (for example, it can   be written as $1-\prod_i(1-f_i(\vx)) $).  
The negation \(\prod_i(1-f_i(\vx))\) of this single polynomial is a tautology, namely it is always zero over the Boolean cube, and hence is in the ideal generated by the Boolean axioms $\{x_j^2-x_j\}_j$. Thus, $\cF$ can be refuted using only the Boolean axioms (though not necessarily efficiently). From this, it is not hard to show that there is no functional lower bound against the  function 
$g(\vx)=\frac{1}{f(\vx)}, \quad \forall\vx\in\bits^n$
(as in \Cref{eq:intro:functional-lower-bound-mthd} in the Functional Lower Bound Method \Cref{def:single-functional-lower-bound-method}). 

It is interesting to note that recently Grochow~\cite{Gro23} showed that even a low-depth IPS fragment constitutes a sufficiently strong proof system in our sense.

We also note that the barrier is not sensitive to a specific arithmetization scheme, as long as it  translates Boolean formulas to a polynomial that computes the same Boolean function over the Boolean cube (where possibly 1 is flipped with 0; see \Cref{sec:barriers}).







\section{Preliminaries}

\subsection{Notation}
\label{sec:notation}

For a natural number $n\in\N$, $[n]:=\{1,\dots,n\}$. 
We assume that $0\in\N$.
We call $\bits^n$ the \emph{Boolean cube}. 

Let \G\ be a ring (we usually work with sufficiently large fields denoted \F\ or fields of zero characteristic, and this is specified when important). Denote by $\G[\vx]$ the ring of (commutative) polynomials with coefficients from $ \G $ and variables $\vx:=\{x_1,x_2,\,\dots\,\}$. A \emph{polynomial} is a formal linear combination of monomials, whereas a \emph{monomial} is a product of variables. Two polynomials are \emph{identical} if all their monomials have the same coefficients. The (total) degree of a monomial is the sum of all the powers of variables in it. 
The (total) \emph{degree} of a polynomial is the maximal total degree of a monomial in it. The degree of an \emph{individual} variable in a monomial is its power. The \emph{individual degree} of a monomial is the maximal individual degree of its variables. The 
individual degree of a polynomial $f$, denoted $\ideg f$, 
is the maximal individual degree of its monomials. For a polynomial $f$ in $\G[\vx,\vy]$ with $\vx,\vy$ being pairwise disjoint sets of variables, the \emph{individual $\vy$-degree} of $f$ is the maximal individual degree of a $\vy$-variable only in $f$. 

Given $n$ variables \vx\ and a vector 
$\valpha\in \N^n$ we denote by $\vx^{\valpha}$ the monomial
$\prod_{i=1}^n x_i^{\alpha_i}$. Using this notation we have 
$\deg(\vx^{\valpha}) =\sum_{i\in[n]}{\alpha_i} $, 
denoting the \textit{total degree} of $\vx^{\valpha}$.
Denote by $\coeff{\vx^\va\vy^\vb}(f)$ the coefficient 
of $\vx^\va\vy^\vb$ in $f$.

We say that a polynomial is \emph{homogeneous} whenever every monomial in it has the same (total) degree. For a polynomial $f(\xbar)$, the \emph{degree-$d$ homogeneous slice of $f(\xbar)$} (degree-$d$ slice, for short) is a polynomial defined by the degree $d$ monomials of $f(\xbar)$. We say that a polynomial is \emph{multilinear} whenever the individual degrees of each of its variables are at most 1. 

For two polynomials $g(\vy),h(\vx)$ We denote by $g(y_i/h(\vx))$ the \textit{substitution} in $g(\vy)$ of the variable $y_i \in \vy$ by the polynomial $h(\vx)$. 

\subsection{Algebraic Circuits}

Algebraic circuits and formulas over the ring \G\ compute polynomials in $\G[\vx]$ via addition and multiplication gates, starting from the input variables and constants from the ring. 
More precisely, an \emph{algebraic circuit} $C$ is a finite directed acyclic graph (DAG) with \textit{input nodes} (i.e., nodes of in-degree zero) and a single \textit{output
 node} (i.e.,  a node of out-degree zero).  Edges are labelled by ring \G\ elements.  Input nodes are labelled
 with variables or scalars from the underlying ring. In this work (since we work with constant-depth circuits)
 all other nodes have unbounded \emph{fan-in} (that is, unbounded in-degree) and are labelled by either an addition
 gate $+$ or a product gate $\times$.
Every node in an algebraic circuit $C$ \emph{computes}
 a polynomial in $\G[\vx]$ as follows: an input node computes %
%
the variable or scalar that labels it. A $+$ gate
computes the linear combination of all the polynomials computed by its incoming nodes, where the coefficients of the linear combination are determined by the corresponding incoming edge labels. A $\times$ gate computes the product of all the polynomials computed by its incoming nodes (so edge labels in this case are not needed). The polynomial
 computed by a node $u$ in an algebraic circuit $C$ is denoted $\widehat u$. Given a circuit $C$, we denote by
 $\widehat C$ the polynomial computed by $C$, that is, the polynomial computed by the output node of $C$.  The \emph{\textbf{size}} of a circuit $C$ is the number of
 nodes in it, denoted $|C|$, and the \emph{\textbf{depth}} of a circuit is the length of the longest directed path 
in it (from an input node to the output node). 
The \textbf{\emph{product-depth }}of the circuit is the maximal number of product gates in a directed path from
an input node to the output node.
For excellent treatises on algebraic circuits and their
 complexity see Shpilka and Yehudayoff \cite{SY10} as
  well as Saptharishi \cite{Sap17-survey}.


Let $\overline{X} = \langle X_1,\ldots,X_d\rangle$ be a sequence of pairwise disjoint sets of variables, called \emph{a variable-partition}. We call a monomial $m$ in the variables $\bigcup_{i\in [d]}X_i$  \emph{set-multilinear} over the variable-partition $\overline{X}$ if it contains exactly one variable from each of the sets $X_i$, i.e. if there are $x_i\in X_i$  for all $i\in [d]$ such that $m = \prod_{i\in [d]}x_i$. A polynomial $f$ is set-multilinear over $\overline{X}$ if it is a linear combination of set-multilinear monomials over $\overline{X}$. For a sequence $\overline{X}$  of sets of variables, we denote by $\F_{\sml}[\overline{X}]$ the space of all polynomials that are set-multilinear over $\overline{X}$.

We say that an algebraic circuit $C$ is set-multilinear over $\overline{X}$ if $C$ computes a polynomial that is set-multilinear over $\overline{X}$, and each internal node of $C$  computes a polynomial that is set-multilinear over some sub-sequence of $\overline{X}$.

\subsection{Symmetric Polynomials}
We denote by $S_n$ the \emph{permutation group} over $n$ elements. Concretely, an element $\sigma\in S_n$ can be identified with a (permutation)
function $\sigma:[n]\to[n]$. Given $n$
variables \vx\ we denote by $\permutx$ the application of 
$\sigma\in S_n$ to the variables; namely, their renaming 
according to $\sigma$. In this way, $f(\permutx)\in\F[\vx]$
is the result of renaming in $f\in\F[\vx]$ all the variables
\vx\ according to $\sigma$.

\begin{definition}[Symmetric polynomial]\label{def:sym-poly} 
Given a polynomial $f(\vx)\in\F[\vx]$ we say that 
\demph{$f(\vx)$ is symmetric} if $f(\vx)$ is invariant 
(i.e., stays the same) under permutation of all the  variables:
\begin{equation}
f(\vx) = f(\permutx),
\end{equation}
for every $\sigma\in S_n$.
\end{definition}
We will also consider polynomials that are not fully symmetric, in the sense that they stay the same, i.e. are invariant, under a specific subgroup $G$, which is not necessarily $S_n$. 
In this case we call this polynomial \demph{invariant under }$G$.

\begin{definition}[Elementary symmetric polynomial \ednx] 
The \demph{elementary symmetric polynomial of degree 
$0\le d\le n $} over the $n$ 
variables \vx\ is defined as follows:
\[
 \ednx:=\sum_{\valpha\in{n\choose d}}{\vx^{\valpha}}\;,
\]
where for the sake of convenience we set 
$\el_{0,n}(\vx):=1$. \end{definition}

Note that \ednx\ is multilinear, and that there is only a \emph{single} homogeneous multilinear 
symmetric polynomial over $n$ variables \vx, up to scalar multiplication:
\begin{fact}\label{fact:80:84}
If $f\in\F[\vx]$ is a symmetric multilinear and homogeneous 
polynomial of degree $d$, then $f(\vx) = \lambda\cd\ednx$, for 
some $\lambda\in\F$.
\end{fact}

\Cref{fact:80:84} is immediate, 
because otherwise there was a pair of distinct 
multilinear monomials of total-degree $d$,  
$\vx^\alpha\neq \vx^{\alpha'}$, with 
$\valpha\neq\valpha'\in\{0,\dots,n\}^d$ and with 
different respective coefficients $\lambda\neq\lambda'\in\F$. 
Then, there is a permutation $\sigma\in S_n$ of the 
$n$ variables \vx\ such that $f(\permutx)$ contains $\lambda\vx^{\valpha'}$ instead of 
$\lambda'\vx^{\valpha'}$, in contradiction to the symmetry of $f$.
   
Recall that for two polynomials $g(\vy),h(\vx)$ We denote by $g(y_i/h(\vx))$ the \textbf{substitution} in $g(\vy)$ of the variable $y_i \in \vy$ by the polynomial $h(\vx)$. 

\begin{proposition}[The fundamental  theorem of symmetric polynomials]
\label{prop:fund-thm-sym-polynomials}
Every symmetric polynomial $f\in\F[\vx]$ with $n$ variables can be written as a \emph{polynomial} in the elementary symmetric polynomials. 
That is, there is a $g(\vy)\in\F[\vy]$ such that 
\begin{equation}
f(\vx) = g(y_1/\el_{1,n}(\vx),\dots,y_n/\el_{n,n}(\vx)).
\end{equation}
%

Moreover, if $f(\vx)$ is multilinear then $g(\vy)$ is linear, that is, $f(\vx)$
can be written as a \emph{linear} combination of elementary
symmetric polynomials:
\begin{equation}
f(\vx) = \sum_{i=0}^n \lambda_i \el_{i,n}(\vx), \text{~~~~with $\lambda_i\in\F$}.
\end{equation}
\end{proposition}

\begin{proof}
For a proof of the first part of \Cref{prop:fund-thm-sym-polynomials} see \cite[Chap.~7, Theorem 3]{CoxLittleOShea15}. For a proof of the second part, proceed by induction on the degree $d$ of $f$ as follows: write $f(\vx)=A + B$ where $B$ is the sum of all monomials of total degree $<d$ in $f$. 
Then, $A$ is the (homogeneous polynomial) which consists of the sum of all (multilinear) monomials of degree precisely $d$
in $f$. Note that $A$ must be symmetric, since under any 
permutation of variables $\sigma$, monomials in $A(\permutx)$
remain of degree $d$ and monomials in $B(\permutx)$ remain of 
degree at most $d-1$. 
Thus, $A$ is a symmetric, homogeneous and multilinear polynomial of degree $d$, which by \Cref{fact:80:84} means that $A=\lambda_d\cd \edn(\vx)$ for some $\lambda_d\in\F$
(and where $n=|\vx|$).
\end{proof}

\subsection{Algebraic Proof Systems}

For a survey about algebraic proof systems and their relations to algebraic complexity see the survey \cite{PT16}.
Grochow and Pitassi~\cite{GP18}  suggested the following algebraic proof system which is essentially a Nullstellensatz proof system \cite{BeameIKPP96} written as an algebraic circuit.

\begin{definition}[Nullstellensatz refutations] \label{def:NS}
Let $f_1(\vx),\ldots,f_m(\vx),p(\vx)$ be a collection of polynomials in $\F[x_1,\ldots,x_n]$ over the field \F. A \demph{Nullstellensatz refutation  of the axioms $\{f_j(\vx)=0\}_{j\in [m]}$}, showing that the set of axioms do not have a solution from the Boolean cube is a sequence of polynomials $\{g_i(\vx)\}_{j\in[m]}$,
such that (the equality in what follows stands for a formal polynomial identity): 
\[
    \sum_{i\in[m]}g_i(\vx)\cd f(\vx) + \sum_{i\in[n]}h_i(\vx)\cd(x_i^2-x_i) = 1\,.
\]
The \textbf{degree} of the refutation is $\max\{\deg(g_i(\vx)\cd f_i(\vx)):\; i\in[m]\}$.\footnote{It can be shown that $\max\{\deg(g_i(\vx)\cd f_i(\vx)):\; i\in[m]\} \ge \max\{\deg(h_i(\vx))+2:\; i\in[n]\}$, hence there is no need to count the degrees of the $h_i$'s in the size.} 
\end{definition}

Notice that the definition above adds the
 equations $\{x_i^2-x_i=0\}_{i=1}^n$, called the  
\demph{Boolean axioms} denoted  $\vx^2-\vx$, to the system $\{f_j(\vx)=0\}_{j=1}^m$. This allows  to refute systems of equations that have no solution over $\bits^n$ (though they may be solvable over $\F$ in general).

A proof in the  Ideal Proof System is given as  a \emph{single} polynomial. We provide below the \emph{Boolean} version of  IPS (which includes the Boolean axioms), namely the version that establishes the unsatisfiability over 0-1 of a set of polynomial equations.  In what follows we follow the notation in \cite{FSTW21}:

\begin{definition}[Ideal Proof System (IPS),
Grochow-Pitassi~\cite{GP18}]\label{def:IPS} Let $f_1(\vx),\ldots,f_m(\vx),p(\vx)$ be a collection of polynomials in $\F[x_1,\ldots,x_n]$ over the field \F. An \demph{IPS proof of $p(\vx)=0$ from axioms $\{f_j(\vx)=0\}_{j\in [m]}$}, showing that $p(\vx)=0$ is semantically  implied from the assumptions $\{f_j(\vx)=0\}_{j\in [m]}$ over $0$-$1$ assignments, is an algebraic circuit $C(\vx,\vy,\vz)\in\F[\vx,y_1,\ldots,y_m,z_1,\ldots,z_n]$ such that (the equalities in what follows stand for  formal polynomial identities\footnote{That is, $C(\vx,\vnz,\vnz)$ computes the zero polynomial and $C(\vx,f_1(\vx),\ldots,f_m(\vx),x_1^2-x_1,\ldots,x_n^2-x_n)$ computes the polynomial $p(\vx)$.}; recall the notation $\widehat C$ for the \emph{polynomial} computed by circuit $C$):
        \begin{enumerate}
                \item $\widehat C(\vx,\vnz,\vnz) = 0$;
                \item $\widehat C(\vx,f_1(\vx),\ldots,f_m(\vx),x_1^2-x_1,\ldots,x_n^2-x_n)=p(\vx)$.
        \end{enumerate}
        The \demph{size of the IPS proof} is the size of the circuit $C$. An \IPS\ proof  $C(\vx,\vy,\vz)$ of  $1=0$ from $\{f_j(\vx)=0\}_{j\in[m]}$ is called an \demph{IPS refutation} of $\{f_j(\vx)=0\}_{j\in[m]}$ (note that in this case  it must hold that  $\{f_j(\vx)=0\}_{j\in [m]}$ have no common solutions in $\bits^n$).
        If $\widehat C$ is of individual degree $\le 1$ in each $y_j$ and $z_i$, then this is a \demph{linear} IPS refutation (called \emph{Hilbert} IPS by Grochow-Pitassi~\cite{GP18}), which we will abbreviate as \lIPS. If $\widehat C$ is of individual degree $\le 1$ only in the $y_j$'s then we say this is an \lbIPS\ refutation (following \cite{FSTW21}). If  $\widehat C(\vx,\vy,\vnz)$ is of individual degree $\le k$ in each $x_j$ and $y_i$ variables, while $\widehat C(\vx,\vnz,\vz)$ is not necessarily bounded in its individual degree, then this is called an 
\demph{individual degree-$k$ \lbIPS\ refutation}. 
        
If $C$ is of depth at most $d$, then this is  called a \emph{depth-$d$ \IPS\ refutation}, and further called a \emph{depth-$d$ \lIPS\ refutation} if $\widehat C$ is linear in $\vy,\vz$, and a depth-$d$ \lbIPS\ refutation if $\widehat C$ is linear in $\vy$, and \emph{depth-$d$ multilinear \lbIPS} refutation if $\widehat C(\vx,\vy,\vnz)$ is linear in $\vx,\vy$. 
\end{definition}

The variables $\vy,\vz$ are  called the \emph{placeholder} \emph{variables} since they are used as placeholders for the axioms. Also,
note that the first equality in the definition of IPS means that the polynomial computed by $C$ is in the ideal
generated by $\overline y,\overline z$, which in turn, following the second equality, means that $C$ witnesses
the fact that $1$ is in the ideal generated
 by $f_1(\vx),\ldots,f_m(\vx),x_1^2-x_1,\ldots,x_n^2-x_n$ (the existence of this witness, for unsatisfiable set
of polynomials, stems from the Nullstellensatz \cite{BeameIKPP96}).


\subsubsection{Oblivious Algebraic Branching Programs}

Algebraic branching programs (ABPs) is a model whose strength lies between that of algebraic circuits and algebraic formulas. 
(We use notation from \cite{FSTW21}.)

\begin{definition}[Nisan~\cite{Nisan91}; ABP]\label{def:roABP}
        An \emph{algebraic branching program (ABP) with unrestricted weights} of \emph{depth} $D$ and \emph{width} $\le r$, on the variables $x_1,\ldots,x_n$, is a directed acyclic graph such that:
        \begin{itemize}
                \item The vertices are partitioned into $D+1$ layers $V_0,\ldots,V_D$, so that $V_0=\{s\}$ ($s$ is the source node), and $V_D=\{t\}$ ($t$ is the sink node). Further, each edge goes from $V_{i-1}$ to $V_{i}$ for some $0< i\le D$.
                \item $\max|V_i|\le r$.
                \item Each edge $e$ is weighted with a polynomial $f_e\in\F[\vx]$.
        \end{itemize}
The \demph{(individual) degree} $d$ of the ABP is the maximum (individual) degree of the edge polynomials $f_e$. The \demph{size} of the ABP is the product $n\cdot r\cdot d\cdot D$.
Each $s$-$t$ path is said to compute the polynomial which is the product of the labels of its edges, and the algebraic branching program itself computes the sum over all $s$-$t$ paths of such polynomials.

The following are restricted ABP variants:
        \begin{itemize}
                \item An algebraic branching program is said to be \demph{oblivious} if for every layer $\ell$, all the edge labels in that layer are \emph{univariate} polynomials in a \emph{single} variable $x_{i_\ell}$.
                \item An oblivious branching program is said to be a \demph{read-once} oblivious ABP (\textup{\textbf{roABP}}) if each $x_i$ appears in the edge label of exactly one layer, so that $D=n$. That is, each $x_i$ appears in the edge labels in  exactly one layer. The layers thus define a \demph{variable order}, which will be assumed to be  $x_1<\cdots<x_n$ unless otherwise specified.


\end{itemize}
\end{definition}

The class of roABPs is essentially equivalent to 
non-commutative ABPs (\cite{ForbesShpilka13b}), a model at least
as strong as non-commutative formulas.  The study of non-commutative ABPs was initiated by  Nisan~\cite{Nisan91}, who proved lower bounds for non-commutative ABPs (and thus also for roABPs, in any order).  In terms of proof complexity, Tzameret~\cite{Tza11-I&C} studied a proof system with lines given by roABPs, and  Li, Tzameret 
and Wang~\cite{LTW18} showed that IPS over non-commutative formulas is quasipolynomially  equivalent in power to the Frege proof system. Since non-commutative ABPs and roABPs are essentially equivalent,
this last result motivates proving lower bounds for roABP-IPS as a way of attacking lower bounds for the Frege proof system.


\subsection{Coefficient Dimension and roABPs}\label{sec:coeff-dim}
In this section, we define the \emph{coefficient dimension} measure and recall basic properties. Full proofs of these claims can be found for example in the thesis of Forbes~\cite{Forbes14}.
Again, we use notation from \cite{FSTW21}.

We first define the \emph{coefficient matrix} of a polynomial.
This matrix is formed from a polynomial $f\in\F[\vx,\vy]$ by arranging its coefficients into a matrix.  That is, the coefficient matrix has rows indexed by monomials $\vx^\va$ in $\vx$, columns indexed by monomials $\vy^\vb$ in $\vy$, and the $(\vx^{\va},\vy^{\vb})$-entry is the coefficient of $\vx^\va\vy^\vb$ in the polynomial $f$.  We now define this matrix, recalling that $\coeff{\vx^\va\vy^\vb}(f)$ is the coefficient of $\vx^\va\vy^\vb$ in $f$ (see \Cref{sec:notation}).

\begin{definition}[Coefficient matrix]
\label{def:coeff-matrix}
        Consider $f\in\F[\vx,\vy]$.  Define the \demph{coefficient
        matrix of $f$} as the scalar matrix
        \[
                (C_f)_{\va,\vb}
                \eqdef
                \coeff{\vx^\va\vy^\vb}(f)
                \;,
        \]
        where coefficients are taken in $\F[\vx,\vy]$, 
        for $\sum_{j=1}^{|\va|}{|a_j|}, \sum_{j=1}^{|\vb|}{|b_j|}\le
        \deg f$.
\end{definition}

We now give the related definition of \emph{coefficient dimension}, which looks at the dimension of the row- and column-spaces of the coefficient matrix. Recall that $\coeff{\vx|\vy^\vb}(f)$ extracts the coefficient of $\vy^\vb$ in $f$, where $f$ is treated as a polynomial in $\F[\vx][\vy]$.

\begin{definition}[Coefficient space]
\label{defn:coefficient-space}
        Let $\coeffs{\vx|\vy}:\F[\vx,\vy]\to\subsets{\F[\vx]}$ be         the \demph{space of $\F[\vx][\vy]$ coefficients}, defined
        by
        \[
                \coeffs{\vx|\vy}(f)
                \eqdef
                \left\{
                        \coeff{\vx|\vy^\vb}(f)
                \right\}_{\vb\in\N^n}
                \;,
        \]
        where coefficients of $f$ are taken in $\F[\vx][\vy]$.
        Similarly, define $\coeffs{\vy|\vx}:\F[\vx,\vy]\to\subsets{\F[\vy]}$         by taking coefficients in $\F[\vy][\vx]$.
\end{definition}

The following basic lemma shows that the rank of the coefficient matrix equals the coefficient dimension, which follows from simple linear algebra.

\begin{lemma}[Coefficient matrix rank equals dimension 
of polynomial space; Nisan~\cite{Nisan91}]
\label{res:y-dim_eq-x-dim}
        Consider $f\in\F[\vx,\vy]$.  Then, the rank of the coefficient matrix $C_f$ (\Cref{def:coeff-matrix}) obeys
        \[
                \rank C_f
                =
                \dim\coeffs{\vx|\vy}(f)=\dim\coeffs{\vy|\vx}(f)
                \;.
                \qedhere
        \]
\end{lemma}

Therefore, the ordering of the partition ($(\vx,\vy)$ versus $(\vy,\vx)$) does not influence  the resulting dimension. 

%

We now state a convenient normal form for roABPs (see for example Forbes~\cite[Corollary 4.4.2]{Forbes14}).

\begin{lemma}[Characterising roABP as a matrix product]
\label{res:roABP-normal-form}
A polynomial $f\in\F[x_1,\ldots,x_n]$ is computed by width-$r$ roABP
iff there exist $n$ matrices 
$A_i(x_{i})\in\F[x_{i}]^{r\times r}$, for $i\in[n]$, 
each of (individual) degree $\le \deg f$ such that
 $f=(\prod_{i=1}^n A_i(x_{i}))_{1,1}$. 
\end{lemma}

Using this normal form we can characterise  roABP-width as follows.

\begin{lemma}[roABP-width equals dimension of coefficient space]
\label{res:roABP-width_eq_dim-coeffs}
        Let $f\in\F[x_1,\ldots,x_n]$ be a polynomial. If $f$ is computed
        by a width-$r$ roABP then $r \ge \max_i\dim\coeffs{\vx_{\le
        i}|\vx_{>i}}(f)$.
        Conversely, $f$ is computable by a width-$\left(\max_i\dim\coeffs{\vx_{\le
        i}|\vx_{>i}}(f)\right)$ roABP.
\end{lemma}

We use the following  closure properties of roABPs, taken from \cite{FSTW21}.

\begin{fact}\label{fact:roABP:closure}
        If $f,g\in\F[\vx]$ are computable by width-$r$ and width-$s$ \emph{roABPs} respectively, then
        \begin{itemize}
                \item $f+g$ is computable by a width-$(r+s)$ roABP.
                \item $f\cdot g$ is computable by a width-$(rs)$ roABP.
        \end{itemize}

        \noindent Further, roABPs are also closed under the following operations.
        \begin{itemize}
                \item If $f(\vx,\vy)\in\F[\vx,\vy]$ is computable by a width-$r$ roABP in some variable order then the partial substitution $f(\vx,\vaa)$, for $\vaa\in\F^{|\vy|}$, is computable by a width-$r$ roABP in the induced order on $\vx$, where the degree of this roABP is bounded by the degree of the roABP for $f$.
                \item If $f(z_1,\ldots,z_n)$ is computable by a width-$r$ roABP in variable order $z_1<\cdots<z_n$, then $f(x_1y_1,\ldots,x_ny_n)$ is computable by a $\poly(r,\ideg f)$-width roABP in variable order $x_1<y_1<\cdots<x_n<y_n$.
        \end{itemize}
\end{fact}

\subsection{Evaluation Dimension}\label{sec:eval-dim}

While coefficient dimension measures the size of a polynomial $f(\vx,\vy)$ by taking all coefficients in $\vy$, \emph{evaluation dimension} is a somewhat relaxed complexity measure due to Saptharishi~\cite{Saptharishi12} that measures the size by taking all possible evaluations in $\vy$ over the field.  This measure will be important for our applications as one can restrict such evaluations to the Boolean cube and obtain circuit lower bounds against a \emph{family} of polynomials that  compute $f(\vx,\vy)$ as a \emph{function} on the Boolean cube. 

\begin{definition}[Evaluation dimension; Saptharishi~\cite{Saptharishi12}]\label{defn:evaluation-space}
        Let $S\subseteq \F$. Let $\evals{\vx|\vy,S}:\F[\vx,\vy]\to\subsets{\F[\vx]}$ be the \demph{space of $\F[\vx][\vy]$ evaluations over $S$}, defined by
        \[
                \evals{\vx|\vy,S}(f(\vx,\vy))
                \eqdef
                \left\{
                        f(\vx,\vbb)
                \right\}_{\vbb\in S^{|\vy|}}
                \;.
        \]
        Define $\evals{\vx|\vy}:\F[\vx,\vy]\to\subsets{\F[\vx]}$ to be $\evals{\vx|\vy,S}$ when $S=\F$.
Similarly, define $\evals{\vy|\vx,S}:\F[\vx,\vy]\to\subsets{\F[\vy]}$ by replacing $\vx$ with all possible evaluations $\vaa\in S^{|\vx|}$, and likewise define $\evals{\vy|\vx}:\F[\vx,\vy]\to\subsets{\F[\vy]}$.
\end{definition}

The equivalence between evaluation dimension and coefficient dimension was shown by Forbes-Shpilka~\cite{ForbesShpilka13b} by appealing to interpolation.

\begin{lemma}[Evaluation dimension lower bounds dimension of coefficient space; Forbes-Shpilka~\cite{ForbesShpilka13b}]\label{res:evals_eq-coeffs}
        Let $f\in\F[\vx,\vy]$.  For any $S\subseteq\F$ we have that $\evals{\vx|\vy,S}(f)\subseteq\spn \coeffs{\vx|\vy}(f)$ so that $\dim \evals{\vx|\vy,S}(f)\le \dim \coeffs{\vx|\vy}(f)$. In particular, if $|S|>\ideg f$ then $\dim\evals{\vx|\vy,S}(f)=\dim\coeffs{\vx|\vy}(f)$.
\end{lemma}

Note that evaluation dimension and coefficient dimension are equivalent when the field is large enough (and $|S|$ is bigger than the individual degree of the polynomial). 
However, when restricting our attention to inputs from the Boolean cube this equivalence no longer holds, but evaluation dimension will be still useful as it \emph{lower bounds} coefficient dimension.

\subsection{Multilinear Polynomials and Multilinear Formulas}

We now turn to multilinear polynomials and classes that respect multilinearity such as multilinear formulas. We first state some well-known facts about multilinear polynomials (taken from \cite{FSTW21}).

\begin{fact}\label{fact:multilinearization}
        For any two multilinear polynomials $f,g\in\F[x_1,\ldots,x_n]$, $f=g$ as polynomials iff they agree on the Boolean cube $\bits^n$.  That is, $f=g$ iff  $f|_{\bits^n}=g|_{\bits^n}$.

        Further, there is a \demph{multilinearization} map $\ml:\F[\vx]\to\F[\vx]$ such that for any $f,g\in\F[\vx]$,
        \begin{enumerate}
                \item $\ml(f)$ is multilinear.
                \item $f$ and $\ml(f)$ agree on the Boolean cube, that is, $f|_{\bits^n}=\ml(f)|_{\bits^n}$.
                \item $\deg \ml(f)\le \deg f$.
                \item $\ml(fg)=\ml(\ml(f)\ml(g))$, and if $f$ and $g$ are defined on disjoint sets of variables then $\ml(fg)= \ml(f)\ml(g)$.
                \item $\ml$ is linear, so that for any $\alpha,\beta\in \F$, $\ml(\alpha f+\beta g)=\alpha \ml(f)+\beta\ml(g)$.
                \item $\ml(x_1^{a_1}\cdots x_n^{a_n})=\prod_i x_i^{\min\{a_i,1\}}$.
                \item If $f$ is the sum of at most $s$ monomials ($s$-sparse) then so is $\ml(f)$.
        \end{enumerate}
        Also, if $\hat{f}$ is a function $\bits^n\to\F$ that only depends on the coordinates in $S\subseteq[n]$, then the unique multilinear polynomial $f$ agreeing with $\hat{f}$ on $\bits^n$ is a polynomial only in $\{x_i\}_{i\in S}$.

\end{fact}

\para{Multilinear Formulas.}
We shall consider the model of  multilinear formulas.

\begin{definition}[Multilinear formula]\label{def-ml-fmla}
        An algebraic formula is a \demph{multilinear formula} if the polynomial computed by \emph{each} gate of the formula is multilinear (as a formal polynomial, that is, as an element of $\,\mathbb{F}[x_1,\ldots,x_n]$). The \demph{product depth} is the maximum number of multiplication gates on any input-to-output path in the formula.
\end{definition}

Raz~\cite{Raz09,Raz06} gave lower bounds for multilinear formulas using the above notion of coefficient dimension, and 
Raz-Yehudayoff~\cite{ry08,RazYehudayoff09} gave simplifications and extensions to constant-depth multilinear formulas.

\begin{theorem}[Raz-Yehudayoff~\cite{Raz09,ry08,RazYehudayoff09}]\label{thm:full-rank-lb}
        Let $f\in\F[x_1,\ldots,x_{2n},\vz]$ be a multilinear polynomial in the set of variables $\vx$ and auxiliary variables $\vz$.  Let $f_\vz$ denote the polynomial $f$ in the ring $\F[\vz][\vx]$.  Suppose that for any partition $\vx=(\vu,\vv)$ with $|\vu|=|\vv|=n$ that
        \[
                \dim_{\F(\vz)} \coeffs{\vu|\vv} f_\vz \ge 2^n
                \;.
        \]
        Then $f$ requires $\ge n^{\Omega(\log n)}$-size to be computed as a multilinear formula, and for $d=o(\nicefrac{\log n}{\log\log n})$, $f$ requires $n^{\Omega((\nicefrac{n}{\log n})^{\nicefrac{1}{d}}/d^2)}$-size to be computed as a multilinear formula of product-depth-$d$.
\end{theorem}

\subsection{Monomial Orders}\label{sec:mon-ord}

We recall here the definition and properties of a
\emph{monomial order}, following Cox, Little and 
O'Shea~\cite{CoxLittleOShea15}.
We  abuse notation and
associate a monomial $\vx^\va$ with its exponent vector
$\va$, so that we can extend this order to the 
exponent vectors. Note that in this definition ``$1$''
is a monomial, and that scalar multiples of monomials
such as $2x$ are not considered monomials. We now define
a monomial order, which will be total order on monomials
with certain natural properties.

\begin{definition}
A \demph{monomial ordering} is a total order $\prec$ on
the monomials in $\F[\vx]$ such that
        \begin{itemize}
                \item For all $\va\in\N^n\setminus\{\vnz\}$, $1\prec \vx^{\va}$.
                \item For all $\va,\vb,\vc\in\N^n$, $\vx^\va\prec\vx^\vb$ implies $\vx^{\va+\vc}\prec\vx^{\vb+\vc}$.
        \end{itemize}

        For nonzero $f\in\F[\vx]$, the \demph{leading monomial of $f$ (with respect to a monomial order $\prec$)}, denoted $\LM(f)$, is the largest monomial in $\supp(f)\eqdef\{\vx^\va:\coeff{\vx^\va}(f)\ne 0\}$ with respect to the monomial order $\prec$. The \demph{trailing monomial of $f$}, denoted $\TM(f)$, is defined analogously to be the smallest monomial in $\supp(f)$. The zero polynomial has neither leading nor trailing monomial.

        For nonzero $f\in\F[\vx]$, the \demph{leading (resp.\ trailing) coefficient of $f$}, denoted $\LC(f)$ (resp.\ $\TC(f)$), is $\coeff{\vx^\va}(f)$ where $\vx^\va=\LM(f)$ (resp.\ $\vx^\va=\TM(f)$).
\end{definition}

In contrast to \cite{FSTW21}, we will also use the existence of monomial orderings that \emph{respect degree} in the sense that if $\deg(M)>\deg(N)$ for two monomials $M, N$, then $M\succ N$. 

The following is a simple lemma about  leading or trailing monomials (or coefficients) being  homomorphic with respect to multiplication.

\begin{lemma}[\cite{FSTW21}]
\label{res:hom_LM-TM_mult}
Let $f,g\in\F[\vx]$ be nonzero polynomials. Then the leading monomial and trailing monomials and coefficients are homomorphic with respect to multiplication, that is, $\LM(fg)=\LM(f)\LM(g)$ and $\TM(fg)=\TM(f)\TM(g)$, as well as $\LC(fg)=\LC(f)\LC(g)$ and $\TC(fg)=\TC(f)\TC(g)$.
\end{lemma}

We shall use the well-known fact that for any set of polynomials the dimension of their span in $\F[\vx]$ is equal to the number of \emph{distinct} leading or trailing monomials in their span.

\begin{lemma}\label{res:dim-eq-num-TM-spn}
        Let $S\subseteq\F[\vx]$ be a set of polynomials. Then $\dim \spn S=\nLM{\spn S}=\nTM{\spn S}$.  In particular, $\dim \spn S\ge \nLM{S},\nTM{S}$.
\end{lemma}


\section{Degree Lower Bounds} 


\subsection{Symmetric Instances}
\label{sec:symmetric-deg-lower-bounds}

In this section we show that all symmetric unsatisfiable instances are hard for Nullstellensatz degree, as well as some vector invariant polynomial instances. 



Let $\xbar$ denote the set $\{x_1,  \ldots, x_n\}$.

\begin{fact}\label{fac:unsatisfiable-symmetric-poly-is-beta}
Let $f(\vx)\in\F[\vx]$ be symmetric and unsatisfiable over 0-1 assignments (i.e., $f(\vx)=0$ has no 0-1 solutions). Then $f(\vx)$ is of the form $t(\vx)-\beta$ with $0\neq\beta\in\F$ and $t(\vx)$ a symmetric polynomial. 
\end{fact}

\newcommand{\vzero}{\ensuremath{\overline 0}}

\begin{proof}
If $f(\vx)$ is symmetric and does not contain a constant 
term, i.e., $f(\vzero)=0$, then $f(\vx)=0$ is 
satisfiable. Thus, $f(\vx)=t(\vx)-\beta$, where $\beta = 
f(\vzero)$. The polynomial $t(\vx)$ is symmetric because 
$f(\vx)$ is (since permuting the variables of $f(\vx)$ per \Cref{def:sym-poly} of symmetric polynomials does not change the constant term $f(\vzero)$).
\end{proof}

\begin{lemma}\label{lem:sym-poly-ref-degree-high}
Let $n \ge 1$ and  $1 \le d \le n$, and let \F\ be a field of characteristic strictly greater than $\max{(2^{n},n^d)}$.
%
Let $\beta \in \mathbb{F}\setminus \{0,1,\ldots,
n^d\}$ and $f$ be a \emph{multilinear} polynomial such that 
\begin{equation}\label{eq:198}
f(\xbar) \cdot (\ednx -\beta) = 1 \mod{\xbar^2 - \xbar}.
\end{equation}
Then,  $n-d < \deg(f) \le n$. 
\end{lemma}

\begin{proof}
Note that \(\ednx -\beta=0\) is unsatisfiable whenever 
$\beta \in \mathbb{F}\setminus \{0,1,\ldots, n^d\}$ and the characteristic of \F\ is greater than $n^d$. This is because $\ednx$ contains $n^d$ distinct monomials with the coefficient 1, which evaluates to either 0 or 1 under Boolean assignments. Moreover, the characteristic of \F\ needs to be greater than $2^n$ so that for any nonzero $\gamma_\ell\in\F$ we have $\gamma_\ell\cd 2^n\neq 0$ in \F\ (which is needed in \Cref{eq:378:31}; see the ensuing explanation there).  

\uline{$\le n$:} This is clear as $f$ is multilinear.

\uline{$> n-d$:} Begin by observing that $\beta \in \mathbb{F}\setminus \{0,1,\ldots, n^d\}$ implies that $\ednx-\beta$ is never zero on the Boolean cube $\bits^n$, so that the by \Cref{eq:198} for $\vx\in\bits^n$ the expression
        \begin{align*}
                f(\vx)= \frac{1}{\ednx -\beta}
                \;,
        \end{align*}
        is well defined.
        Now observe that this implies that \emph{$f$ is a symmetric polynomial}. To see this, let us define $g(\xbar)$ to be the symmetrizing polynomial for $f(\xbar)$, i.e., \[g(x_1, x_2, \ldots, x_n) = \frac{1}{n!}\cdot \sum_{\sigma \in S_n} f(\sigma(x_1), \sigma(x_2), \ldots, \sigma(x_n)).\]
        Then, we see $\frac{1}{n!}\cdot \sum_{\sigma \in S_n} f(\sigma(x_1), \sigma(x_2), \ldots, \sigma(x_n)) =  \frac{1}{n!} \cdot \sum_{\sigma \in S_n} \frac{1}{\ensuremath{\mathbf{e}}_{d,n}(\sigma(\xbar)) - \beta}$. As $\ednx$ is symmetric, we see that $g(\xbar) = f(\xbar)$, which means $f(\xbar)$ is symmetric.%
\footnote{See \cite[Proposition 5.3]{FSTW21} for a spelling out of this argument. Another way to show that $f(\vx)$ is symmetric is this: a multilinear polynomial is uniquely determined by the values it gets over the Boolean cube. Here we know that the function over the Boolean cube is symmetric (as a function; namely,  stays the same under permutation of variables. In yet another words, the function depends only on the number of 1s in its input). This means that the polynomial itself is symmetric (because it is multilinear).}


The following claim is the main technical observation of the lower bound and is a generalisation of the \cite{FSTW21}
lower bounds for the case of linear symmetric polynomials,
i.e., when $d=1$. The idea is that the
multilinearization of the product $\left(\ednx-\beta\right) \cd \eknx $ contains nonzero monomials of degree $\ge 1$, given that $d+k\le n$, hence it cannot equal the polynomial $1$. 
On the other hand, note that when $d+k>n$ we cannot make sure that this product, when we multiply out terms,
does not yield cancellations of monomials 
potentially resulting  in the 1 polynomials. 


\begin{claim}[Multilinearizing the product of 
elementary symmetric polynomials yields high degree]
\label{cla:action-of-Sn-x-e-d}

Let $n \ge 1$ and $1 \le d \le n$. If $k$ is such that $k\le n-d$, then
$$
    \ednx \cd\eknx = 
    2^{d+k}\cd\elpx{d+k}{n}+ [\text{\small degree $\le d+k-1$ terms}] \mod \vx^2-\vx
                \;.
$$
\end{claim}

\begin{proofclaim}
\begin{align}
    \ednx \cd\eknx 
        &=\sum_{\va\in {n \choose d}} \vx^{\va} \cd 
          \sum_{\vb\in {n \choose k}} \vx^{\vb} 
         =\sum_{\va\in {n \choose d},~{\vb\in {n \choose k}}}
          \vx^{\va}\cd \vx^{\vb} 
          \notag\\
        &=\sum_{\va\in {n \choose d},~{\vb\in {n \choose k}}
            \atop \va\cap \vb = \emptyset}
          \vx^{\va}\cd \vx^{\vb} ~+
          \sum_{\va\in {n \choose d},~{\vb\in {n \choose k}}
          \atop \va\cap \vb \neq \emptyset}
           \vx^{\va}\cd \vx^{\vb} \notag\\
        &=2^{d+k}\cd\sum_{\vc\in {n \choose d+k}}
          \vx^{\vc}~ +
          \sum_{j=1}^{\min(d,k)}\sum_{\va\in {n \choose d},
          ~{\vb\in {n \choose k}} \atop |\va\cap \vb| =j}
           \vx^{\va}\cd \vx^{\vb}\,,\label{eq:323:50}
        \end{align}
where $2^{d+k}$ in the last line follows  by considering all possible partitions of each $(d+k)$-subset $\vc\in {n \choose d+k}$ into two disjoint subsets $\va\in {n
 \choose d},~{\vb\in {n \choose k}}$.
Now, consider the second summand in \Cref{eq:323:50}. If we multilinearize this polynomial we get a multilinear polynomial of degree at most $d+k-1$, concluding the claim. To be more precise, we have
\begin{multline}\label{eq:330:27}
    \sum_{j=1}^{\min(d,k)}\sum_{\va\in {n \choose d},
     ~{\vb\in {n \choose k}} \atop |\va\cap \vb| =j}
     \vx^{\va}\cd \vx^{\vb} 
    ={d+k-1\choose 1} \cd2^{d+k-2}\cd
        \sum_{\vc\in {n \choose d+k-1}}\vx^{\vc}
        ~ + {d+k-2\choose 2} \cd 2^{d+k-3}\cd
        \sum_{\vc\in {n \choose d+k-2}}\vx^{\vc} 
    \\
        ~ + {\max(d,k)\choose \min(d,k)}  
        \cd \sum_{\vc\in {n \choose \max(d,k)}}\vx^{\vc} 
        ~~~\mod \vx^2-\vx\,,
\end{multline}
where ${d+k-1\choose 1} \cd 2^{d+k-2}$ in the first summand
 in the right-hand side of the equation  means that for each
 choice of the single element in $\vc = \va\cup\vb$ that is
  common to both  $\va\in {n \choose d}$ and $ {\vb\in {n \choose
   k}} $ (there are ${d+k-1\choose 1} $ such choices)
    we can
     attribute the rest of the elements in $\vc$ to \va\ or
      \vb\ in $2^{d+k-2}$ different ways. Similar considerations
       apply to the rest of the summands in \Cref{eq:330:27}.
\end{proofclaim}       

Suppose for contradiction that $\ell=\deg(f)\le n-d$. Then, using
the fact that modulo $\vx^2-\vx$ the symmetric polynomial $f(\vx)$ must be symmetric and multilinear and hence a linear
combination of elementary symmetric polynomials (\Cref{prop:fund-thm-sym-polynomials}), we have 
        \begin{align}
        \notag
                1
                &=f(\vx)\cd\left(\ednx -\beta\right) \mod \vx^2-\vx\\
                \notag
                &=\left(\sum_{k=0}^{\ell} \gamma_k \eknx\right)\cd
                \left(\ednx -\beta\right) \mod \vx^2-\vx\\
                \notag
                &=\left(\sum_{k=0}^{\ell} \gamma_k \eknx \cd \ednx
                 \right)-\left(\sum_{k=0}^{\ell} 
                 \beta \gamma_k \eknx \right)\mod \vx^2-\vx\\
                 \notag
                &=\left(\sum_{k=0}^{\ell} \gamma_k 
                2^{d+k}\cd\elpx{d+k}{n}+ [\text{\small degree                 $\le d+k-1$ terms}]
                \right)-\left(\sum_{k=0}^{\ell} 
                \beta \gamma_k \eknx \right)\mod \vx^2-\vx\\
                \label{eq:378:31}
                &=\gamma_{\ell} 2^{n}\cd\elpx{d+\ell}{n}+ 
                [\text{\small degree $\le d+\ell-1$ terms}]
                \mod \vx^2-\vx\,.
        \end{align}
Where the penultimate equation is by invoking \Cref{cla:action-of-Sn-x-e-d},
which we can since $\ell\le n-d$. 
By assumption that $\deg(f)=\ell$, we know that 
$\gamma_\ell\ne 0$. By assumption  $\ell\le n-d$, we have that 
$\elpx{d+\ell}{n}$ is of degree at most $n$ (if $\ell=n-d$ it is $\elpx{n}{n}$) and at least $d$ (in case $\ell=0$), where $d\ge 1$ by assumption. Since the characteristic of the field is greater 
than $2^{n}$ we know that $\gamma_{\ell} 2^{n}\neq 0$.
This shows that $1$ (a multilinear degree 0 polynomial) equals because  \Cref{eq:378:31} (a multilinear degree $d+\ell$ polynomial) modulo $\vx^2-\vx$, which is a contradiction to the uniqueness of representation of multilinear polynomials modulo $\vx^2-\vx$.  Thus, we must have $\deg(f)>n-d$.
\end{proof}  

We now generalise \Cref{lem:sym-poly-ref-degree-high} from a lower bound for an \emph{elementary} symmetric polynomial to 
a lower bound for all symmetric polynomials.

\begin{corollary}[Single unsatisfiable  symmetric polynomials require high degree refutations]
\label{cor:NS-full-symmetric-deg-lower bound}
Assume that $n \ge 1$ and  $1 \le d \le n$, \F\ 
is a field of characteristic 
 greater than $\max(2^n,n^d)$, and 
$f(\vx)$ is a symmetric polynomial of degree  $d$ such that $f(\vx)-\beta$ has no 0-1 solution, for $\beta \in \mathbb{F}$. Suppose that $g$ is a multilinear polynomial such that 
\[
g(\xbar) \cdot (f(\vx) - \beta) = 1 \mod{\xbar^2 - \xbar}.
\]
Then, the degree of $g(\vx)$ is at least $n-d+1$.
Accordingly, the degree of every Nullstellensatz refutation
of $f(\vx) - \beta$ is at least $n+1$.  
\end{corollary}

\begin{proof}
By \Cref{prop:fund-thm-sym-polynomials} we can 
write the symmetric polynomial $f(\vx)$ as 
$\sum_{i=0}^d \lambda_i  e_d(\vx)$, with $\lambda_i\in\F$. 
Thus,
\begin{align}
    \notag
    1 &=  g(\xbar) \cdot (f(\vx) - \beta) \mod{\xbar^2 - \xbar}
    \\
    \notag
    &= g(\xbar) \cdot \left(
    \sum_{i=0}^d \lambda_i  e_i(\vx) - \beta\right)
    \mod{\xbar^2 - \xbar}
    \\
    &=g(\xbar) \cdot 
    \left(\lambda_d  e_d(\vx) + \sum_{i=0}^{d-1}\lambda_i  
    e_i(\vx)-\beta\right) \mod{\xbar^2 - \xbar} \notag
    \\
    &=
    g(\xbar) \cdot \lambda_d  e_d(\vx) +
    g(\xbar) \cdot \left(\sum_{i=0}^{d-1}\lambda_i  e_i(\vx)
    -\beta \right)
    \mod{\xbar^2 - \xbar}. 
    \label{eq:419:41} 
\end{align}

Assume for contradiction that $\deg(g)\le n-d$. 
By \Cref{cla:action-of-Sn-x-e-d}
$\deg\left(\ml\left
(g(\xbar) \cdot \lambda_d  e_d(\vx)
\right)\right)= \deg(g)+d$ 
(using also that $\lambda_d\neq0$, by assumption on degree of $f$), while
 $\deg\left(\ml\left(g(\xbar) 
 \cdot( \sum_{i=0}^{d-1}\lambda_i  e_i(\vx)
    -\beta)\right)\right)\le 
\deg(g)+d+1$. 
Therefore, \Cref{eq:419:41} is a (multilinear) polynomial of
degree at most $\deg(g)+d$. We assumed that $\deg(g)\le n-d$,
meaning that \Cref{eq:419:41} is a (multilinear) polynomial of
degree at most $n$ (and at least $1$ since $d\ge 1$).
But as before, this is a contradiction to the uniqueness of the
multilinear polynomial $1$. 
\end{proof}

%
%
%

\paragraph{Comment.}
%
Some of the symmetric lower bounds in this section section can 
possibly be obtained by using the subset sum lower bounds in \cite{FSTW21}. One way to do that would be to \emph{lift} the hardness of subset sum to derive the hardness for $\ednx - \beta$ (for $d \geq 2$). We believe that~\cite{FKS16} present preliminary ideas which should allow for such an approach to lower bounds. Specifically, they note that any $\ednx$ is equal to a product of $d$ affine forms over the Boolean cube. While this is a useful direction, this approach apparently can work to derive lower bounds only for some specific $\beta$'s. Whereas, in \Cref{cor:NS-full-symmetric-deg-lower bound} we can obtain such lower bounds for all $\beta$'s that makes the instance unsatisfiable over Boolean assignments. 
On the other hand, assuming this approach could be made to work, a possible benefit is this: consider the case of two-axioms $\{f_1 = \edn - \beta_1, f_2 = \sum_{i=1}^n x_i - \beta_2\}$. If we can show that $\edn - \beta_1$ is in the ideal generated by $f_2$ above, then we will be able to obtain lower bounds for the interesting case of two axioms $\{f_1, f_2\}$.





\subsection{Vector Invariant Polynomials}
\label{sec:invariant-axioms-deg-lower-bounds}


Here we show hardness for an instance that is not a  subset sum variant (formally, it is not a substitution instance of $\sum_{i=1}^n x_i$, for $n=\omega(1)$). Our hard instance is inspired by  ideas from \emph{invariant theory}. 

\paragraph{Polynomial Invariants.} We start with a gentle introduction to polynomial invariants. For a detailed introduction please refer to~\cite{CoxLittleOShea15}. Let $\GLnF$ denote the set of all $n\times n$ matrices over the field $\mathbb{F}$. Let $A \in \GLnF$, then we can think of $A$ acting on the polynomials in the polynomial ring $\mathbb{F}[x_1, \ldots, x_n]$ as follows. Let $A = \left(a_{i,j}\right) \in \GLnF$ and $f(x_1, \ldots, x_n) \in \mathbb{F}[x_1, \ldots, x_n]$, then 
\[g(x_1, \ldots, x_n) = f(a_{1,1}x_1+ \ldots + a_{1,n}x_n, \ldots, a_{n,1}x_1 + \ldots + a_{n,n}x_n).\]
More compactly, let $\xbar = (x_1 ~x_2 \ldots ~ x_n)^{\text{T}}$, then $g(\xbar) = f(A \cdot \xbar)$. 

We say that a polynomial $f(\xbar)$ is \emph{invariant} under the action of a finite matrix group $G \subset \GLnF$, if $f(\xbar) = f(A\cdot \xbar)$ for every $A \in G$. A set of all polynomials that are invariant under $G$ is denoted by $\mathbb{F}[\xbar]^G$. 

As an example, consider $\mathcal{S}_n$, the set of all $n \times n$ permutation matrices. Then, $\mathbb{F}[\xbar]^{\mathcal{S}_n}$ is the set of all symmetric polynomials. This is arguably the most well-studied class of invariant polynomials. In the previous section, we studied exactly this class of polynomials. 

Here, we will consider a different class of invariant polynomials known as \emph{vector invariants}, which is  a well-studied class of invariant polynomials. 
Intuitively, vector invariants are a class of polynomials that are invariant under the action of \emph{a decomposable} group, that is, a group that can be written down as a direct sum (taken say $n$ times) of a smaller group. The action on the bigger group is then defined by \emph{diagonally extending} the action on the smaller group. See~\cite{Richman,CampbellHughes,DerksenKemper} for formal definitions of vector invariants and many interesting results about them. 

Here, we will define the specific vector invariants relevant for our hard instances.

\paragraph{Vector Invariant Polynomials and the Hard Instance.} 

Let $\ubar = (u_1 ~u_2)^T$ and $\vbar = (v_1 ~v_2)^T$. Let $R$ be a linear transformation that maps $u_1$ to $u_1$, $u_2$ to $u_2$, $v_1$ to $u_1+v_1$ and $v_2$ to $u_2+v_2$. That is, 
\[R = \left(\begin{array}{cccc}
     1 & 0 & 0 & 0\\
     0 & 1 & 0 & 0\\
     1 & 0 & 1 & 0\\
     0 & 1 & 0 & 1
\end{array}\right),\]
which has the property that 
\[(u_1, u_2, u_1+v_1, u_2+v_2)^T = R \times (u_1, u_2, v_1, v_2)^T.\]

This specific action is based on the vector invariants of $U_2(\mathbb{F})$ studied in an influential paper of~\cite{Richman}. (The definition of $U_2(\mathbb{F})$ is very similar to $R$ above and can be found in~\cite{Richman,CampbellHughes} or in~\cite{DerksenKemper}.) 

Let $p(\ubar,\vbar) \in \mathbb{F}[\ubar, \vbar]$ be equal to $u_1v_2 - v_1u_2$, that is, it is the determinant of the following matrix.
\[
\left(\begin{array}{cc}
    u_1 & u_2 \\
    v_1 & v_2
\end{array}\right)
\]

Then, it is easy to see that $p(\ubar, \vbar) = p(R \cdot \ubar, \vbar)$. That is, if we apply the linear transformation given by $R$ to the variables of $p(\ubar, \vbar)$, then the polynomial stays invariant. This polynomial on $4$ variables is our main ingredient in the hard instance. We now describe our hard instance.


Let $\mathbb{F}$  be a field of characteristic greater than or equal to $5$ for the rest of this section\footnote{The degree lower bound in this section will work for characteristics $3$ or more, but for size lower bounds, we will need characteristics $5$ or more. For the sake of uniformity, we assume that the characteristic is $5$ or more throughout.}. Let $\xbar := \{x_1,x_2,\ldots, x_{2n}\}$ and $\ybar := \{y_1, y_2, \ldots, y_{2n}\}$ 
be commuting variables over $\mathbb{F}$.  Let 
$$\widetilde{Q}(\xbar, \ybar) := \left(\prod_{i \in [2n],~i: \text{ odd}} (x_iy_{i+1} - y_ix_{i+1})\right).
$$

Informally, $\widetilde{Q}(\xbar, \ybar)$ is obtained by taking $n$ copies of the polynomial $p(\ubar, \vbar)$ (with variable renaming) and extending the action of $R$ on the polynomial thus obtained. 

Finally, the hard instance is defined as follows.

\begin{equation}\label{eq:Qxy}
Q(\xbar, \ybar) := \widetilde{Q}(\xbar, \ybar) - \beta,
\end{equation}
where $\beta \in \mathbb{F}$. 
This polynomial has several interesting properties, summarized as follows:
\begin{fact}
\label{fac:prop-Q}
    The polynomial $Q(\xbar, \ybar)$ defined above has the following properties.
    \begin{enumerate}
        \item $Q(\xbar, \ybar)$ is a multilinear polynomial of degree $2n$. 
        \item The polynomial is invariant under the following action: for every odd $i\in[2n]$ 
(it is sufficient for the present work to think of
actions, denoted $\hookrightarrow$, as substitutions
of variables by polynomials)
$$
x_{i}\hookrightarrow x_{i}~~~~~~~~~~~~ 
x_{i + 1} \hookrightarrow x_{i+1}~~~~~~~~~~
y_{i}\hookrightarrow x_i + y_i~~~~~~~~~  
y_{i}\hookrightarrow x_{i+1} + y_{i+1}\,.
$$ 

        \item \label{fac:matrix-item} For $i \in [2n]$ and $i$ odd, let $a_i := (x_i y_{i+1} - y_ix_{i+1})$. Then, $a_i$ is the determinant of the matrix $M_i = \left(\begin{array}{cc}
             x_i & x_{i+1} \\
             y_i & y_{i+1}
        \end{array}\right).$ Moreover, $a_i \in \{-1,0,1\}$ over the Boolean cube. 
        \item $Q(\xbar, \ybar)$ is not satisfiable as long as $\beta$ is greater than or equal to $2$. 
    \end{enumerate}
\end{fact}

\subsubsection{Degree lower bound for $Q(\xbar, \ybar)$}
\label{sec:degreeLB-Q}

\paragraph{Notation.} For any set $A \subseteq [2n]$, let $\xbar_A$ denote the monomial $\prod_{i \in A} x_i$. Let $\widetilde{\xbar_{{A}}}$ denote the product $\prod_{i \in A} (1-x_i)$. Similarly, let $\ybar_A = \prod_{i \in A} y_i$ and $\widetilde{\ybar_{{A}}}$ denote the product $\prod_{i \in A} (1-y_i)$. 
Let $\mathbbm{1}_A$ denote a $2n$-length vector in which the $i$th bit is $1$ 
if and only if $i \in A$. 
For a monomial $m$ and a polynomial $f(\xbar, \ybar)$, let $\coeff{m}(f(\xbar, \ybar))$ denote the coefficient of the monomial $m$ in $f(\xbar, \ybar)$. 

\paragraph{The degree lower bound.} We define $f(\xbar, \ybar)$ to be the unique multilinear polynomial  that it is equal to $1/Q(\xbar, \ybar)$ over $\{0,1\}^{2n}$, i.e., over the Boolean cube of dimension $2n$. We will show that it contains a degree-$2n$ monomial with a non-zero coefficient. This will show that the degree of $f(\xbar, \ybar)$ is at least $2n$. 


Let $S = \{i \in [2n] \mid i \text { is odd}\}$ and $T = \{i \in [2n] \mid i \text { is even}\}$. 


\begin{lemma}
\label{lem:xsyt-C}
    There exists a $\beta > 2$ such that the coefficient of the monomial  $\xbar_S\cdot \ybar_T$ in $f(\xbar, \ybar)$ is non-zero. 
\end{lemma}
\begin{proof}
    From the uniqueness of the evaluations of multilinear polynomials over the Boolean cube, we know that 
    $$f(\xbar, \ybar) = \sum_{A \subseteq [2n], B \subseteq [2n]} f(\mathbbm{1}_A, \mathbbm{1}_B ) \xbar_A \cdot \widetilde{\xbar_{\overline{A}}} \cdot \ybar_B \cdot \widetilde{\ybar_{\overline{B}}},$$
    where $\mathbbm{1}_A$ is the indicator vector of the set $A$.

    First, note that to analyse the coefficient of $\xbar_S\cdot \ybar_T$ in $f(\xbar, \ybar)$, we can set $x_i =0$ for $i \notin S$ and similarly, $y_i=0$ for $i \notin T$. This is because, setting variables outside the set $S,T$ to zero does not change the coefficient of $\xbar_S\cdot \ybar_T$ in $f(\xbar, \ybar)$. 
Now, notice that if $A \nsubseteq S$, then $\xbar_A$ will become zero under the above assignment. Similarly, if $B \nsubseteq T$ then $\ybar_B$ will be set to zero. Thus, it suffices to sum over $A \subseteq S$ and $B \subseteq T$, if we want to understand the coefficient of $\xbar_S \cdot \ybar_T$. 
    Overall, we get 
    %
    $$\coeff{\xbar_S \cdot \ybar_T}\left(f(\xbar, \ybar)\right) = \coeff{\xbar_S\cdot \ybar_T}\left(\sum_{A \subseteq S,~B \subseteq T} f(\mathbbm{1}_A, \mathbbm{1}_B ) \xbar_A \cdot \widetilde{\xbar_{S\setminus A}} \cdot \ybar_B \cdot \widetilde{\ybar_{T \setminus B}}\right)\,. $$

    Observe that, $f(\mathbbm{1}_S, \mathbbm{1}_T) = \frac{1}{1-\beta}$, because $Q(\mathbbm{1}_S, \mathbbm{1}_T)=1-\beta$, and for any $A \subsetneq S$ or $B \subsetneq T$, $f(\mathbbm{1}_A, \mathbbm{1}_B) = -\frac{1}{\beta}$, because $Q(\mathbbm{1}_A, \mathbbm{1}_B)=-\beta$ 
    (since at least one of the $u_i$s is zeroed out by the assignment; see \Cref{fac:matrix-item} above). 

    Hence, we can now simplify the above summation as follows.
    \begin{align*}
        \coeff{\xbar_S \cdot \ybar_T}\left(f(\xbar, \ybar)\right) & = \coeff{\xbar_S\cdot \ybar_T}\left(\sum_{A \subseteq S, B \subseteq T} f(\mathbbm{1}_A, \mathbbm{1}_B ) \xbar_A \cdot \widetilde{\xbar_{S\setminus A}} \cdot \ybar_B \cdot \widetilde{\ybar_{T \setminus B}}\right)\\
        & = -\frac{1}{\beta} \cdot C + \frac{1}{1-\beta},
    \end{align*}
    where $C \in \mathbb{Z}$. Note that for as long as $\frac{\beta}{1-\beta}$ is not integral, the above number is non-zero for any integral value of $C$. We can ensure this by appropriately picking a value for $\beta$. For example, $\beta=3$ will work here.  This finishes the proof of the lemma.  
\end{proof}

In fact, we can compute the coefficient of $\xbar_X\cdot \ybar_T$ exactly. This understanding about the exact value of the coefficient will crucial for lifting the degree lower bounds to obtain size lower bounds in subsequent sections.

\paragraph{Computing the exact coefficient of $\xbar_S\cdot \ybar_T$} Recall, here $S = \{i \in [2n] \mid i \text{ is odd}\}$ and $T = [2n]\setminus S$. 
    In the following claim, we obtain the exact coefficient of $\xbar_S\cdot \ybar_T$. We will use this calculation again for \Cref{it:787} in \Cref{lem:four-conditions}.

    \begin{lemma}
    \label{lem:xsyt}
        The coefficient of $\xbar_S\cdot \ybar_T$ in $f(\xbar, \ybar)$ is equal to $\frac{1}{\beta} + \frac{1}{1-\beta}$. That is, the constant $C$ in the above computation is $-1$. 
    \end{lemma}
    \begin{proof}
        In order to understand $C$, let us simplify the term we want to analyse. 
        We will use the following simple fact about binomial coefficients in the proof.
        \begin{fact}
        \label{fact:binomial}
            $$\sum_{0 \leq j \leq n-1} {{n}\choose{j}} (-1)^{n-j} = -1\,.$$
        \end{fact}
        We have the following
        \begin{align}
            & \coeff{\xbar_S\cdot \ybar_T}\left(\sum_{A \subseteq S, B \subseteq T} f(\mathbbm{1}_A, \mathbbm{1}_B ) \xbar_A \cdot \widetilde{\xbar_{S\setminus A}} \cdot \ybar_B \cdot \widetilde{\ybar_{T \setminus B}}\right)  \nonumber \\
            = ~&\coeff{\xbar_S\cdot \ybar_T}\left(\sum_{A \subsetneq S, B \subsetneq T} f(\mathbbm{1}_A, \mathbbm{1}_B ) \xbar_A \cdot \widetilde{\xbar_{S\setminus A}} \cdot \ybar_B \cdot \widetilde{Y_{T \setminus B}}\right) + \coeff{\xbar_S\cdot \ybar_T}\left(\frac{1}{1-\beta} \cdot \xbar_S \cdot \ybar_T\right) \nonumber\\
            + ~& \coeff{\xbar_S\cdot \ybar_T}\left(\sum_{B \subsetneq T} f(\mathbbm{1}_S, \mathbbm{1}_B ) \xbar_S\cdot \ybar_B \cdot \widetilde{\ybar_{T \setminus B}}\right) + \coeff{\xbar_S\cdot \ybar_T}\left(\sum_{A \subsetneq S} f(\mathbbm{1}_A, \mathbbm{1}_T ) \xbar_A\cdot  \widetilde{\xbar_{S \setminus A}} \cdot \ybar_T \right) \label{eq:term2}  \\
            = ~ &\coeff{\xbar_S\cdot \ybar_T}\left(\sum_{A \subsetneq S, B \subsetneq T} \frac{-1}{\beta} \cdot \xbar_A \cdot \widetilde{\xbar_{S\setminus A}} \cdot \ybar_B \cdot \widetilde{\ybar_{T \setminus B}}\right) + \left(\frac{1}{1-\beta} \right) \nonumber\\
            + ~& \coeff{\xbar_S\cdot \ybar_T}\left(\sum_{0 \leq |B| \leq n-1, B \subsetneq T} -\frac{1}{\beta} \cdot  \xbar_S\cdot \ybar_B \cdot \widetilde{\xbar_{T \setminus B}}\right) + \coeff{\xbar_S\cdot \ybar_T}\left(\sum_{0 \leq |A| \leq n-1, A \subsetneq S} -\frac{1}{\beta} \cdot \xbar_A\cdot  \widetilde{\xbar_{S \setminus A}} \cdot \ybar_T \right). \label{eq:term3-4}
        \end{align}

        Observe that \Cref{eq:term2} has $4$ terms in it. The second term arises from taking $A = S$ and $B = T$. Due to our choice of $S,T$, the coefficient of this term is simply $1/(1-\beta)$.  
        
        Let us now analyse Terms 1, 3, and 4 from \Cref{eq:term3-4}.

        \begin{align*}
             \text{Term 1} = ~&\coeff{\xbar_S\cdot \ybar_T}\left(\sum_{A \subsetneq S, B \subsetneq T} \frac{-1}{\beta} \cdot \xbar_A \cdot \widetilde{\xbar_{S\setminus A}} \cdot \ybar_B \cdot \widetilde{\ybar_{T \setminus B}}\right) \\
             = ~&\coeff{\xbar_S\cdot \ybar_T}\left(\sum_{0 \leq |A| \leq n-1, ~0 \leq |B| \leq n-1, ~A\subsetneq S, ~B\subsetneq T}\frac{-1}{\beta} \cdot \xbar_A \cdot \widetilde{\xbar_{S\setminus A}} \cdot \ybar_B \cdot \widetilde{\ybar_{T \setminus B}}\right) \\
            = ~ &\coeff{\xbar_S\cdot \ybar_T}\left(\frac{-1}{\beta} \cdot \left(\sum_{0 \leq |A| \leq n-1,~A\subsetneq S}\xbar_A \cdot \widetilde{\xbar_{S\setminus A}}\right)\cd\left( \sum_{0 \leq |B| \leq n-1, B\subsetneq T} \ybar_B \cdot \widetilde{\ybar_{T \setminus B}}\right)\right)\\
            = ~ &\left(\frac{-1}{\beta} \cdot \left(\sum_{0 \leq j \leq n-1} {{n}\choose{j}} (-1)^{n-j} \right)\cd\left(\sum_{0 \leq j \leq n-1} {{n}\choose{j}} (-1)^{n-j}\right)\right) \\
            = ~ & \left(\frac{-1}{\beta} \cdot \left(\sum_{0 \leq j \leq n-1} {{n}\choose{j}} (-1)^{n-j} \right)^2\right) \\
            =~& -\frac{1}{\beta} \quad \quad \quad~~~ \text{ (using~\Cref{fact:binomial})}.
        \end{align*}

        Similarly, 
        \begin{align*}
            \text{Term 2} = ~& \coeff{\xbar_S\cdot \ybar_T}\left(\sum_{0 \leq |B| \leq n-1, B \subsetneq T} -\frac{1}{\beta} \cdot  \xbar_S\cdot \ybar_B \cdot \widetilde{\ybar_{T \setminus B}}\right) \\
            = ~& \coeff{\ybar_T}\left(\sum_{0 \leq |B| \leq n-1, B \subsetneq T} -\frac{1}{\beta} \cdot  \ybar_B \cdot \widetilde{\ybar_{T \setminus B}}\right) \\
            = ~& -\frac{1}{\beta} \cdot \left( \sum_{0 \leq j \leq n-1}  {{n}\choose{j}} (-1)^{n-j}\right) = \frac{1}{\beta} \quad \quad \quad~~~ \text{ (using~\Cref{fact:binomial})}.
        \end{align*}
        And by symmetry, we also get that 
        \[\text{Term 3} = \coeff{\xbar_S\cdot \ybar_T}\left(\sum_{0 \leq |A| \leq n-1, A \subsetneq S} -\frac{1}{\beta} \cdot  \xbar_A \cdot \widetilde{\xbar_{S \setminus A}} \cdot \ybar_T\right) = \frac{1}{\beta}\,. \]

        And now, using all the above values in the computation, we get that 
        \[\coeff{\xbar_S \ybar_T} \left(f(\xbar,\ybar)\right) = \frac{1}{\beta} + \frac{1}{1-\beta} = \frac{1}{\beta (1-\beta)}\,.\]
    \end{proof}

As a corollary of this lemma, we get that the refutation of $Q(\xbar, \ybar)$ ought to have degree at least $2n$. Specifically, we get the following statement. 
\begin{theorem}
    \label{thm:ns-degreeLB-Q}
    Let $\mathbb{F}$ be any field of characteristic $5$ or more and let $\beta \notin \{-1, 0, 1\}$. Then, $Q(\xbar, \ybar)$, $\{x_i^2-x_i\}_i$, and   $\{y_i^2-y_i\}_i$ are unsatisfiable and any polynomial $f(\xbar, \ybar)$ with $f(\xbar, \ybar) = 1/Q(\xbar, \ybar)$ for $\xbar \in \{0,1\}^{2n}$ and $\ybar \in \{0,1\}^{2n}$, satisfies that degree of $f(\xbar, \ybar)$ is at least $2n$.  
\end{theorem}





\section{Lifting Degree-to-Size $\mathbb{I}$: Symmetric Instances}

In this section we show how to use Nullstellensatz degree
lower bounds to obtain size lower bounds on a stronger proof system. In this sense, we ``lift'' a weak lower bound to a lower bound against a stronger  model. This is done using a \emph{gadget} $g$, or a \emph{lift}, which is simply a substitution in the original hard instance. Namely, given a polynomial $f(\vx)\in\F[\vx]$ we define a new polynomial $f'(\vy):=f(g(x_1),\dots,g(x_n))$, with $g(x_i)\in\F[\vy]$, for all $i\in[n]$. 


Recall the functional lower bound method
which is a reduction from algebraic circuit lower bounds to proof complexity lower bounds developed in \cite{FSTW21} shown above in 
\Cref{def:single-functional-lower-bound-method}.

\subsection{Size Lower Bounds for Symmetric  Instances via Lifting}
\label{Size Bounds for General Symmetric Hard Instance}

Here we  show how to lift the Nullstellensatz degree 
lower bound in \Cref{cor:NS-full-symmetric-deg-lower bound} to size lower bounds. 
In particular, we show that the lifting used 
in \cite{FSTW21} on the subset sum $\sum_i x_i -\beta$,
applies to \emph{every} unsatisfiable symmetric polynomial.

Henceforth, in this section we will assume $\F[\vx]$ is
equipped with some monomial order $\prec$.  
Unlike \cite{FSTW21} wherein the IPS proof size 
lower bound argument holds for \emph{any} chosen monomial order,
 our argument uses monomial orders that respect 
degree, i.e., $\deg(M)>\deg(N) \Rightarrow M \succ N$.
For concreteness, one can consider the 
graded lexicographic ordering on monomials 
\textbf{grlex} (see Cox, Little 
and O'Shea~\cite[Definition 5, page 58]{CoxLittleOShea15}).
 
\subsubsection{roABP-IPS Lower Bounds in Fixed Order}
\label{sec:roABP-IPS-lwbd-fxd-ord}

Here we show an exponential  lower bound 
against \roAlbIPS\ for any sufficiently low degree 
symmetric polynomial, where the roABPs are in a \emph{fixed} order of variables.  In the next section we extend this to a lower bound for \emph{every} order of variables.

For $\vx,\vy$ variables, with $|\vx|=|\vy|=n$, 
we use $\vx\circ\vy$ to denote the entry-wise product $(x_1 y_1,\dots,x_n y_n)$.
In other words, the \emph{gadget} we use is the 
mapping 
$$
x_i\mapsto x_iy_i\,,
$$
which substitutes the variable $x_i$ by $x_iy_i$, for
every $i$. 
We use $\ind{S}\in\bits^n$ to denote the indicator vector for a set $S$. 

We will need the following lemma that bounds from below
the number of distinct leading monomials of the set
of substitutions in a symmetric polynomial of high 
enough degree. We need this bound because, unlike 
\cite{FSTW21}, for symmetric polynomials of degree bigger than 1
we do not have a tight degree lower bound of $n$.
We do not attempt to optimise this lower bound, rather 
show an $2^{\Omega(n)}$ lower bound, whenever $d=\log^{O(1)}(n)$. 

\begin{lemma}\label{lem:num-dstnct-LTs}
Let $f(\vx)$ be a symmetric polynomial
with $n$ variables of degree $d=\log^k(n)$, for some constant $k$, that has no Boolean roots. 
Let $g(\vx,\vy)\cd f(\vx\circ\vy) =1 \mod \baxioms$. 
Then, 
\begin{equation}
\label{eq:1367:169}
\left|\LM
 \Big(
    \left\{
        \ml(g(\vx,\ind{S}))
         : S\subseteq[n]
    \right\}
 \Big)
    \right|
    \ge 2^{\Omega(n)}\,.
\end{equation}
\end{lemma}

\begin{proof}
Assume for simplicity that $n$ is even and let 
$$
\cD:=\{S\;:\; S\subseteq [n] \text{ and } |S|=n/2 \}.
$$

We are going to show that 
$
\left|
\LM
 \Big(
    \left\{
        \ml(g(\vx,\ind{S}))
         : S\in\cD
    \right\}
 \Big)
    \right|
    \ge 2^{\Omega(n)}\,
$, 
which is enough to conclude 
\Cref{{eq:1367:169}}.
    
\begin{claim}\label{cla:fully-sym-counting}
There are $2^{\Omega(n)}$ many pairwise disjoint sets
$S_1,\dots,S_\ell\subseteq \cD$~(with $\ell=2^{\Omega(n)}$) that
induce distinct leading monomials in 
$    \ml(g(\vx,\ind{S_i}))$, in the sense that:
$\forall i\neq j\in[\ell], 
\LM
 \left(
    \ml(g(\vx,\ind{S_i}))
 \right) 
 \neq  
\LM
 \left(
    \ml(g(\vx,\ind{S_j}))
 \right).
$
\end{claim}
\begin{proofclaim}\label{cla:fully-sym-count}
The idea is to show that most pairs of sets in $\cD$
induce distinct leading monomials.
More precisely, for any given $S\in\cD$ there are only $n^{\log ^{O(1)}n}$ many sets $S'\in \cD$ that induce the same leading monomial, 
namely 
$\LM
 \left(
    \ml(g(\vx,\ind{S}))
 \right) 
 =  
\LM
 \left(
    \ml(g(\vx,\ind{S'}))
 \right)$.
Since the number of subsets in $\cD$ is $2^{\Omega(n)}$
we get that there exists at least 
$2^{\Omega(n)}/n^{\log ^{O(1)}n} = 2^{\Omega(n)}$ sets
$S_1,\dots,S_\ell$~(with $\ell=2^{\Omega(n)}$) that
induce distinct leading monomials  $\ml(g(\vx,\ind{S_i}))$. 
For that purpose, it is sufficient to show that for every set 
$S\in\cD$ there exists 
a set $\cL_S\subseteq\cD$, such that 
\begin{enumerate}
\item $S\in\cL_S$; and \label{it:2066-1}
\item $|\cL_S|=n^{\log ^{O(1)}n}$; and \label{it:2066-2}
\item if $T\in\cD\setminus\cL_S$ and  \label{it:2066-3}
$T'\in\cL_S$, then 
$\LM
 \left(
    \ml(g(\vx,\ind{T}))
 \right) 
 \neq  
\LM
 \left(
    \ml(g(\vx,\ind{T'}))
 \right).
$
\end{enumerate}
(This is sufficient to conclude the claim because starting from
an arbitrary $S\in\cD$, we can pick an $S'\in\cD\setminus\cL_S$ 
such that $S,S'$ induce distinct leading monomials. After which
we pick $S''\in\cD\setminus(\cL_S\cup\cL_{S'})$ which induces
yet another distinct leading monomial. Since the size
of each $|\cL_S|,|\cL_{S'}|,\dots$ is $n^{\log ^{O(1)}n}$ we can continue
this process at least $2^{\Omega(n)}/n^{\log ^{O(1)}n}$ many times.)
 
So, let $S\in\cD $.
By assumption that $g(\vx,\vy)\cd f(\vx\circ\vy) =1 \mod \baxioms$,
we have
\begin{equation}
\label{eq:g-times-f-hadamard}
    \ml(g(\vx,\ind{S}))\cd f(\vx\circ\ind{S}) =1 
    \mod\baxioms\,, 
\end{equation}
(note that multilinearizing $g(\vx,\ind{S})$ does not affect
the equality, since we work modulo \baxioms).

Because our lifting turns every variable $x_i$ into 
$x_iy_i$ we have that  \(\ml(g(\vx,\ind{S}))\) is a (multilinear
symmetric) polynomial that \emph{depends  on the variables
$x_i$, for $i\in S$}
(that is, each nonzero monomial in this polynomial
has only variables $x_i$, for $i\in S$ [though not necessarily
all of such $x_i$'s]). 
Similarly, $f(\vx\circ\ind{S})$ is a (symmetric) polynomial that \emph{depends on the variables $x_i$, for $i\in S$}.
Since $f(\vx)$ has no Boolean roots, $f(\vx\circ\vy)$ also does not
have Boolean roots (if there was a Boolean root 
for the latter, there was also a Boolean root for the former; note that $x_i y_i\in\bits$ whenever $x_i,y_i\in\bits$).
This, together with \Cref{eq:g-times-f-hadamard}, mean that the 
conditions of \Cref{cor:NS-full-symmetric-deg-lower bound} are met, so we have
$$
|S|-d+1\le \deg(\ml(g(\vx,\ind{S})))\le|S|.
$$
Since we assumed that our monomial ordering respects degree,
\begin{equation}\label{eq:1386:30}
|S|-d+1\le\deg(\LM(\ml(g(\vx,\ind{S}))))\le|S|.
\end{equation}
Given a set $S'\subseteq[n]$, denote by $\widehat x_{S'}$ 
the \emph{corresponding multilinear monomial} $\prod_{i\in S'}x_i$.
And conversely, given a monomial $M$, denote by $S_M$ its ``support'',
namely, the \emph{corresponding subset of} $[n]$ such that 
$\widehat x_{S_{M}}= M$.
 
Denote by $M_0$ the monomial $\LM (\ml(g(\vx,\ind{S}))$ (for the $S\in\cD$ we fixed above).
Note that $M_0$ does not necessarily equal $\widehat x_{S}$, 
because our degree lower bound in \Cref{eq:1386:30}
is not tight. In other words, $S_{M_0}$ does not necessarily equal
$S$, but rather we only know that $M_0$ consists of at least $|S|-d+1$  variables
$x_i$, for $i\in S$ (and no other variables):   
$$
S_{M_0}\in\set{S'\subseteq S \;:\; |S|-d+1\le |S'|\le|S| }.
$$

Note that by construction of the lifting, if $S'\in\cD$ (i.e., $|S'|=n/2$) and $\LM(\ml(g(\vx,\ind{S'}))=M_0$, 
then $S'\supseteq S_{M_0}$ (because the only variables in  
$\ml(g(\vx,\ind{S'}))$ are $x_i$, for $i\in S$).
Thus, by \Cref{eq:1386:30}
$$
\left|
    \left\{S'\in\cD \;:\; 
    S'\supseteq S_{M_0} \right\}
\right|\le 
    {n/2+d-1 \choose d-1} \approx  
    \left(\frac{n/2+d-1}{e}\right)^{d-1} \le  n^{\log^c n}
$$
for some constant $c$ and sufficiently big $n$ (since 
$d=\log^{k}n $, for a constant $k$).

Note that putting
$\cL_S:= \left\{S'\in\cD \;:\; 
    S'\supseteq S_{M_0} \right\}$ we obtain an $\cL_S$ that meets all three conditions  \Cref{it:2066-1} to \Cref{it:2066-3}
above. 
\end{proofclaim}

This concludes the proof of \Cref{lem:num-dstnct-LTs}.
\end{proof}

\begin{theorem}\label{thm:roABP-fixed-order-lower-bound}
Let $f(\vx)$ be an unsatisfiable symmetric polynomial
with $n$ variables of degree $d=\log^{O(1)} n$. Then, any 
\roAlbIPS\ refutation of $f(\vx\circ\vy)=0$ is 
of size $2^{\Omega(n)}$, when the variables are ordered
such that $\vx<\vy$ (i.e., \vx-variables come 
before \vy-variables).
\end{theorem}

\begin{proof}
This is similar to similar to \cite[Proposition 5.8]{FSTW21}, that we repeat for convenience. 

Let $g(\vx,\vy)$ be a polynomial such that 
$
g(\vx,\vy)\cd f(\vx\circ\vy) = 1 \text{ over~ } 
\vx,\vy\in\bits^n.
$
Hence, 
\begin{equation}\label{eq:1266:31}
g(\vx,\vy) = \frac{1}{f(\vx\circ\vy)}  \text{ ~~over~ }
\vx,\vy\in\bits^n.
\end{equation}
We show that $\dim\coeffs{\vx|\vy} g\ge 2^{\Omega(n)}$.
This will conclude the proof by 
\Cref{res:roABP-width_eq_dim-coeffs} which will give the
roABP size (width) lower bound and by the functional lower 
bound reduction in \Cref{def:single-functional-lower-bound-method}.

First, observe that $f(\vx\circ\vy)=0$ is unsatisfiable
over $\vx,\vy\in\bits^n$, since $f(\vx)=0$ is. 
Thus, the right hand side of \Cref{eq:1266:31} is defined.

By lower bounding coefficient dimension by the 
evaluation dimension over the Boolean cube 
(\autoref{res:evals_eq-coeffs}),
        \begin{align*}
                \dim\coeffs{\vx|\vy} g
                &\ge \dim\evals{\vx|\vy,\bits} g\\
                &=\dim \{g(\vx,\ind{S}) : S\subseteq[n]\}\\
                &\ge\dim \{\ml(g(\vx,\ind{S})) : S\subseteq[n]\}
                \;.
        \end{align*}
         Here we used that dimension is non-increasing under linear maps. 
For $S\subseteq[n]$, denote by 
$\vx_S:=\{x_i \;:\; i\in S\}$ and note that for $\vx\in\bits^n$,         \[
                g(\vx,\ind{S})=\frac{1}{f(\vx_S)}
                \;.
        \]
        It follows that $\ml(g(\vx,\ind{S}))$ is a multilinear polynomial only depending on $\vx|_S$ (\autoref{fact:multilinearization}), and by its functional behavior it follows from 
\Cref{lem:sym-poly-ref-degree-high} 
that $\deg \ml(g(\vx,\ind{S}))\in[|S|-d+1,|S|]$. 

Since $\ml(g(\vx,\ind{S}))$ is multilinear 
we can use \Cref{lem:num-dstnct-LTs} to lower bound 
the number of distinct leading monomial of
$\ml(g(\vx,\ind{S}))$, when $S$ ranges over subsets 
of $[n]$:
\[
    \left|\LM \Big(\{\ml(g(\vx,\ind{S}))
         : S\subseteq[n]\}\Big)
    \right|
    \ge 2^{\Omega(n)}\,.
\]

Therefore, we can lower bound the dimension of the above space by the number of leading monomials (\autoref{res:dim-eq-num-TM-spn}),
        \begin{align*}
                \dim\coeffs{\vx|\vy} g
                &\ge\dim \{\ml(g(\vx,\ind{S})) : S\subseteq[n]\}\\
                &\ge\left|\LM \Big(\{\ml(g(\vx,\ind{S})) : S\subseteq[n]\}\Big)\right|\\
                &\ge 2^{\Omega(n)}\,.
                \qedhere
        \end{align*}
\end{proof}

\subsubsection{roABP-IPS in Any Order Lower Bounds}

Here we extend the results of the previous section to any variable order, which will imply lower bounds against \roAlbIPS\ of any variable order as well as multilinear formulas IPS. This extends the corresponding results in \cite{FSTW21} to lifting of the subset sum to lifting of every symmetric instance. 
 
Given a symmetric polynomial $f(\vq)$ with $n$ variables $q_1,\dots, q_n$, by \Cref{prop:fund-thm-sym-polynomials} we have  
$$
f(\vq):= g(y_1/\el_{1,n}(\vq),\dots,y_n/\el_{n,n}(\vq))
$$ 
for some polynomial $g(\vy)$. We now define a polynomial that will embed in itself the lifting from \Cref{sec:roABP-IPS-lwbd-fxd-ord}
of the symmetric polynomial $f(\vq$) under suitable partial Boolean substitutions. In other words, we define a polynomial over the  new variables $\vz,\vx$ such that under suitable Boolean assignments to the \vz\ variables we obtain the hard polynomial from \Cref{thm:roABP-fixed-order-lower-bound}. For each different such suitable Boolean assignment to \vz\ we will get an instance that is hard for a different variable ordering, concluding that our instance is hard for any variable ordering.

 Consider the polynomial  
\begin{equation}
f'(\vw):=g(y_1/\el_{1,m}(\vw),\dots,y_n/\el_{n,m}(\vw))
\end{equation}
 for $m= {2n\choose 2}$ and $\vw = \{w_{i,j}\}_{i<j\in[2n]}$.
We now apply a similar gadget to \cite{FSTW21},
defined by the mapping 
$$
w_{i,j}\mapsto z_{i,j}x_ix_j\,,
$$
which substitutes the $m$ variable $w_{i,j}$ by $m+2n$ variables $\{z_{i,j}\}_{i<j\in[2n]}$, $x_1,\dots,x_{2n}$:
\begin{equation}
\label{eq:1588:15}
f^\star(\vz,\vx):=g(y_1/(\el_{1,m}(\vw))_{w_{i,j}\mapsto z_{i,j}x_ix_j},\dots,y_n/(\el_{n,m}(\vw))_{w_{i,j}\mapsto z_{i,j}x_ix_j})\,,
\end{equation}
where $(\el_{j,m}(\vw))_{w_{i,j}\mapsto z_{i,j}x_ix_j}$ means that we apply the lifting $w_{i,j}\mapsto z_{i,j}x_ix_j$ to the \vw\ variables.  

Let $f\in\F[\vx,\vy,\vz]$. 
We denote by $f_\vz$ the polynomial $f$ considered as a polynomial in $\F[\vz][\vx,\vy]$, namely as a polynomial whose indeterminates are \vx,~\vy\ and whose scalars are from the ring $\F[\vz]$. 
We will consider the dimension of a (coefficient) matrix
when the entries are taken from the ring $\F[\vz]$, and where the dimension is considered over the field of rational functions $\F(\vz)$.  
Note that for any $\vaa\in\F^{|\vz|}$ we have that $f_\vaa(\vx,\vy)=f(\vx,\vy,\vaa) \in\F[\vx,\vy]$.  
 We use the following simple lemma:

\begin{lemma}[\cite{FSTW21}]\label{res:coeff-dim:fraction-field}
Let $f\in\F[\vx,\vy,\vz]$. Then for any  $\vaa \in\F^{|\vz|}$
        \[
                \dim_{\F(\vz)}\coeffs{\vx|\vy} f_\vz(\vx,\vy) \ge \dim_\F \coeffs{\vx|\vy}f_\vaa(\vx,\vy)
                \;.
        \]
\end{lemma}

\begin{proposition}
\label{res:lbs-fn:any-order:coeff-dim}
Let $n\ge 1$, $m={n \choose 2}$,  and $\F$ be a field with $\chara(\F)>\max(2^{4n+2m},n^d)$.  
Let $f\in\F[\vq]$ be a symmetric polynomial with $n$ variables
of degree $d=O(\log n)$, and $f^\star(\vz,\vx)$ be as in \Cref{eq:1588:15}.
Suppose that  $g\in\F[z_1,\ldots,z_{m},x_1,\ldots,x_{2n}]$ be a polynomial such that
\[
                g(\vz,\vx)=\frac{1}{f^\star(\vz,\vx)-\beta}
                \;,
\]
        for $\vz\in\bits^{\binom{2n}{2}}$ and $\vx\in\bits^{2n}$, and $\beta\in\F$ that makes  $f^\star(\vz,\vx)-\beta=0$ as well as $f(\vq)-\beta=0$ each unsatisfiable over Boolean values.\footnotemark ~Let $g_\vz$ denote $f$ as a polynomial in $\F[\vz][\vx]$.  Then, for any partition $\vx=(\vu,\vv)$ with $|\vu|=|\vv|=n$,
        \[
                \dim_{\F(\vz)}\coeffs{\vu|\vv} g_\vz \ge 2^{\Omega(n)}
                \;.
        \]
\end{proposition}
\iddolater{Explain the footnote: Explain this with an exact computation of beta depending the char}
\footnotetext{Observe that  when the characteristic is sufficiently large there always exists a $\beta$ that makes both 
$f^\star(\vz,\vx)-\beta=0$ 
and $f(\vq)-\beta=0$ unsatisfiable over Boolean values for the variables $\vz,\vx,\vq$.}

\begin{proof}
We proceed as in \cite{FSTW21} to embed $\frac{1}{f(\vu\circ\vv)-\beta}$ in this instance via a restriction of $\vz$. Define the $\vz$-evaluation $\vaa\in\bits^{\binom{2n}{2}}$ to restrict $g$ to sum over those $x_ix_j$ in the natural matching between $\vu$ an $\vv$, so that
        \[
                \alpha_{ij}
                =\begin{cases}
                        1       &       x_i=u_k, x_j=v_k\\
                        0       &       \text{else}
                \end{cases}
                \;.
        \]
It follows that $g(\vu,\vv,\vaa)=\frac{1}{f(\vu\circ\vv)-\beta}$ for $\vu,\vv\in\bits^{n}$. Thus, by using the  our lower bound for a fixed partition (\Cref{thm:roABP-fixed-order-lower-bound})
 and the relation between the coefficient dimension in $f_\vz$ versus $f_\vaa$ (\Cref{res:coeff-dim:fraction-field}),
        \begin{align*}
                \dim_{\F(\vz)}\coeffs{\vu|\vv} g_\vz(\vu,\vv)
                &\ge \dim_\F \coeffs{\vu|\vv}g_\vaa(\vu,\vv)\\
                &\ge 2^{\Omega(n)}
                \;.
                \qedhere
        \end{align*}
\end{proof}
\iddo{rephrase proof above and below, as it's verbatim FSTW21. \textbf{Why did I restrict to log n degree? For proof above I think.}}

\begin{corollary}\label{res:lbs-fn:lbs-ips:vary-order}
Let $n\ge 1$, $m={n \choose 2}$,  and $\F$ be a field with $\chara(\F)>\max(2^{4n+2m},n^d)$.   
Let $f\in\F[\vq]$ be a symmetric polynomial with $n$ variables
of degree $d=O(\log n)$, and $f^\star(\vz,\vx)$ be as in \Cref{eq:1588:15}. 
Let $\beta\in\F$ be such that $f^\star(\vz,\vx)-\beta=0$ and $f(\vq)-\beta=0$ are each unsatisfiable over Boolean values. 
Then, any \roAlbIPS\ refutation (in any variable order) of  $f^\star(\vz,\vx)-\beta=0$ requires $2^{\Omega(n)}$-size.
\end{corollary}

\begin{proof}

Consider the polynomial $g\in\F[x_1,\ldots,x_{2n},z_1,\ldots,z_{\binom{2n}{2}}]$ such that
\[
                g(\vz,\vx)=\frac{1}{f^\star(\vz,\vx)-\beta}
                \;,
\]
        for $\vx\in\bits^{2n}$, $\vz\in\bits^{\binom{2n}{2}}$. \textcolor[rgb]{0.752941,0.752941,0.752941}{}

We continue similar to the proof of \cite[Corollary 5.14, part 1]{FSTW21}. Suppose that $f(\vx,\vz)$ is computable by a width-$r$ roABP in \emph{some} variable order.  By pushing the $\vz$ variables into the fraction field, it follows that $f_\vz$ ($f$ as a polynomial in $\F[\vz][\vx]$) is also computable by a width-$r$ roABP over $\F(\vz)$ in the induced variable order on $\vx$ (\Cref{fact:roABP:closure}). \iddo{Check if the fact indeed justify this.} 
By combining the coefficient dimension lower bound of \Cref{res:lbs-fn:any-order:coeff-dim} with its relation to roABPs (\Cref{res:roABP-width_eq_dim-coeffs}), and by splitting $\vx$ in half along its variable order we  obtain that any roABP computing $g$ requires 
width $\ge 2^n$ in any variable order.
The roABP-\lbIPS\ lower bound follows immediately from this functional lower bound for $g$ along with our reduction (\Cref{def:single-functional-lower-bound-method}).
\iddo{Check Lin' as before!!}
\end{proof}

\iddolater{Comment about the inability to get multilinear formulas IPS lower bounds.}





\section{Lifting Degree-to-Size $\mathbb{II}$: Vector Invariant Polynomials}
\label{sec:Invariant-size-lower-bounds}

In this section, we will prove an roABP-IPS$_{\text{LIN}}$ size lower bound for the vector-invariant-inspired hard instance. 
    
 In~\cite{FSTW21}, roABP-IPS$_{\text{LIN}}$ size lower bound was proved for a \emph{lifted version} of the subset sum instance. In a similar spirit, one may try to lift the 
 polynomial $Q(\xbar, \ybar)$ from \Cref{eq:Qxy}. This does not seem very straightforward, because 
 unlike the subset sum instance, which is linear, here we have a large degree instance. Fortunately, to obtain an roABP-IPS$_{\text{LIN}}$ size lower bound in the order $\xbar < \ybar$ we do not need any lift (in  \Cref{sec:invrnt-any-ordr}
we do need to  use a lift to get our result to work for any 
order). We can prove a lower bound for $Q(\xbar, \ybar)$ directly. Moreover, we obtain a lower bound that holds over all fields of characteristic greater than or equal to $5$. No such roABP-IPS$_{\text{LIN}}$ size lower bound was known over small characteristic. 

On the other hand, the polynomial $Q(\xbar, \ybar)$ is not computable by an roABP in this order. In fact, provably an roABP in the order $\xbar < \ybar$ requires size $\exp(\Omega(n))$ to compute $Q(\xbar, \ybar)$. 

\subsection{roABP-IPS$_{\text{LIN}}$ Size Lower Bound for $Q(\xbar, \ybar)$}
\label{sec:roabp-sizeLB-Q}

    We will now state our main theorem in this section. 
    \begin{theorem}
        \label{thm:roabp-sizeLB-Q}
        Let $\mathbb{F}$ be a field of characteristic $\geq 5$ and let $\beta \notin \{-1, 0, 1\}$. Then, $Q(\xbar, \ybar)$, $\{x_i^2-x_i\}_i$, and   $\{y_i^2-y_i\}_i$ is unsatisfiable and  any 
        roABP-IPS$_{\text{LIN}}$ refutation in order of the variables where $\xbar$ precedes $\ybar$ requires size $\exp(\Omega(n))$. 
    \end{theorem}

\paragraph{Characterising the monomials of degree $2n$ in $f(\xbar,\ybar)$.} In \Cref{sec:invariant-axioms-deg-lower-bounds} we computed the coefficient of one of the degree $2n$ monomials of $f(\xbar, \ybar)$, namely the coefficient of $\xbar_S\cdot \ybar_T$, where $S = \{i \in [2n] \mid i \text{ odd}\}$ and $T = [2n]\setminus S$. 
Next, we will use the calculations from \Cref{lem:xsyt} to obtain the coefficients of the monomials of degree $2n$ in $f(\xbar,\ybar)$. We first start with some notation. 

For $\sigma \in \{0,1\}^n$, let  $s_{\sigma,i} = 2i-1$ if $\sigma_i
 =0$ and $s_{\sigma,i} = 2i$ if $\sigma_i =1$, for $i \in [n]$. And let $S_\sigma = \cup_{i \in [n]} \{s_{\sigma, i}\}
\subset[2n]$ (and accordingly, $\overline{S_\sigma}:=[2n]\setminus S_\sigma$). For example, if $\sigma = 0^n$, then $S_\sigma = S$ and $\overline{S_\sigma} = T$. 

The monomial $\xbar_{S_\sigma}\ybar_{\overline{S_\sigma}}$
is defined by picking the $x$-variable from the first 
 column of $M_i$ if $\sigma_i=0$ and  from the second column of $M_i$ if $\sigma_i=1$, where $M_i$ is an defined in \Cref{fac:matrix-item}, and then picking  the $y$-variables from the \emph{other} columns of each matrix. For example: if $\sigma = (0,1,0)$, then $\xbar_{S_\sigma} \ybar_{\overline{S_{\sigma}}} = $
$\bf{\textcolor[rgb]{0.690196,0.690196,0.690196}{x_1 y_2} 
\textcolor[rgb]{0.401961,0.4401961,0.4461}{x_4 y_3}\textcolor{black}{x_5 y_6}}$.

We now prove some  additional properties of $f(\xbar, \ybar)$, which will help us with the roABP-IPS$_{\text{LIN}}$ size lower bound. 

\begin{lemma}\label{lem:four-conditions}
    The polynomial $f(\xbar, \ybar)$ has the following properties over the Boolean cube.
    \begin{enumerate}    
        \item For any odd $j \in [2n]$, let   $\phi_j$ be the mapping  $x_{j}\hookrightarrow x_{j},  
x_{j + 1} \hookrightarrow x_{j+1}, 
y_{j}\hookrightarrow x_j + y_j, \text{ and } 
y_{j}\hookrightarrow x_{j+1} + y_{j+1}$. It leaves all the other variables unchanged. Then, $\phi_j(f(\xbar, \ybar)) = f(\xbar, \ybar)$ over the Boolean cube for any odd $j \in [2n]$.  
        \item \label{it:787} $f(\xbar, \ybar)$ has $2^n$ monomials of degree $2n$ with non-zero coefficients.  
        Moreover, for every $\sigma \in \{0,1\}^n$, the coefficient of $\xbar_{S_\sigma}\ybar_{\overline{S_\sigma}}$ in $f(\xbar, \ybar)$ is either $\frac{1}{\beta(1-\beta)}$ or $\frac{1}{\beta(1+\beta)}$. 
        \item If $S \subseteq [2n]$ and $0<|S|<n$, then for any $T \subseteq [2n]$, the coefficient of $\xbar_S\ybar_T$ is equal to $0$ and the coefficient of $\xbar_T\ybar_S$ is equal to $0$ in $f(\xbar, \ybar)$. 
        \item There are no other monomials of degree $2n$ that have non-zero coefficients in $f(\xbar, \ybar)$. 
    \end{enumerate}
\end{lemma}
Let us first prove \Cref{thm:roabp-sizeLB-Q} assuming this lemma. 
\begin{proof}[Proof of \Cref{thm:roabp-sizeLB-Q}]
    It is easy to see that $Q(\xbar, \ybar)$, $\{x_i^2-x_i\}_i$, and $\{y_i^2-y_i\}_i$ are unsatisfiable, as long as $\beta>2$. Now let $[f(\xbar, \ybar)]_{2n}$ denote the degree $2n$ slice of the polynomial $f(\xbar, \ybar)$. We now get the following inequalities. 
    \begin{align}
        \dim \coeff{\xbar|\ybar}(f(\xbar, \ybar)) \geq~ & \dim \coeff{\xbar|\ybar}([f(\xbar, \ybar)]_{2n}) \nonumber \\
        \geq ~& \left| \left\{ \xbar_{S_\sigma} \cdot \ybar_{\overline{S_{\sigma}}} \mid \sigma \in \{0,1\}^n\right\}\right|\nonumber\\
        =~ & 2^n \label{eq:dim-f}
    \end{align}
Here, the first inequality comes from the fact that the overall coefficient dimension is at least as much as that of the degree $2n$ slice. The second inequality follows from Parts 2, 3, and 4 of \Cref{lem:four-conditions}. 

The final equality follows from Part 2 in \Cref{lem:four-conditions}. The main observation is that as long as $\mathbb{F}$ has characteristic greater than or equal to $5$, all the coefficients of the monomials in the set $\left\{ \xbar_{S_\sigma} \cdot \ybar_{\overline{S_{\sigma}}} \mid \sigma \in \{0,1\}^n\right\}$ are non-zero. 

To conclude, note that the lower bound of $2^n$ on the coefficient dimension of $f(\xbar, \ybar)$ implies that $f(\xbar, \ybar)$ requires width $2^n$ to be computed as an roABP in the order $\xbar < \ybar$ (\Cref{res:roABP-width_eq_dim-coeffs}). 
\end{proof}

Now, we prove \Cref{lem:four-conditions}. 
\begin{proof}[Proof of \Cref{lem:four-conditions}]
\

\noindent\textbf{Part 1}: Over the Boolean hypercube, $f(\xbar, \ybar) = 1/Q(\xbar, \ybar)$. Hence, we get the following over the Boolean hypercube. 
\[\phi_j(f(\xbar, \ybar)) = \phi_j\left(\frac{1}{\left(\prod_{i \in [2n], i: \text{ odd}}x_iy_{i+1} - y_ix_{i+1}\right) - \beta}\right) = \frac{1}{\left(\phi_j\left(\prod_{i \in [2n], i: \text{ odd}}x_iy_{i+1} - y_ix_{i+1}\right)\right) - \beta} \]

As $\phi_j(x_jy_{j+1} - y_jx_{j+1}) = x_jy_{j+1} - y_jx_{j+1}$ and as $\phi_j$ keeps all other terms unchanged, we get that 
\[\frac{1}{\left(\phi_j\left(\prod_{i \in [2n], i: \text{ odd}}x_iy_{i+1} - y_ix_{i+1}\right)\right) - \beta} = \frac{1}{\prod_{i \in [2n], i: \text{ odd}} \left(x_iy_{i+1} - y_ix_{i+1}\right) - \beta} = \frac{1}{Q(\xbar, \ybar)} = f(\xbar, \ybar)\]

\noindent\textbf{Part 2}: {This is similar to the proof of \Cref{lem:xsyt}}. In fact, the computations for Terms 1, 3, and 4 remain exactly the same as before. The only thing that is different is the computation for Term 2. 

We will say that $\sigma \in \{0,1\}^n$ is odd if it has odd number of $1$s and even otherwise. {For example, in \Cref{lem:xsyt}, $\sigma = 0^n$ and hence it w}as even. Note that, when $\sigma$ is even, then Term 2 gives a coefficient of $\frac{1}{1-\beta}$. Hence, in this case, the coefficient of $\xbar_{S_\sigma}\cdot \ybar_{\overline{S_{\sigma}}}$ is the same as the coefficient of $\xbar_S\cdot \ybar_T$. That is, it is $\frac{1}{\beta(1-\beta)}$.  On the other hand, if $\sigma$ is odd, then Term 2 becomes $\frac{1}{-1-\beta}$. So, in this case, the coefficient of $\xbar_{S_\sigma}\cdot \ybar_{\overline{S_{\sigma}}}$ is $\frac{1}{\beta} + \frac{1}{-1-\beta}$. This is equal to $\frac{1}{\beta(1+\beta)}$. 
    
\noindent\textbf{Part 3}: Let us rewrite $f(\xbar, \ybar)$ as follows. \[f(\xbar, \ybar) = \sum_{m:~ \text{monomial in } x \text{ variables}} f_m(\ybar)\cd m = \sum_{U \subseteq [2n]} f_U(\ybar) \xbar_U\,. \]
    Here, we use the Boolean axioms for  simplification to get the second equality.
    First, observe that for $U = \emptyset$, $f_\emptyset(\ybar) = \frac{1}{-\beta}$ by using the fact that $f(\xbar, \ybar) = 1/Q(\xbar, \ybar)$. 
    
    Now, we will prove the statement by induction on the size of $U$. For $|U|=1$, 
    \begin{align*}
        \coeff{x_i}(f(\xbar, \ybar) & = \coeff{x_i}\left(\frac{1}{Q(\xbar, \ybar)}\right)\\
        \coeff{x_i}\left(f(\xbar, \ybar)\mid_{x_j =0 \forall j \neq i} \right)& = \coeff{x_i}\left(\frac{1}{Q(\xbar, \ybar)}\mid_{x_j =0 \forall j \neq i} \right)\\ 
        f_{x_i}(\ybar)x_1 + f_{\emptyset}(\ybar) &= \frac{1}{-\beta}\,.
    \end{align*}
    But we saw that $f_\emptyset(\ybar) = \frac{1}{-\beta}$, hence $f_{x_i}(\ybar) = 0$. 

    For the inductive case, we have a very similar argument. Let $|U|<n$. 

    \begin{align*}
        \coeff{\xbar_U}(f(\xbar, \ybar)) & = \coeff{x_U}\left(\frac{1}{Q(\xbar, \ybar)}\right)\\
        \coeff{X_U}\left(f(\xbar, \ybar)\mid_{x_j =0 \forall j \notin U} \right)& = \coeff{\xbar_U}\left(\frac{1}{Q(\xbar, \ybar)}\mid_{x_j =0 \forall j \notin U} \right)\\ 
        \sum_{V \subseteq U} f_V(\ybar)\xbar_V &= \frac{1}{-\beta}\,.
    \end{align*}
    By induction hypothesis, and by using the fact that $f_\emptyset(\ybar) = \frac{1}{-\beta}$, we get that $f_U(\ybar) = 0$. By a similar argument, we can also prove that the coefficient of $\xbar_T\ybar_S$  is equal to $0$. This finishes the proof of Part 3. 

\noindent\textbf{Part 4}: Assume for the sake of contradiction there is a monomial of degree $2n$ in $f(\xbar, \ybar)$ that does not have the structure as described in Part 2. Then it must contain an index $i \in [2n]$ such that the monomial contains $x_i$ as well as $y_i$ in it, that is, the monomial looks like this $x_iy_i\xbar_S \ybar_T$, where $|S|+|T|+2 = 2n$. 

    If either $|S\cup \{i\}| < n $ or $|T\cup \{i\}| < n $, then by Part 3 above, we know that the coefficient of the monomial will be zero. 
    
    Now suppose that $|S\cup \{i\}| = n $ and $|T\cup \{i\}| = n $.  By Part 1 above, we know that the set of invariants of $Q(\xbar, \ybar)$ and that of $f(\xbar, \ybar)$ are the same. Let us apply the map $\phi_i$, that is $x_i \hookrightarrow x_i, y_i \hookrightarrow x_i + y_i, x_{i+1} \hookrightarrow x_{i+1}, y_{i+1} \hookrightarrow x_{i+1} + y_{i+1}$. This gives us $x_i (x_i+y_i)\phi_i(\xbar_S\ybar_T)$. 
    
    Consider the case when $i+1 \notin T$. Then, we get $x_i (x_i+y_i)\phi_i(\xbar_S\ybar_T) = x_i^2\xbar_S\ybar_T + x_iy_i\xbar_S\ybar_T$. So, if the coefficient of $x_iy_i\xbar_S\ybar_T$ is non-zero, then after applying the map, the coefficients of the resulting monomials are also non-zero. Notice that, we are working modulo the Boolean  axioms, which means, the expression above will simplify to $x_i\xbar_S\ybar_T + x_iy_i\xbar_S\ybar_T$. Now notice that by Part 3 above, the coefficient of the monomial $x_i\xbar_{S}\ybar_T$ must be zero, as $|T|<n$, which is a contradiction. 

    On the other hand, if $i+1 \in T$, then $x_i (x_i+y_i)\phi_i(\xbar_S\ybar_T) = x_i(x_i+y_i) \xbar_S \ybar_{T \setminus \{i+1\}} (x_{i+1} + y_{i+1})$. After expanding, we see that one of the resulting monomials is again $x_i\xbar_{S}\ybar_T$. So, if $x_iy_i\xbar_S \ybar_T$ has a non-zero coefficient, then so should this monomial. But by Part 3 above, we know that this monomial must have coefficient $0$, which gives a contraction. 
\end{proof}




\subsection{Coefficient Dimension in any Variable Order}\label{sec:invrnt-any-ordr}

In the previous section we proved a lower bound on the coefficient dimension in the $\xbar|\ybar$ variable partition. In this section we extend that result to give bounds on the coefficient dimension in any order. 

To achieve this, we use a similar lifting as before. 
Namely, we lift the instance to a new polynomial, $P(\ubar, \zbar)$, using the new auxiliary variables \vz. The polynomial $P(\ubar, \zbar)$ has the property that given a partition of $\ubar$ variables into two equal parts, there exists a $0/1$ assignment to the auxiliary variables that reveals a hard \emph{planted} instance of $Q$ (from our previous section) inside $P$. 
Below, we will start with the description of the hard instance.

\paragraph{Hard instance.} Let $\ubar = \{u_1, u_2, \ldots, u_{4n}\}$, let $m = {{4n}\choose{4}}$, and $\zbar = \{z_1, z_2, \ldots, z_m\}$. Let $P(\ubar, \zbar) \in \mathbb{F}[\ubar, \zbar]$ be defined as follows.
\[P(\ubar, \zbar) = \left(\prod_{i<j<k<\ell \in [4n] } 1 - z_{i,j,k,\ell} + z_{i,j,k,\ell} (u_iu_\ell - u_ju_k)\right) - \beta\]

We will prove the following theorem about this polynomial. 

\begin{theorem}
    \label{thm:any-order}
    Let $\mathbb{F}$ be a field of characteristic $\geq 5$ and let $P(\ubar, \zbar)$ be as defined above. Then $P(\ubar, \zbar), \{u_i^2-u_i\}_i, \{z_i^2-z_i\}_i$ is unsatisfiable as long as $\beta \notin \{-1, 0, 1\}$. 
\iddo{deleted  g that is not used in the statement}   
%
%
    And any roABP-IPS$_{\text{LIN}}$ refutation of this instance requires $\exp(\Omega(n))$ size.  Any multilinear-formula-IPS requires $n^{\Omega(\log n)}$ size and any product-depth-$\Delta$ multilinear-formula-IPS requires size $n^{\Omega((n/\log n)^{1/\Delta}/\Delta^2)}$. 
\end{theorem}

\begin{remark}
We note the following salient points regarding the result. 
\begin{itemize}
    \item Let $N$ be the total number of variables in $P(\ubar,\zbar)$. Then, $N = O(m) = O({{4n}\choose{4}}) = O(n^4)$. Hence, in terms of $N$, the roABP size lower bound is $2^{\Omega(N^{1/4})}$. 

    \item The polynomial $P(\ubar, \zbar)$ is not multilinear. It is however a relatively easy-to-compute polynomial. Namely, it has a product-depth $2$ formula of polynomial size.

    \item The lower bound holds over all characteristics (as long as it is $\geq 5$). No IPS lower bounds over finite fields were known before. 
\end{itemize}
\end{remark}

The following technical lemma is used to prove the above theorem. 

\begin{lemma}
    \label{lem:any-order}
    Let $\mathbb{F}$ be a field with characteristic $\geq 5$, let $\beta \notin \{-1, 0, 1\}$, and let $g(\ubar, \zbar) \in \mathbb{F}[\ubar, \zbar]$ be a polynomial such that 
    \[g(\ubar, \zbar) = \frac{1}{P(\ubar, \zbar)},\]
    for $\ubar \in \{0,1\}^{4n}$ and $\zbar \in \{0,1\}^m$. Let $g_{\zbar}$ denote the polynomial in $\mathbb{F}[\zbar][\ubar]$. Then for any partition $\ubar = (\vbar, \wbar)$ such that $|\vbar| = |\wbar| = 2n$
    \[\dim_{\mathbb{F}[\zbar]} \coeff{\vbar|\wbar}( g_{\zbar}) \geq 2^n \]
\end{lemma}

We will first prove \Cref{thm:any-order} using \Cref{lem:any-order}. 

\begin{proof}[Proof of \Cref{thm:any-order}]
{\textit{roABP} lower bound:}  Assume that $g(\ubar, \zbar)$ is computable by width-$r$ roABP in some order of variables. By \emph{pushing} the $\zbar$ variables into the fraction field, $g_{\zbar} \in \mathbb{F}[\zbar][\ubar]$ is also computable by width-$r$ roABP over $\mathbb{F}[\zbar]$ in the induced order of the variables in $\ubar$. (This follows from \Cref{fact:roABP:closure}), which states that roABPs are closed under variable substitutions.) 

Now, by splitting the variables $\ubar$ in two halves in this order, say into $\vbar, \wbar$, we obtain a lower bound on the coefficient dimension of the roABP with partition $\vbar|\wbar$ using \Cref{lem:any-order}.  To conclude, note that the lower bound of $2^n$ on the coefficient dimension  implies a width $2^n$ lower bound for the roABP (\Cref{res:roABP-width_eq_dim-coeffs}). 

The lower bounds stated in \Cref{thm:any-order} for multilinear formulas and constant-depth multilinear formulas follow from our lower bound on the coefficient dimension and the results of Raz~\cite{Raz09} and Raz-Yehudayoff~\cite{RazYehudayoff09}. 
\end{proof}

We now conclude this section with the proof of \Cref{lem:any-order}.

\begin{proof}[Proof of \Cref{lem:any-order}]
    Given a partition of $\ubar$ into two equal parts $(\vbar, \wbar)$, consider the polynomial $g(\vbar, \wbar, \zbar)$. We would like to bound the coefficient dimension  of $g(\vbar, \wbar, \zbar)$ for partition $\vbar|\wbar$. For this, we will first view $g \in \mathbb{F}[\zbar][\vbar,\wbar]$ and try to bound the coefficient dimension in the field of rational functions $\mathbb{F}[\zbar]$. Following the notation from~\cite{FSTW21} (Lemma~5.12), we will denote this quantity by $\dim_{\mathbb{F}[\zbar]} \coeff{\vbar|\wbar} ~g_{\zbar}(\vbar, \wbar)$. 

    Another way to bound the coefficient dimension is to consider $\abar \in \{0,1\}^m$ and evaluate $\zbar \leftarrow \abar$ such that $g(\vbar, \wbar, \abar) \in \mathbb{F}[\vbar, \wbar]$. Thus, study the coefficient dimension over $\mathbb{F}$. We will denote this quantity by $\dim_{\mathbb{F}} \coeff{\vbar|\wbar}~ g_{\abar}(\vbar, \wbar)$. 

    It is known that $\dim_{\mathbb{F}[\zbar]} \coeff{\vbar|\wbar} ~g_{\zbar}(\vbar, \wbar) \geq \dim_{\mathbb{F}} \coeff{\vbar|\wbar} ~g_{\abar}(\vbar, \wbar)$. Therefore, it will suffice to lower bound the latter for an appropriate evaluation $\abar$ of the $\zbar$ variables. 

    Specifically, we will design $\zbar$-evaluation, that is, $\abar$ as follows. 
    \[\alpha_{i,j,k,\ell} = \left\{ \begin{array}{cl}
        1 & \text{if } i \in [2n], i~\text{odd },\\
        & u_i = v_i, u_j = v_{i+1},   \\
        & u_k = w_i, u_\ell = w_{i+1} \\
         0 & \text{ otherwise}
    \end{array}\right.\]
    For this evaluation, notice that when $z_{i,j,k,\ell} = 1$, we have that $i \in [2n]$ and is odd. Moreover, we get that the corresponding term in $P(\vbar, \wbar, \abar)$ becomes $\left(1 - 1 + 1 \cdot (v_iw_{i+1} - w_iv_{i+1})\right) = (v_iw_{i+1} - w_iv_{i+1})$. On the other hand, when $z_{i,j,k,\ell} = 0$, the corresponding term just becomes $\left(1 - 0 + 0\right) = 1$. Therefore, we get that \[g(\vbar, \wbar, \abar) = \frac{1}{\left(\prod_{i \in [2n], i \text{ odd}} (v_iw_{i+1} - w_i v_{i+1})\right) - \beta} = \frac{1}{Q(\vbar, \wbar)} = f(\vbar, \wbar)\] 
    for $\vbar, \wbar \in \{0,1\}^{2n}$. Now, using \Cref{eq:dim-f}, we get the lower bound of $2^n$ on the $\dim_{\mathbb{F}} \coeff{\vbar|\wbar} ~g_{\abar}(\vbar, \wbar)$, which concludes this proof.   
\end{proof}

\section{Lifting Degree-to-Size $\mathbb{III}$:  Lower Bounds against Constant-Depth Refutations}\label{sec:indiv-degree-size-lower-bounds}

In this section we prove lower bounds for constant-depth IPS, that is, IPS refutations that are computable by constant-depth algebraic circuits. The main theorem of this section is the following.

\begin{theorem}\label{thm:constant-depth-lower-bounds}
    Let $n,\Delta$ and $\delta$ be positive integers, and assume that $\mathrm{char}(\F) = 0$. Let $g$ be a polynomial of individual degree at most $\delta$ so that it agrees with 
    \[
    \frac{1}{\sum_{i,j,k,\ell\in[n]}z_{ijk\ell}x_ix_jx_kx_\ell - \beta}\quad\text{ over Boolean values.}
    \]
    Then any circuit of product-depth at most $\Delta$ computing $g$ has size at least
    \[
    n^{\Omega\left(\frac{(\log n)^{2^{1 - 2\Delta}}}{\delta^2\cdot\Delta}\right)}.
    \]
\end{theorem}

Using the \Cref{def:single-functional-lower-bound-method}
this gives a constant-depth IPS refutation lower bound (with the
same parameters and field) for the instance 
$\sum_{i,j,k,\ell\in[n]}z_{ijk\ell}x_ix_jx_kx_\ell - \beta$.

Previously, \cite{GHT22} proved lower bounds for the same instance for multilinear refutations, i.e., in the case where $\delta = 1$. Our  result improves on that work in two ways. Firstly, when $\delta = 1$, the result improves slightly the exponent in the expression (from $1/(2^{2\Delta} - 1)$ to $1/2^{2\Delta - 1}$). Secondly, and more importantly, the result gives lower bounds for larger individual degrees.

The lower bound shows a natural \emph{trade-off} between the depth and individual degree of refutations. It gives superpolynomial lower bounds for refutations of individual degree $\mathsf{poly}(\log\log n)$ for any constant depth refutations, and for any fixed depth we get lower bounds up to individual degree $\log^\epsilon n$ for some small $\epsilon$ depending on the depth.

Our hard instance does have a small constant-depth refutation, but of high individual degree. This is obtained by substituting the gadget $z_{ijk\ell}x_ix_jx_kx_\ell$ to the small depth-$3$ refutations of the standard  instance of knapsack ($\sum_i x_i-\beta$) given in \cite{FSTW21}. The obtained refutation is a polynomial of individual degree $O(n^3)$.
 
To prove our lower bound for the bounded individual degree refutations, we employ the framework put forward by Amireddy \emph{et al.} in \cite{AGK0T23} to prove constant-depth algebraic circuit lower bounds. They show that the lower bounds for constant-depth algebraic circuits originally proved in \cite{LST21} can also be obtained more directly via homogeneous constant-depth circuit lower bounds without going through the final hardness escalation step through set-multilinear circuits that was used in \cite{LST21}.

Amireddy \textit{et al.}~\cite{AGK0T23} suggest that their approach could also be used to prove functional lower bounds for constant-depth algebraic circuits. This seems viable since their framework puts leaner requirements for the polynomials than that of \cite{LST21}. Our work achieves this goal,  showing  how to modify their framework to obtain  \cref{thm:constant-depth-lower-bounds}. 

The rest of this section is organized as follows. First, we recall the framework of \cite{AGK0T23}, and discuss some modifications needed to prove our functional lower bounds. Then we recall an intermediate hard instance used in  \cite{GHT22} and prove lower bounds for that instance for the affine projections of partial $\APP$ complexity measure used in \cite{AGK0T23}
(while in \cite{GHT22} a different measure was used).
After this, we are ready to prove our main lower bound for this section.

\subsection{Lower Bounds via Affine Projections of Partials}

We introduce some notation that matches and extends that of \cite{AGK0T23}. For any non-negative integers $n$ and $k$ we denote by $M(n,k)$ the number of distinct monomials of degree exactly $k$ in $n$ variables, and by $M_\leq(n,k)$ the number of distinct monomials of degree at most $k$ in $n$ variables. The following lemma gives useful approximations for these quantities and is shown in  \cite[Lemma 2.2]{AGK0T23} (the only new information is that in the first item, the upper bound also applies to $M_\leq(n,k)$).

\begin{lemma}\label{lemma:bounds-for-M}
    Let $n\geq k\geq\ell$ and $m$ be positive integers. Then,
    \begin{itemize}
        \item[(i)] $(n/k)^k\leq M(n,k)\leq M_\leq(n,k)\leq (6n/k)^k$;
        \item[(ii)] $(n/2k)^\ell\leq \frac{M(n,k+\ell)}{M(n,k)}\leq (2n/k)^\ell$;
        \item[(iii)] $\frac{M(\ell,m)}{M(k,m)}\geq (\ell/k)^m.$
    \end{itemize}
\end{lemma}

The following quantity is crucial for the analysis of \cite{AGK0T23}. Let $d_1,\ldots,d_t$ be non-negative integers such that $d := \sum_{i\in [t]}d_i \geq 1$, and let $k < d$. Then define 
\[
\residue_k(d_1,\ldots,d_t) := \frac{1}{2} \min_{k_1,\ldots,k_t\in \integer}\sum_{j\in [t]}\left|k_j - \frac{k}{d}\cdot d_j \right|
\]
For constants $a,b$ we write $a\approx_c b$ if $a\in [b/c,b]$ and $a\approx b$ if $a\approx_c b$ for some unspecified constant $c$.

\subsubsection{The $\APP$ Measure}\label{se:APP-measure}

Let us now recall the \emph{Affine Projections of Partials} ($\APP$) measure that \cite{AGK0T23} used to prove their lower bounds. They actually considered in addition  another measure, the \emph{Shifted Partials} measure, but $\APP$ seems to be much more amenable to proving functional lower bounds than the Shifted Partials measure, and thus we consider here exclusively the $\APP$ measure.

To define the measure, let $k$ and $n_0$ be non-negative integers, let $P$ be a polynomial in $\F[x_1,\ldots,x_n]$, and let $L = \langle\ell_1,\ldots,\ell_n\rangle$ be a tuple of affine forms over the variables $z_1,\ldots,z_{n_0}$
(an affine form is a linear form $\sum_i\alpha_i x_i + a$, for some 
scalars $\alpha_i$ and $a$).
 We denote $\pi_L$ the affine projection that maps each $x_i$ to $\ell_i$. Now define
\[
\APP_{k,n_0}(P) := \max_L\dim\left\langle\pi_L\left(\partial^kP\right)\right\rangle.
\]

The first key lemma in \cite{AGK0T23} is the following structural lemma about the space of partial derivatives. It shows that the space can be realized as suitable shifts of partial derivatives of smaller arities.

\begin{lemma}[\cite{AGK0T23}]\label{lemma:structure-AGKOT}
    Let $n$ and $t$ be positive integers and $Q_1,\ldots,Q_t$ be non-constant, homogeneous polynomials in $\F[x_1,\ldots,x_n]$ with degrees $d_1,\ldots,d_t$, respectively. Let $d := \deg(Q_1\cdots Q_t) = \sum_{i = 1}^t d_i$ and $k < d$ be a non-negative integer. Then
    \[
    \left\langle\partial^k\left(Q_1\cdots Q_t\right)\right\rangle\subseteq \sum_{\substack{S\subseteq [t], k_0,\ell_0\\ k_0 + \frac{k}{d - k}\cdot\ell_0\leq k - \residue_k(d_1,\ldots,d_t)}}\left\langle \xbar^{\ell_0}\cdot\partial^{k_0}\left(\prod_{i\in S} Q_i\right)\right\rangle.
    \]
\end{lemma}

The lemma above is then used to prove the following upper bound for the $\APP$ measure.

\begin{lemma}[\cite{AGK0T23}]\label{lemma:APP-upperbound-on-product-of-homogeneous}
    Let $Q = Q_1\cdots Q_t$ be a homogeneous polynomial in $\F[x_1,\ldots,x_n]$ of degree $d = d_1 + \ldots d_t \geq 1$, where $Q_i$ is homogeneous and $d_i := \deg(Q_i)$ for all $i\in [t]$. Then for any non-negative integers $k < d$ and $n_0\leq n$,
    \[
    \APP_{k,n_0}(Q) \leq 2^t\cdot d^2\cdot\max_{\substack{k_0,\ell_0\geq 0\\ k_0 + \frac{k}{d-k}\cdot\ell_0\leq k - \residue_k(d_1,\ldots,d_t)}} M(n,k_0)\cdot M_\leq(n_0,\ell_0).
    \]
\end{lemma}
There is a minor difference in the lemma as stated above and as stated in \cite{AGK0T23}, and this is the term $M_\leq(n_0,\ell_0)$. In \cite{AGK0T23}, the authors consider only sequences $L = \langle\ell_1,\ldots,\ell_n\rangle$ of linear polynomials in the projections rather than affine ones, and thus they can replace $M_\leq(n_0,\ell_0)$ by $M(n_0,\ell_0)$. However with affine polynomials their proof gives the bound above.

As we are interested in functional lower bounds for low-individual degree polynomials, we need the following modification of \cref{lemma:structure-AGKOT}. A very important ingredient in the lemma is the fact that the individual-degree bound $\delta$ only appears in the term $\residue_k(d'_1,\ldots,d'_t)/\delta$, and thus it only decreases the influence of the residue-term. 

\begin{lemma}
    Let $n$, $t$, $d$ and $\delta$ be positive integers and let $Q_1,\ldots,Q_t$ be non-constant, homogeneous polynomials in $\F[x_1,\ldots,x_n]$ with degrees $d'_1,\ldots,d'_t$, respectively, so that $d' := \deg(Q_1,\ldots,Q_t) = \sum_{i = 1}^t d'_i$ is between $d$ and $\delta\cdot d$, and let $k < d$ be a non-negative integer. Then 
    \[
    \left\langle\partial^k\left(Q_1\cdots Q_t\right)\right\rangle\subseteq \sum_{\substack{S\subseteq [t], k_0,\ell_0\\ k_0 + \frac{k}{d - k}\cdot\ell_0\leq k - \residue_k(d'_1,\ldots,d'_t)/\delta}}\left\langle \vx^{\ell_0}\cdot\partial^{k_0}\left(\prod_{i\in S} Q_i\right)\right\rangle.
    \]
\end{lemma}

\begin{proof}[Proof sketch.]
    Our proof is in essence the same as that of \cref{lemma:structure-AGKOT} in \cite{AGK0T23}. We sketch the argument for completeness.

    As in the proof of \cref{lemma:structure-AGKOT}, denote by $\mathcal{V}$ the set
    \[
    \sum_{\substack{S\subseteq [t], k_0,\ell_0\\ k_0 + \frac{k}{d - k}\cdot\ell_0\leq k - \residue_k(d'_1,\ldots,d'_t)/\delta}}\left\langle \vx^{\ell_0}\cdot\partial^{k_0}\left(\prod_{i\in S} Q_i\right)\right\rangle.
    \]
    Let $\mu\colon [k]\rightarrow x$ be an arbitrary total function. We want to show that $\partial_{\mu([k])}\in \mathcal{V}$. Let $S\subseteq [t]$ and denote by $\bar{S}$ the complement of $S$ relative to $[t]$, i.e. the set $[t]\setminus S$. Let $\tilde{\kappa}\colon\bar{S}\rightarrow 2^{[k]}$ be such that $|\tilde{\kappa}_i| > \frac{k}{d'}\cdot d'_i$ for all $i\in\bar{S}$. With this bound we have that
    \[
    |\tilde{\kappa}_i| - \frac{k}{\delta\cdot d}\cdot d_i'\geq \frac{1}{\delta^2}\cdot\left(|\tilde{\kappa}_i| - \frac{k}{d'}\cdot d_i'\right)
    \]
    for any $i\in\bar{S}$. The usefulness of this inequality will be apparent later in the proof. Define a polynomial $R_{S,\tilde{\kappa}}$ as
    \[
        R_{S,\tilde{\kappa}} := \sum_{\substack{\kappa\colon [t]\rightarrow 2^{[k]}\\ \kappa\text{ extends }\tilde{\kappa}\\\bigsqcup_{i\in [t]}\kappa_i = [k]\\ \forall i\in S, |\kappa_i|\leq \frac{k}{d'}\cdot d'_i}}\prod_{i\in [t]}\partial_{\mu(\kappa_i)}Q_i.
    \]
    Now as in the proof of \cref{lemma:structure-AGKOT}, one can show that 
    \[
    \partial_{\mu([k])}\left(\prod_{i\in[t]}Q_i\right) = \sum_{S\subseteq [t]}\sum_{\substack{\tilde{\kappa}\colon\bar{S}\rightarrow 2^{[k]}\\ \forall i\in\bar{S}, |\kappa_i|> \frac{k}{d'}\cdot d'_i}}R_{S,\tilde{\kappa}},
    \]
    and thus to prove the claim it suffices to show that $R_{S,\tilde{\kappa}}\in\mathcal{V}$ for all $S$ and $\tilde{\kappa}$. To prove this, one argues by induction on the size of $S$. If $|S| = 0$, then $R_{S,\tilde{\kappa}} = 0$ and thus $R_{S,\tilde{\kappa}}\in \mathcal{V}$.

    \allowdisplaybreaks

    Suppose then that $S\subseteq [t]$ is non-empty and let $\tilde{\kappa}\colon \bar{S}\rightarrow 2^{[k]}$ be a function so that $|\tilde{\kappa}_i|> \frac{k}{d'}\cdot d'_i$ for any $i\in\bar{S}$. Let $\kappa\colon [t]\rightarrow 2^{[k]}$ be a function extending $\tilde{\kappa}$ so that $\bigsqcup_{i\in [t]}\kappa_i = [k]$ and $|\kappa_i|\leq \frac{k}{d'}\cdot d'_i$ for any $i\in S$. Denote by $P_S$ the set $\bigsqcup_{i\in S}\kappa_i$, and define the polynomial
    \[
    \mathcal{U}_{S,\kappa} := \left(\partial_{\mu(P_S)}\prod_{i\in S}Q_i\right)\cdot \prod_{i\in \bar{S}} \partial_{\kappa_i} Q_i.
    \]
    Again, as in the proof of \cref{lemma:structure-AGKOT}, one shows that
    \[
    R_{S,\tilde{\kappa}} = \mathcal{U}_{S,\kappa} - \sum_{\substack{T\subsetneq S\text{ and }\kappa'\colon S\setminus T\rightarrow 2^{[k]}\\ \forall i\in S\setminus T, |\kappa_i'|> \frac{k}{d'}\cdot d'_i}}R_{T,\tilde{\kappa}\sqcup\kappa'}.
    \]
    By the induction hypothesis, we know that $R_{T,\tilde{\kappa}\sqcup\kappa'}\in\mathcal{V}$ for any $T\subsetneq S$ and $\kappa'$. Hence it suffices to show that $\mathcal{U}_{S,\kappa}\in \mathcal{V}$. By its definition, we have that
    \[
    \mathcal{U}_{S,\kappa}\in\left\langle x^{\ell_0}\cdot\partial^{k_0}\left(\prod_{i\in S}Q_i\right)\right\rangle,
    \]
    where $k_0 := |\mu(P_S)| = |P_S| = \sum_{i\in S} |\kappa_i|$ and $\ell_0 := \sum_{i\in \bar{S}}\deg\left(\partial_{\mu(\kappa_i)}Q_i\right) = \sum_{i\in \bar{S}}\left(d'_i - |\kappa_i|\right)$. Moreover
    
    \begin{align*}
        k - k_0 - \frac{k}{d-k}\cdot\ell_0 & = k - \sum_{i\in S} |\kappa_i| - \frac{k}{d- k}\cdot \sum_{i\in \bar{S}}\left(d'_i - |\kappa_i|\right)\\
        & = \sum_{i\in\bar{S}}|\kappa_i| - \frac{k}{d- k}\cdot\sum_{i\in \bar{S}}\left(d'_i - |\kappa_i|\right)\\
        & = \sum_{i\in\bar{S}}\left(|\kappa_i| - \frac{k}{d- k}\cdot\left(d'_i - |\kappa_i|\right)\right)\\
        & = \sum_{i\in\bar{S}}\frac{\delta\cdot d}{d - k}\cdot\left(|\kappa_i| - \frac{k}{\delta\cdot d}\cdot d_i'\right)\\
        & \geq \sum_{i\in\bar{S}}\frac{\delta\cdot d}{d - k}\cdot\frac{1}{\delta^2}\cdot\left(|\tilde{\kappa}_i| - \frac{k}{d'}\cdot d_i'\right)\\
        & \geq \sum_{i\in\bar{S}}\frac{1}{\delta}\cdot\left(|\kappa_i| - \frac{k}{d'}\cdot d_i'\right)\\
        & = \frac{1}{\delta}\cdot\sum_{i\in\bar{S}}\left(|\kappa_i| - \frac{k}{d'}\cdot d_i'\right)\\
        & = \frac{1}{2\delta}\cdot\sum_{i\in\bar{S}}\left(|\kappa_i| - \frac{k}{d'}\cdot d_i'\right) + \frac{1}{2\delta}\cdot\sum_{i\in S}\left(|\kappa_i| - \frac{k}{d'}\cdot d_i'\right)\\
        &\quad\quad + \frac{1}{2\delta}\cdot\sum_{i\in\bar{S}}\left(|\kappa_i| - \frac{k}{d'}\cdot d_i'\right) -\frac{1}{2\delta}\cdot\sum_{i\in S}\left(|\kappa_i| - \frac{k}{d'}\cdot d_i'\right)\\
        & = \frac{1}{2\delta}\cdot\sum_{i\in [t]}\left(|\kappa_i| - \frac{k}{d'}\cdot d_i'\right) + \frac{1}{2\delta}\cdot\sum_{i\in\bar{S}}\left(|\kappa_i| - \frac{k}{d'}\cdot d_i'\right) -\frac{1}{2\delta}\cdot\sum_{i\in S}\left(|\kappa_i| - \frac{k}{d'}\cdot d_i'\right)\\
        & = \frac{1}{2\delta}\cdot\left(k - \frac{k}{d'}\cdot d'\right) + \frac{1}{2\delta}\cdot\sum_{i\in\bar{S}}\left(|\kappa_i| - \frac{k}{d'}\cdot d_i'\right) -\frac{1}{2\delta}\cdot\sum_{i\in S}\left(|\kappa_i| - \frac{k}{d'}\cdot d_i'\right)\\
        & = \frac{1}{2\delta}\cdot\sum_{i\in [t]}\left| |\kappa_i| - \frac{k}{d'}\cdot d_i'| \right|\\
        & \geq \frac{1}{\delta}\cdot\residue_k(d_1',\ldots,d_t').
    \end{align*}
    Hence we have that $\mathcal{U}_{S,\kappa}\in\mathcal{V}$ as $k_0 + \frac{k}{d-k}\cdot\ell_0\leq k - \frac{1}{\delta}\cdot\residue_k(d_1',\ldots,d_t')$.
\end{proof}

As a consequence we have the following variant of \cref{lemma:APP-upperbound-on-product-of-homogeneous}, which gives an upper bound for the $\APP$ measure that we can use to prove the functional lower bounds.

\begin{corollary}\label{corollary:APP-upper-bound-functional}
    Let $d$ and $\delta$ be non-negative integers, and let $Q = Q_1\cdots Q_t$ be a homogeneous polynomial in $\F[x_1,\ldots,x_n]$ of degree $d' = d'_1 + \ldots d'_t \geq 1$, where $Q_i$ is homogeneous; $d'_i := \deg(Q_i)$ for all $i\in [t]$ and $d\leq d'\leq \delta\cdot d$. Then for any non-negative integers $k < d$ and $n_0\leq n$,
    \[
    \APP_{k,n_0}(Q) \leq 2^t\cdot \delta^2\cdot d^2\cdot\max_{\substack{k_0,\ell_0\geq 0\\ k_0 + \frac{k}{d-k}\cdot\ell_0\leq k - \residue_k(d_1,\ldots,d_t)/\delta}} M(n,k_0)\cdot M_\leq(n_0,\ell_0).
    \]
\end{corollary}

\subsubsection{Low-Depth Homogeneous Formulas Have High Residue}

Secondly \cite{AGK0T23} showed that any small low-depth homogeneous formula has a representation as a small sum of products of homogeneous polynomials, and moreover the residue of the degrees of the polynomials in the products is relatively large for suitably chosen parameters.
\begin{lemma}[\cite{AGK0T23}]\label{lem:residue-lower-bound}
    Suppose $C$ is a homogeneous formula of product-depth $\Delta\geq 1$ computing a homogeneous polynomial in $\F[x_1,\ldots,x_n]$ of degree $d$, where $d^{2^{1-\Delta}} = \omega(1)$. Then there exist homogeneous polynomials $\{Q_{i,j}\}_{i,j}$ in $\F[x_1,\ldots,x_n]$ such that $C = \sum_{i = 1}^s Q_{i,1}\cdots Q_{i,t_i}$ for some $s\leq \mathrm{size}(C)$. Fixing an arbitrary $i\in [s]$, let $t := t_i$ and let $d_j := \deg(Q_{i,j})$ for $j\in [t]$. Then \[
    \residue_k(d_1,\ldots,d_t)\geq\Omega\left(d^{2^{1 - \Delta}}\right),\] where $k := \left\lfloor\frac{\alpha\cdot d}{1 + \alpha}\right\rfloor$, $\alpha := \sum_{\nu = 0}^{\Delta - 1}\frac{(-1)^{\nu}}{\tau^{2^\nu - 1}}$ and $\tau := \left\lfloor d^{2^{1 - \Delta}}\right\rfloor$.
\end{lemma}

For us the following variant of the previous lemma will be important.

\begin{lemma}\label{cor:residue-in-bounded-individual-degree}
    Let $\delta\geq 1$ be a positive integer, and suppose $C$ is a homogeneous formula of product-depth $\Delta\geq 1$ computing a homogeneous polynomial in $\F[x_1,\ldots,x_n]$ of degree $d'\in [d,\delta\cdot d]$, where $d^{2^{1-\Delta}} = \omega(1)$. Then there exist homogeneous polynomials $\{Q_{i,j}\}_{i,j}$ in $\F[x_1,\ldots,x_n]$ such that $C = \sum_{i = 1}^s Q_{i,1}\cdots Q_{i,t_i}$ for some $s\leq \mathrm{size}(C)$. Fixing an arbitrary $i\in [s]$, let $t := t_i$ and let $d'_j := \deg(Q_{i,j})$ for $j\in [t]$. Then \[\residue_k(d'_1,\ldots,d'_t)\geq\Omega\left(\frac{d^{2^{1 - \Delta}}}{\delta}\right),\] where $k := \left\lfloor\frac{\alpha\cdot d}{1 + \alpha}\right\rfloor$, $\alpha := \sum_{\nu = 0}^{\Delta - 1}\frac{(-1)^{nu}}{\tau^{2^\nu - 1}}$ and $\tau := \left\lfloor d^{2^{1 - \Delta}}\right\rfloor$.
\end{lemma}

\begin{proof}[Proof sketch:]
    First we note that by an argument identical to the proof of \cref{lem:residue-lower-bound} in \cite{AGK0T23} one can find an appropriate decomposition of $C$ as a sum of products of homogeneous polynomials, and show that 
    \[
    \frac{1}{2}\min_{k_1,\ldots,k_t\in \integer}\sum_{j\in [t]}\left|k_j - \frac{k}{d}\cdot d_j'\right|\geq\Omega\left(d^{2^{1-\Delta}}\right).
    \]
    As $d\leq d'\leq \delta\cdot d$ we further have that
    \begin{align*}
        \residue_k(d_1',\ldots,d_t') & = \frac{1}{2}\min_{k_1,\ldots,k_t\in \integer}\sum_{j\in [t]}\left|k_j - \frac{k}{d'}\cdot d_j'\right|\\
        & = \frac{1}{2}\sum_{j\in [t]}\min\left\{\left\{\frac{k}{d'}d'_j\right\},1 - \left\{\frac{k}{d'}d'_j\right\}\right\}\\
        & = \frac{1}{2}\sum_{j\in [t]}\min\left\{\left\{\frac{d}{d'}\cdot\frac{k}{d}\cdot d'_j\right\},1 - \left\{\frac{d}{d'}\cdot\frac{k}{d}\cdot d'_j\right\}\right\}\\
        & = \frac{1}{2}\cdot\frac{d}{d'}\sum_{j\in [t]}\min\left\{\left\{\frac{k}{d}\cdot d'_j\right\},\frac{d'}{d} - \left\{\frac{k}{d}\cdot d'_j\right\}\right\}\\
        & \geq \frac{1}{2}\cdot\frac{d}{d'}\sum_{j\in [t]}\min\left\{\left\{\frac{k}{d}\cdot d'_j\right\},1 - \left\{\frac{k}{d}\cdot d'_j\right\}\right\}\\
        & = \frac{1}{2}\cdot\frac{d}{d'}\min_{k_1,\ldots,k_t\in \integer}\sum_{j\in [t]}\left|k_j - \frac{k}{d}\cdot d_j'\right|\geq\Omega\left(\frac{d^{2^{1 - \Delta}}}{\delta}\right).
     \end{align*}
    
\end{proof}

\subsubsection{High Residue Implies Lower Bounds}

Finally, \cite{AGK0T23} show that high residue together with $\APP$ lower bounds imply lower bounds for homogeneous formulas.

\begin{lemma}[\cite{AGK0T23}]\label{lem:high-residue-implies-lower-bounds}
    Let $P = \sum_{i = 1}^s Q_{i,1}\cdots Q_{i,t_i}$ be a homogeneous polynomial in $\F[x_1,\ldots,x_n]$ of degree $d$, where $Q_{i,j}$ are homogeneous and 
    \[\APP_{k,n_0}(P)\geq 2^{-O(d)}\cdot M(n,k)\]
    for some $1 < k < d/2$ and $n_0\leq n$ such that $n_0\approx 2(d-k)\cdot \left(\frac{n}{k}\right)^{\frac{k}{d-k}}$. If there is some $\gamma > 0$ so that for all $i\in [s]$,
    \[
    \residue_k(\deg(Q_{i,1}),\ldots,\deg(Q_{i,t_i}))\geq\gamma,
    \]
    then $s\geq 2^{-O(d)}\cdot\left(\frac{n}{d}\right)^\gamma$.
\end{lemma}

Again, for our purposes we need a small modification of this result.

\begin{lemma}\label{lemma:residue-implies-functional-lower-bounds}
    Let $d,\delta$ be positive integers. Let $P = \sum_{i = 1}^s Q_{i,1}\cdots Q_{i,t_i}$ be a homogeneous polynomial in $\F[x_1,\ldots,x_n]$ of degree $d'$, where $Q_{i,j}$ are homogeneous, $d\leq d'\leq \delta\cdot d$ and
    \[\APP_{k,n_0}(P)\geq 2^{-O(d)}\cdot M(n,k)\]
    for some $1 < k < d/2$ and $n_0\leq n$ such that $n_0\approx 2(d-k)\cdot \left(\frac{n}{k}\right)^{\frac{k}{d-k}}$. If there is some $\gamma > 0$ so that for all $i\in [s]$,
    \[
    \residue_k(\deg(Q_{i,1}),\ldots,\deg(Q_{i,t_i}))\geq\gamma,
    \]
    then $s\geq 2^{-O(d)}\cdot\delta^{-2}\cdot\left(\frac{n}{d}\right)^{\gamma/\delta}$.
\end{lemma}

\begin{proof}[Proof sketch.]
The proof is again like the proof of \cref{lem:high-residue-implies-lower-bounds}, and we  sketch the argument for completeness.

    By \cref{corollary:APP-upper-bound-functional} and the sub-additivity of $\APP$ we have that
    \[
    \APP_{k,n_0}(P)\leq \sum_{i = 1}^s \APP_{k,n_0}\left(Q_{i,1}\cdots Q_{i,t_i}\right)\leq s\cdot 2^t\cdot \delta^2\cdot d^2\cdot\max_{\substack{k_0,\ell_0\geq 0\\ k_0 + \frac{k}{d-k}\cdot\ell_0\leq k - \gamma/\delta}} M(n,k_0)\cdot M_\leq(n_0,\ell_0).
    \]
On the other hand, by assumption, $\APP_{k,n_0}(P)\geq 2^{-O(d)}\cdot M(n,k)$. These together yield two integers $k_0,\ell_0\geq 0$ satisfying
\begin{equation}\label{eq:bound}
    k_0 + \frac{k}{d - k}\cdot\ell_0\leq k - \gamma/\delta
\end{equation}
and 
\[
    s \geq 2^{-O(d)}\cdot 2^{-t}\cdot \delta^{-2}\cdot d^{-2}\cdot \frac{M(n,k)}{M(n,k_0)\cdot M_\leq(n_0,\ell_0)}.
\]
As in the proof of \cref{lem:high-residue-implies-lower-bounds} one shows using \cref{lemma:bounds-for-M} that this yields
\[
s\geq 2^{-O(d)}\cdot \delta^{-2}\cdot\frac{\left(n/d\right)^{k - k_0 - \frac{k}{d - k}\cdot\ell_0}}{\left(d/k_0\right)^{k_0}\cdot\left(d/\ell_0\right)^{\ell_0}}.
\]
With \ref{eq:bound} and the fact that $x^x\geq e^{-1/e}$ for all $x > 0$ we arrive at
\[
s\geq 2^{-O(d)}\cdot \delta^{-2}\cdot \left(\frac{n}{d}\right)^{\gamma/\delta}.
\]

\end{proof}

\subsection{Knapsack over a Word $w$}

Now we turn to the lower bound proof itself. We first recall from \cite{GHT22} the intermediate hard instance used there to prove the lower bound for multilinear refutations, which is called \emph{knapsack over a word }$w$. We will also use the same intermediate instance to prove our lower bound.

Let $w\in\integer^d$ be a word, and associate with each entry $w_i$ a set of fresh variables $X(w_i)$ of size $2^{|w_i|}$. Denote by $P_w$ the set of indices $i\in [d]$ so that $w_i\geq 0$, and by $N_w$ the set of indices so that $w_i < 0$. For some subset $S\subseteq [d]$ we denote by $w_S$ the sum $\sum_{i\in S}w_i$ and by $w|_S$ the subword of $w$ indexed by the set $S$. We say that a monomial $m$ is set-multilinear on some $w|_S$ if it contains exactly one variable from each of the sets $X(w_i)$ for $i\in S$. 

In the following we fix a useful representation of the variables $X(w_i)$ for all $i\in [d]$. For any $i\in P_w$, we write the variables of $X(w_i)$ in the form $x_\sigma^{(i)}$, where $\sigma$ is a binary string indexed by the integer interval
\[
A_w^{(i)} := \left[\sum_{\substack{i'\in P_w\\ i' < i}}w_{i'} + 1,\sum_{\substack{i'\in P_w\\ i' \leq i}}w_{i'} \right]
\]
Similarly for any $j\in N_w$, we write the variables of $X(w_j)$ in the form $y_\sigma^{(j)}$, where $\sigma$ is a binary string indexed by the set
\[
B_w^{(j)} := \left[\sum_{\substack{j'\in N_w\\ j' < j}}|w_{j'}| + 1,\sum_{\substack{j'\in N_w\\ j' \leq j}}|w_{j'}| \right].
\]
We call the variables of the form $x_\sigma^{(i)}$ the positive variables, or simply $\xbar$ variables, and the variables of the form $y_\sigma^{(j)}$ the negative variables, or simply $\ybar$ variables. For any $S\subseteq P_w$, write $A_w^S$ for the set $\bigcup_{i\in S}A_w^{(i)}$, and define the set $B_w^T$ similarly for any $T\subseteq N_w$.

Each monomial that is set-multilinear on $w|_S$ for some $S\subseteq P_w$ is in one-to-one correspondence with a binary string indexed by the set $A_w^S$. For a set-multilinear monomial $\xbar^\alpha$ on some $w|_S$ we denote by $\sigma(\xbar^\alpha)$ the associated binary string indexed by $A_w^S$ associated with $\xbar^\alpha$.
Similarly any monomial that is set-multilinear on $w|_T$ for some $T\subseteq N_w$ corresponds to a binary string indexed by the set $B_w^T$, and for any $\ybar^\gamma$ that is set-multilinear on some $w|_T$ for $T\subseteq N_w$ we denote by $\sigma(\ybar^\gamma)$ the associated binary string indexed by $B_w^T$.

We say that the word $w$ is $N$-heavy in the case that $|w_{N_w}|\geq |w_{P_w}|$ and $P$-heavy in the case that $|w_{P_w}|\geq|w_{N_w}|$. We call the word balanced if for any $i\in P_w$ there is some $j\in N_w$ so that $A_w^{(i)}\cap B_w^{(j)}\neq\emptyset$ and vice versa.

We define the polynomial $\ks_w$ as follows. For this definition we assume that $w$ is $N$-heavy. Otherwise we switch the roles of the positive and negative variables in the definition. For binary strings $\sigma$ and $\sigma'$ indexed by some sets $A$ and $B$ respectively we write $\sigma\sim\sigma'$ if $\sigma(i) = \sigma'(i)$ for any $i\in A\cap B$.

For $i\in P_w$ and $\sigma\in A_w^{(i)}$ define the polynomial $f_\sigma^{(i)}$ with
\[
f_\sigma^{(i)} := \prod_{\substack{j\in N_w :\\ A_w^{(i)}\cap B_w^{(j)}\neq\emptyset}}\sum_{\substack{\sigma_j\colon B_w^{(j)}\rightarrow\{0,1\}:\\ \sigma_j \sim \sigma}}y_{\sigma_j}^{(j)}
\]
and define the polynomial $\ks_w$ as
\[
\ks_w := \sum_{i\in P_w}\sum_{\sigma\colon A_w^{(i)}\rightarrow\{0,1\}} x_\sigma^{(i)}f_\sigma^{(i)} - \beta
\]
for any $\beta$ so that $\ks_w$ is unsatisfiable over Boolean values.

\subsection{Lower Bounds for the $\APP$ Measures}

This section is devoted to proving a $\APP$ lower bounds for refutations of $\ks_w$. 
We will first prove a lower bound that does not take into account any bounds on the individual degree. The first proof however highlights some key ideas in the proof in a clean manner. After this we discuss how to obtain $\APP$ lower bounds for suitable homogeneous slices of bounded individual degree refutations of $\ks_w$.

Unlike usually in algebraic circuit complexity we are given the refutations of $\ks_w$ only implicitly; we only know something about their functional behaviour. However we still want to prove suitable lower bounds for the dimension of the space spanned by the affine projections of the partial derivatives. The key idea in the proof below is to represent the affine projections of partial derivatives as an alternating sum of suitable partial assignments. This allows us to infer useful information about the structure of the given refutation, and prove the wanted lower bound.

\begin{lemma}\label{lemma:APP-lower-bound}
    Let $h,d$ be positive integers so that $h > 100$, and let $k$ be a parameter in the interval $[\frac{d}{4},\frac{d}{2}]$. Then there are 
    \begin{itemize}
        \item a balanced word $w\in [-h,\ldots,h]^{d}$;
        \item an integer $n_0\leq n$ with $n_0\approx 2(d - k)\left(\frac{n}{k}\right)^{\frac{k}{d-k}}$
    \end{itemize}  
    so that for any polynomial $g$ such that
    \[
    g = \frac{1}{\mathrm{ks}_w}\quad\text{ over Boolean assignments}
    \]
    the following bound holds:
    
        \[\APP_{k,n_0}\left(g\right)\geq 2^{h\cdot k}.\]
\end{lemma}

\begin{proof}
    To construct the word we follow \cite{AGK0T23}. Let $h' = \frac{h\cdot k}{d- k}$, and note that this lies in the interval $[h/3,h]$. The word consists of $k$ many copies of the entry $h$, $k_1 := (d-k)\lceil h'\rceil -kh$ many copies of $-\lfloor h'\rfloor$ and $k_2 := d - k - k_1$ copies of $-\lceil h'\rceil$. The total sum of all the entries is $0$, and thus the word (in any order) is balanced. 
    
    Consider the set of variables $\xbar\cup \ybar$, where $\xbar$ stands for the positive variables of $\ks_w$ and $\ybar$ stands for the negative variables of $\ks_w$. Here $\ks_w$ is defined as if $w$ is $N$-heavy; note that $w$ is in fact both $N$- and $P$-heavy. Take now $n_0 := |\ybar|$. Then we have that $n_0\approx 2(d-k)\cdot\left(\frac{n}{k}\right)^{\frac{k}{d-k}}$ \cite{AGK0T23}.

     Write $g$ in two different ways as a polynomial in $\F[\xbar][\ybar]$ and as a polynomial in $\F[\ybar][\xbar]$. Write $g = \sum_\gamma g_{\gamma}(\xbar) \ybar^\gamma$, where and $g_\gamma$ are polynomials in the positive variables, and write $g = \sum_\alpha g_\alpha(\ybar)\xbar^\alpha$, where $g_\alpha(\ybar)$ are polynomials in the negative variables. \nutan{Mention the required properties e.g. degree etc. for $\gamma, \alpha$.}
     
    Define now a linear map $L$ as follows:
    \[
    L(z) :=
        \begin{cases}
            z,\text{ if }z\in \ybar;\\
            1,\text{ if }z\in \xbar.
        \end{cases}
    \]

    To lower bound the $\APP$ measure with $L$, we consider the set of partial derivatives of $g$ with respect to the set-multilinear monomials over all the positive variables. Now it is easy to see that for a set-multilinear monomial $\xbar^\alpha$ in the positive variables,
    \[
    L\left(\partial_{\alpha}g \right) = \sum_{\substack{\alpha':\\\alpha \subseteq \alpha' }}g_{\alpha'}(\ybar).
    \]
    Hence we need to lower bound the dimension of the space spanned by such sums. First note that in order to lower bound this dimension, it suffices to lower bound the dimension of the multilinearizations of such sums, as multilinearization of the spanning set can only decrease the dimension. Thus we are interested in the following polynomials
    \[
    h_\alpha := \sum_{\substack{\alpha':\\ \alpha \subseteq \alpha' }}\ml\left(g_{\alpha'}(\ybar)\right)
    \]
    for set-multilinear monomials $\xbar^\alpha$ on the positive variables.
    
    We will prove a lower bound on the dimension by showing that the polynomials $h_\alpha$ are linearly independent for all distinct set-multilinear monomials $\xbar^\alpha$. This immediately gives the claimed lower bound as there are $2^{h\cdot k}$ many distinct set-multilinear monomials in the positive variables. We prove the linear independence in a series of claims. First we give a useful representation for the polynomials $h_\alpha$ as sums of reciprocals of knapsack instances. For a monomial $\xbar^\alpha$ we denote by $\tau_{\alpha}$ the Boolean assignment that maps all the variables appearing in $\xbar^\alpha$ to $0$ and all the other $\xbar$-variables to $1$, and leaves $\ybar$ variables untouched.

    \begin{claim}\label{claim:sum-of-knapsacks}
        Let $\xbar^\alpha$ be a set-multilinear monomial in the positive variables. Then
        \[
        h_\alpha = \sum_{\mu\subseteq\alpha}(-1)^{|\mu|}\tau_{\mu}\left(\frac{1}{\ks_w}\right)\quad\text{ over Boolean values.}
        \]
    \end{claim}

    \begin{proof}[Proof of \cref{claim:sum-of-knapsacks}:]
        First note that 
        \[
        h_\alpha = \sum_{\mu\subseteq\alpha}(-1)^{|\mu|}\tau_{\mu}(\ml(g)).
        \]
        This follows from a straightforward argument using the inclusion-exclusion principle. Indeed, note that for any $\mu\subseteq \alpha$ we have that
        \[
        \tau_{\mu}(\ml(g)) = \sum_{\substack{\nu:\\\mu\wedge\nu = \emptyset}}\ml(g_{\nu}(\ybar)).
        \]
        Now if $\alpha'\supseteq\alpha$, then the only $\mu\subseteq\alpha$ so that $\alpha'\wedge\mu = \emptyset$ is $\mu = \emptyset$. Hence the terms $\ml(g_{\alpha'}(\ybar))$ for all the supersets $\alpha'$ of $\alpha$ survive in the expression above. If on the other hand $\alpha'\nsupseteq\alpha$, then there is some $i$ so that $\alpha(i) = 1$, but $\alpha'(i) = 0$. The set of all $\mu\subseteq\alpha$ satisfying $\mu\wedge\alpha' = \emptyset$ forms a non-trivial finite Boolean lattice. It follows by inclusion-exclusion principle that no term of the form $\ml(g_{\alpha'}(\ybar))$ for $\alpha'\nsupseteq\alpha$ survives in the final expression.

        On the other hand by definition
        \[
        \ml(g) = \frac{1}{\ks_w}\quad\text{ over Boolean values,}
        \]
        and hence for any $\mu\subseteq\alpha$
        \[
        \tau_{\mu}(\ml(g)) = \tau_{\mu}\left(\frac{1}{\ks_w}\right)\quad\text{ over Boolean values.}
        \]
        This concludes the proof of \Cref{claim:sum-of-knapsacks}.
    \end{proof}

    For a monomial $\ybar^\alpha$ we denote by $\pi_\alpha$ the partial Boolean mapping that send each variable in $\ybar^\alpha$ to $1$ and any other $\ybar$ variable to $0$, and leaves $\xbar$ variables untouched.

    \begin{claim}\label{claim:main-claim}
        Let $\xbar^\alpha$ be a set-multilinear monomial in the positive variables and let $\ybar^{\gamma}$ be a set-multilinear monomial in the negative variables. Then
        \[
        \pi_{\gamma}(h_\alpha) \neq 0\quad\text{ if and only if }\quad\sigma(\ybar^{\gamma}) = \sigma(\xbar^\alpha).
        \]
    \end{claim}

    \begin{proof}[Proof of \cref{claim:main-claim}:]
        From \cref{claim:sum-of-knapsacks} we know that
        \[
        h_\alpha = \sum_{\mu\subseteq\alpha}(-1)^{|\mu|}\tau_{\mu}\left(\frac{1}{\ks_w}\right)\quad\text{ over Boolean values.}
        \]
        Hence we have the following chain of equalities.
        \begin{align*}
            \pi_{\gamma}(h_\alpha) & = \sum_{\mu\subseteq\alpha}(-1)^{|\mu|}\pi_\gamma\left(\tau_{\mu}\left( \frac{1}{\ks_w}\right)\right)\\
            & = \sum_{\mu\subseteq\alpha}(-1)^{|\mu|}\tau_{\mu}\left(\frac{1}{\pi_\gamma(\ks_w)}\right)\\
            & = \sum_{\mu\subseteq\alpha}(-1)^{|\mu|}\tau_{\mu}\left(\frac{1}{\sum x_\sigma^{(i)} - \beta}\right),
        \end{align*}
        where in the last row the sum in the denominator ranges over those $i$ and $\sigma$ so that $\sigma$ agrees with $\sigma(\ybar^\gamma)$ on the interval $A_w^{(i)}$.

        Let now $f_\gamma$ be the multilinear polynomial so that
        \[
        f_\gamma = \frac{1}{\sum x_\sigma^{(i)} - \beta}\quad\text{ over Boolean values.}
        \]
        From \cite{FSTW21} we know that the leading monomial of $f_\gamma$ is the product of all the variables appearing in the sum. Other monomials in $f_\gamma$ are naturally submonomials of the full-degree monomial.

        Thus we want to analyze the value of  
        \[
        \sum_{\mu\subseteq\alpha}(-1)^{|\mu|}\tau_{\mu}(f_\gamma)
        \]
        for different set-multilinear monomials $\xbar^\alpha$.

        Suppose first that $\sigma(\xbar^\alpha) = \sigma(\ybar^\gamma)$, i.e. $\xbar^\alpha$ is the leading monomial of the polynomial $f_\gamma$. Now $\tau_{\emptyset}(\xbar^\alpha) =  1$, but for any $\mu\neq \emptyset$ we have that $\tau_{\mu}(\xbar^\alpha) = 0$. For any proper submonomial $\xbar^{\mu}$ of $\xbar^\alpha$ we have for all $\mu'\subseteq\alpha$ that $\tau_{\mu'}(\xbar^{\mu})\neq 0$ if and only if $\mu\wedge\mu' = \emptyset$. The set of such $\mu'$'s forms a non-trivial Boolean lattice and thus by the inclusion-exclusion principle
        \[
        \sum_{\substack{\mu'\subseteq\alpha\\\mu\wedge\mu' = \emptyset}}(-1)^{|\mu'|}\tau_{\mu'}\left(\xbar^{\mu}\right) = 0.
        \]
        We have shown that when $\sigma(\xbar^\alpha) = \sigma(\ybar^\gamma)$,
        \[
        \sum_{\mu\subseteq\alpha}(-1)^{|\mu|}\tau_{\mu}(f_\gamma) = a_\alpha,
        \]
        where $a_\alpha$ is the coefficient of $\xbar^\alpha$ in $f_\gamma$.

        Suppose then that $\sigma(\xbar^\alpha)\neq \sigma(\ybar^\gamma)$, and let $\xbar^{\hat{\alpha}}$ be so that $\sigma(\xbar^{\hat{\alpha}}) = \sigma(\ybar^\gamma).$ Now for any $\mu\subseteq\alpha$ and any $\hat{\mu}\subseteq\hat{\alpha}$ we have that $\tau_{\mu}(\xbar^{\hat{\mu}})\neq 0$ if and only if $\mu\wedge\hat{\mu} = \emptyset$. For any fixed $\hat{\mu}\subseteq\hat{\alpha}$ the set of all those $\mu\subseteq\alpha$ that satisfy the latter condition forms again a non-trivial Boolean lattice, and thus by inclusion-exclusion principle we have that
        \[
        \sum_{\substack{\mu\subseteq\alpha\\\mu\wedge\hat{\mu} = \emptyset}}(-1)^{|\mu|}\tau_{\mu}\left(\xbar^{\hat{\mu}}\right) = 0.
        \]
        Above $\hat{\mu}$ was an arbitrary substring of $\hat{\alpha}$ and thus we have that
        \[
        \sum_{\mu\subseteq\alpha}(-1)^{|\mu|}\tau_{\mu}(f_\gamma) = 0.
        \]
        This concludes the proof of \Cref{claim:main-claim}.
\end{proof}

    To finish the proof of \Cref{lemma:APP-lower-bound}, suppose that for some set-multilinear $\xbar^\alpha$ there are $a_{\alpha'}\in\F$ for each $\alpha'\neq\alpha$ so that
    \[
    h_{\alpha} = \sum_{\alpha'\neq \alpha}a_{\alpha'}h_{\alpha'}.
    \]
    Consider then the mapping $\pi_\gamma$, where $\gamma$ is so that $\sigma(\xbar^\alpha) = \sigma(\ybar^\gamma)$. Then by \cref{claim:main-claim} we have that
    \[
    1 = \pi_\gamma(h_\alpha) = \sum_{\alpha'\neq \alpha}a_{\alpha'}\pi_\gamma(h_{\alpha'}) = 0.
    \]
    Hence the polynomials $h_\alpha$ are linearly independent for all distinct set-multilinear $\xbar^\alpha$.
\end{proof}

 \Cref{lemma:APP-lower-bound} demonstrates an $\APP$ lower bound on any refutation of $\ks_w$. In order to derive meaningful circuit lower bounds from these bounds however we need bounds for some low-degree homogeneous polynomial, while  no refutation of $\ks_w$ is homogeneous and low-degree. In the following lemma we show however that assuming a bound on the individual degree of the given refutation, we can infer useful $\APP$ lower bounds for some homogeneous low-degree slice of the refutation.

\begin{lemma}\label{corollary:homogenous-APP-lower-bounds}
    Let $h,d,\delta$ be positive integers so that $h > 100$, and let $k$ be a parameter in $[\frac{d}{4},\frac{d}{2}]$. Then there are
    \begin{itemize}
        \item a balanced word $w\in [-h,\ldots,h]^d$;
        \item an integer $n_0\leq n$ so that $n_0\approx 2(d - k)\left(\frac{n}{k}\right)^{\frac{k}{d-k}}$;
        \item an integer $d'$ between $d$ and $\delta\cdot d$
    \end{itemize}  
    so that for any polynomial $g$ of individual degree at most $\delta$ such that
    \[
    g = \frac{1}{\mathrm{ks}_w}\quad\text{ over Boolean assignments}
    \]
    the following bound holds
        \[\APP_{k,n_0}\left(g_{d'}\right)\geq \frac{2^{hk}}{(\delta - 1)\cdot d + 1},\]
    where $g_{d'}$ denotes the homogeneous $d'$-slice of $g$.
\end{lemma}

\begin{proof}
    We will comment on what needs to be changed in the proof of the previous lemma in order to take into account the bounds on the degrees.

    Consider the slice $\bar{g}$ of $g$ of monomials of degrees between $d$ and $\delta\cdot d$, and write this fragment as $\bar{g} = \sum_\alpha g_{\alpha}(\ybar)\xbar^\alpha$, where now $d\leq |\alpha|\leq \delta\cdot d$.

    For the same choice of partial derivatives and linear function $L$ we again obtain that
    \[
    L(\partial_\alpha\bar{g}) = \sum_{\substack{\alpha':\\\alpha\subseteq\alpha'}}g_{\alpha'},
    \]
    where again the sum runs over those $\alpha'$ satisfying $d\leq |\alpha'|\leq\delta\cdot d$.
    
    Now define again for any set-multilinear $\xbar^\alpha$ the polynomial
    \[\bar{h}_\alpha := \sum_{\substack{\alpha':\\\alpha\subseteq\alpha'}}\ml(g_{\alpha'})\]
    and note as in \cref{claim:sum-of-knapsacks} that
    \[
    \bar{h}_\alpha = \sum_{\mu\subseteq\alpha}\left(-1\right)^{|\mu|}\tau_\mu(\ml(\bar{g})).
    \]
    Then we have that
    \[\pi_\gamma\left(\bar{h}_\alpha\right) = \sum_{\mu\subseteq\alpha}\left(-1\right)^{|\mu|}\tau_\mu\left(\pi_\gamma(\ml(\bar{g}))\right).\]

    Now one can show that the leading monomial of the multilinear polynomial $f_\gamma$ is present in $\pi_\gamma(\ml(\bar{g}))$. This follows from the following claim whose proof can be found in \cite{GHT22}.

            \begin{claim}\label{claim:submonomials}
   For a set-multilinear $\ybar^\gamma$ the leading monomial of the polynomial
    \[
        \sum_{\substack{\gamma'\\ \supp(\gamma') = \gamma}}\mathrm{ml}(g_{\gamma'})
    \]
    is the set-multilinear $\xbar^\alpha$ so that $\sigma(\xbar^\alpha) = \sigma(\ybar^\gamma).$
        \end{claim}

    Now, again by a similar argument one can show that for any set-multilinear $\xbar^\alpha$ and $\ybar^\gamma$
    \[
    \pi_\gamma(\bar{h}_\alpha)\neq 0\quad\text{ if and only if }\quad\sigma(\xbar^\alpha) = \sigma(\ybar^\gamma),
    \]
    and hence the polynomials $\bar{h}_\alpha$ are also linearly independent.
    
    Now, by subadditivity of the $\APP$ measure, we have that
    \[
    2^{hk}\leq \APP_{k,n_0}(\bar{g})\leq \sum_{d' = d}^{\delta\cdot d}\APP_{k,n_0}\left(g_{d'}\right).
    \]
    Hence there is some $d'$ between $d$ and $\delta\cdot d$ so that
    \[
    \APP_{k,n_0}\left(g_{d'}\right)\geq \frac{2^{h\cdot k}}{(\delta - 1)\cdot d + 1}.
    \]
\end{proof}

\subsection{Constant-Depth Lower Bounds for the Lifted Knapsack}

\begin{lemma}\label{lem:homogeneous-lower-bounds}
    Let $h,d$ and $\delta$ be positive integers so that $h > 100$. There is a balanced word $w\in [-h,h]^d$ and an integer $d'$ between $d$ and $\delta\cdot d$ so that for any polynomial $g$ of individual degree at most $\delta$ satisfying
    \[
    g = \frac{1}{\ks_w}\quad\text{ over Boolean values}
    \]
    any homogeneous formula of product-depth $\Delta$ computing the $d'$-slice of $g$ has size at least
    \[
    2^{-O(d)}\cdot\left(\frac{n}{d}\right)^{\Omega\left(\frac{d^{2^{1 - \Delta}}}{\delta^2}\right)}.
    \]
\end{lemma}

\begin{proof}
    Let $k$ be defined as in \cref{lem:residue-lower-bound}. Then $k\in \left[\frac{d}{4},\frac{d}{2}\right]$. Let $w$, $n_0$ and $d'$ be as in \cref{corollary:homogenous-APP-lower-bounds}, and let $F$ be a homogeneous formula of product-depth $\Delta$ computing the $d'$-slice of $g$. 

    We can assume that $\delta < d$ as otherwise the claim is trivial. On the other hand
    \[
        k\cdot 2^h\leq n\leq d\cdot 2^h,\text{ and so }2^h\approx\left(\frac{n}{k}\right),
    \]
    and hence, by \cref{lemma:bounds-for-M},
    \[\frac{2^{h\cdot k}}{(\delta - 1)\cdot d + 1}\geq 2^{-O(d)}\cdot\left(\frac{n}{k}\right)^k\geq 2^{-O(d)}\cdot M(n,k).\] 

By \cref{cor:residue-in-bounded-individual-degree} there exists homogeneous polynomials $\{Q_{i,j}\}_{i,j}$ so that 
    \[
    F = \sum_{i\in [s]}Q_{i,1}\cdots Q_{i,t_i}
    \]
    for some $s\leq\mathrm{size}(F)$ and 
\[
\residue_k\left(\deg\left(Q_{i,1}\right),\ldots,\deg\left(Q_{i,t_i}\right)\right)\geq\Omega\left(\frac{d^{2^{1 - \Delta}}}{\delta}\right)
\]
    for all $i\in [s]$. 
    Now by \cref{lemma:residue-implies-functional-lower-bounds} and \cref{corollary:homogenous-APP-lower-bounds} we have that
    \[
    s \geq 2^{-O(d)}\cdot\left(\frac{n}{d}\right)^{\Omega\left(\frac{d^{2^{1 - \Delta}}}{\delta^2}\right)}.
    \]

\end{proof}

Finally, we are ready to prove our main result of this section

\begin{proof}[Proof of \cref{thm:constant-depth-lower-bounds}:]
    Let $C$ be a circuit of size $s$ and product-depth $\Delta$ computing $g$. Let $h := \lfloor\log n / 2\rfloor$, let $d := \lfloor \log n\rfloor$, and note that $d\cdot 2^{h} < n$ for large enough $n$. We can again assume that $\delta < d$ as otherwise the claim is trivial. 
    
    Let $w\in [-h,h]^d$ be defined as in \cref{lemma:APP-lower-bound}. The polynomial $\ks_w$ is of degree at most $4$ as any $A_w^{(i)}$ overlaps with at most $3$ distinct $B_w^{(j)}$. This is due to the fact that $h'$ as defined in \cref{lemma:APP-lower-bound} lives in the interval $[h/3,h]$. Hence there is some partial assignment to the variables $z_{ijk\ell}$ and $x_i$ that maps $\sum_{i,j,k,\ell}z_{ijk\ell}x_ix_jx_kx_\ell - \beta$ to $\ks_w$ up to renaming of variables. By applying this partial assignment to $C$ we obtain a circuit $C'$ of size at most $s$ and product-depth $\Delta$ computing a polynomial $g'$ of individual degree at most $\delta$ that equals $1/\ks_w$ over Boolean values. We can expand this circuit to a formula $F$ of the same product-depth and size $s^{O(\Delta)}$.

Using the homogenization transformation of \cite{LST21} we can compute the $d'$-slice of $F$ with a homogeneous formula of product depth $2\Delta$ and size $s^{O(\Delta)}\cdot 2^{O(\sqrt{d'})} = s^{O(\Delta)}\cdot 2^{O(d)}$. By \cref{lem:homogeneous-lower-bounds} we have that
    \[
    s^{O(\Delta)}\cdot 2^{O(d)}\geq 2^{-O(d)}\cdot\left(\frac{n}{d}\right)^{\Omega\left(\frac{d^{2^{1 - 2\Delta}}}{\delta^2}\right)}.
    \]
    With the chosen value of $d$ this proves our lower bound.
\end{proof}

\subsection{Relative Strength of Low Individual Degree and Low Depth Refutations}

To finish this section we discuss the strength of the proofs we have just considered, i.e., low depth and low individual degree IPS refutations. We demonstrate the strength by giving upper bounds for few standard benchmark formulas in proof complexity, Tseitin formulas and two variants of the pigeonhole principle. Afterwards we discuss some weaknesses caused by the individual degree restriction.

It was already observed in \cite{GHT22} that Tseitin formulas have small multilinear constant-depth IPS refutations. This simple observation is due to the fact that Tseitin formulas have refutations with small number of monomials in the weaker Nullstellensatz proof system when the Boolean values are represented in the \emph{Fourier} basis with $-1$ representing true and $1$ representing false \cite{Gri98}. Constant depth multilinear IPS can easily simulate this small Nullstellensatz refutation in the usual $\{0,1\}$ basis by small formulas computing the appropriate change of basis.

We turn now to two variants of the pigeonhole principle: the functional pigeonhole principle $\mathrm{PHP}^{n+1}_n$, and the graph pigeonhole principle over bipartite graphs on $n+1$ pigeons and $n$ holes with bounded left-degree, where each pigeon has only a limited number of holes available to fly to. We prove the upper bounds for the CNF encodings of these principles, but both proofs rely on a reduction to the polynomial representation of the functional pigeonhole principle used by Razborov in \cite{Razb98}, where the pigeon axioms are represented by the polynomial constraints
\(x_{i1} + \cdots +x_{in} = 1\), for all $i\in [n+1]$.

We prove first the upper bound for this polynomial encoding. This proof borrows from the proof given by Grigoriev and Hirsch in \cite{GH03} and Raz and Tzameret~\cite{RT06} with small modifications. Recall that along with the pigeon axioms we have the hole axioms $x_{ik}x_{jk} = 0$ for any $i,j\in [n+1]$ and $k\in [n]$. Define auxiliary terms $y_k$ by $y_k := \sum_{i\in [n+1]}x_{ik}$ for every $k\in [n]$. Now $y_k$ are Boolean as one can easily derive $y_k^2 - y_k$ from the hole axioms and the Boolean axioms $x_{ik}^2 - x_{ik}$. On the other hand by adding up all the pigeon axioms we arrive at the expression
\[
y_1 + \cdots+ y_n = n + 1.
\]
This is an unsatisfiable instance of the standard  knapsack formula. The unique multilinear $p$ so that 
\[
p\cdot\left(y_1 + \cdots + y_n - (n+1)\right) \equiv 1 \mod \vy^2 - \vy
\]
is symmetric, and thus can be written as weighted sum of elementary symmetric polynomials. Since these have small multilinear constant depth formulas by the standard  Ben-Or result (see Shpilka and Wigderson \cite[Theorem 5.1]{SW01}), the polynomial $p$ has also a small multilinear constant depth representation. Forbes \emph{et al.}~gave an explicit representation of the polynomial $p$ in \cite{FSTW21}. Moreover it is easy to verify that the certificate for the equivalence modulo the Boolean axioms above is also computable by small multilinear constant depth formula. That is, there are small multilinear constant depth formulas $p,p_1,\ldots,p_n$ so that
\[
1 = p\cdot\left(y_1 + \cdots + y_n - (n+1)\right) + \sum_{k\in [n]}p_k\cdot(y_k^2 - y_k).
\]
Substituting $\sum_{i\in [n+1]}x_{ik}$ for $y_k$ in the expression above and using the small derivations of $y_k^2 - y_k$ from the hole axioms and the Boolean axioms, we obtain a small multilinear constant depth refutation of the polynomial encoding of the functional pigeonhole principle.

Next we prove the upper bounds for the CNF encodings of the mentioned variants of pigeonhole principle. We consider the usual translation of CNFs to polynomial constraints, where, for example, the clause $x\vee \vy\vee z$ is translated to the polynomial constraint $(1 - x)y(1 - z) = 0$. Any Boolean assignment that satisfies the clause satisfies also the polynomial constraint and vice versa.

\begin{lemma}
    There is a constant-depth IPS refutation of $\mathrm{FPHP}^{n+1}_n$ of size polynomial in $n$ and individual degree $2$.
\end{lemma}

\begin{proof}
    Recall that $\mathrm{FPHP}^{n+1}_n$ consists of the following clauses 
    \begin{itemize}
        \item $\bigvee_{k\in [n]}x_{ik} $ for any $i\in [n+1]$; \hfill (pigeon axioms)
        \item $\vx_{ik}\vee\vx_{jk}$ for any $i\neq j\in [n+1]$ and $k\in [n]$; \hfill (hole axioms)
        \item $\vx_{ik}\vee\vx_{i\ell}$ for any $i\in [n+1]$ and $k\neq \ell \in [n]$. \hfill (functionality axioms)
    \end{itemize}
    From the pigeon axioms and functionality axioms one derives easily the polynomials
    \[
    x_{i1} + \cdots + x_{in} - 1\text{ for all }i\in [n+1].
    \]
    Combining these derivations with the multilinear refutations for the polynomial encoding discussed above yields refutations of individual degree $2$.
\end{proof}

\begin{lemma}
    Let $G = ([n+1],[n],E)$ be a bipartite graph with left-degree at most $\delta$. Then there is a constant depth IPS refutation of $\mathrm{PHP}_G$ of size polynomial in $n$ and individual degree $O(\delta).$
\end{lemma}

\begin{proof}
    For this result we need to use slightly modified reduction to a suitable polynomial encoding of the functional pigeonhole principle. Our argument however follows closely the one presented by Grigoriev and Hirsch in \cite{GH03}.

    For $i\in [n+1]$ denote by $N(i)$ the set of neighbours of $i$. Recall that $\mathrm{PHP}_G$ consists of the following clauses
    \begin{itemize}
        \item $\bigvee_{k\in N(i)} x_{ik}$ for every $i\in [n+1]$; \hfill (pigeon axioms)
        \item $\vx_{ik}\vee\vx_{jk}$ for every $i\neq j\in [n+1]$ and $k\in N(i)\cap N(j)$. \hfill (hole axioms).
    \end{itemize}
    For any $i\in [n+1]$ and $k\in N(i)$ we consider an auxiliary polynomials $q_{ik}$ defined as
    \[
    q_{ik} := x_{ik}\cdot \prod_{\substack{\ell < k\\ \ell\in N(i)}}(1 - x_{i\ell}).
    \]
    Note that $q_{ik}$ is the polynomial translation of $\bigvee x_{i\ell}\vee\vx_{ik}$ and hence satisfies $q_{ik}^2 = q_{ik}$ over Boolean values. Now the pigeon axioms translate into the form
    \(1 - \sum_{k\in N(i)}q_{ik}\) 
and the products $q_{ik}q_{jk}$ for any $i\neq j\in [n+1]$ and $k\in N(i)\cap N(j)$ are easily derivable from the hole axioms. 

    Consider now the polynomials $r_k$ defined as 
    \[r_k := \sum_{\substack{i\in [n+1]\\ i\in N(k)}}q_{ik}.\]
    These polynomials are again Boolean valued as one can derive $r_k^2 - r_k$ easily from the pigeon axioms and the product $q_{ik}q_{jk}$. By adding up all the polynomial $r_k$ we end up with an unsatisfiable instance of the basic knapsack formula
    \[r_1 + \cdots + r_n - (n+1).\]
    As observed before, this formula has small multilinear constant depth refutations. Now substituting back in this refutation the definitions of $r_k$ and subsequently those of $q_{ik}$ we obtain a small constant depth refutation of $\mathrm{PHP}_G$ of individual degree $O(\delta).$
\end{proof}

The previous lemma demonstrate an unfortunate weakness of the restriction to bounded individual degree. The restriction has the effect that our proof system is not closed under substitutions --- substituting multilinear formulas to a multilinear formula can yield a polynomial of very high individual degree. This makes it also hard to simulate rule-based propositional proof systems. The straightforward  simulation quickly blows up the individual degree. Currently it is unclear which traditional propositional proof systems can be simulated by low depth and low individual degree IPS refutations, if any.

\section{Barriers: Hardness for Boolean Instances}
\label{sec:barriers}


\iddo{There is another barrier: the second part 
in the Definition of Functional lower bound method 
doesn't work; eg., in constant-depth IPS with constant i.d. we can't carry out step 2.
Yet another barrier: finite fields, Fermat little thm.}

Recall that an instance consisting of a set 
of polynomials $\{f_i(\vx)=0\}_i$, for $f_i(\vx)\in\F[\vx]$, 
is said to be \emph{Boolean} whenever 
$f_i(\vx)\in\bits$ for $\vx\in\bits^{|\vx|}$.
Here we show that for sufficiently strong proof systems
the functional lower bound method cannot lead to lower bounds
for Boolean instances (e.g., CNF formulas). And hence cannot 
settle major open problems in proof complexity about
lower bounds against constant depth Frege proofs with counting
modulo $p$ gates (\ACZ$[p]$-Frege), as well as Threshold 
Logic system (\TCZ-Frege). \iddolater{put this in intro and abstract. Note that AC0 Frege is simulated by constant-depth IPS by GP18 hence it is interesting!}

Let \F\ be field and $\cT\in\F[\vx]$ 
be a set of polynomials (that will stand for the set of polynomials we prove or derive, possibly encoded in different way; e.g., 
as Boolean formulas corresponding to their arithmetization).
We say that a proof system $P$ is a (sound and complete) \emph{proof
system for the set of polynomial equations from $\cT$} if given
a set of polynomial equations $\{f_i(\vx)=0\}_i$~ where $f_i\in\cT$,
there is a $P$-proof (equivalently, a $P$-derivation) of $g(\vx)=0$
with $g\in \cT$ , iff $g(\vx)=0$ is semantically implied by the
equations, where semantic implication means: 
$$
\forall\va\in\F^n
\left(
    \left(
        \bigwedge_i ( f_i(\va)=0)
    \right)
    \Rightarrow
    g(\va)=0
\right).
$$

The following definition and its corollary is a more general setting for the
functional lower bound method from the one in  
\Cref{def:single-functional-lower-bound-method}.

\begin{definition}[Instance admitting a functional lower bound]
\label{def:general-functional-lower-bound-method}
Let $\cC\subseteq\pring$ be a circuit class closed under 
(partial) field-element  assignments
(which stands as the class of ``polynomials with 
small circuits"). Let $\cF:=\{f_i(\vx)=0\}_i$ be a collection of polynomial equations
in $\cC$, such that the  system $\cF$ and $\baxioms$
is unsatisfiable (i.e., does not have a common root).
We say that \demph{$\cF$ admits a functional lower bound against $\cC$-{\rm \lbIPS}}
if the following two hold. 
\begin{enumerate}
\item \label{it:function-lower-bound:1}
\uline{Circuit lower bound for} $\frac{1}{f(\vx)}$: Let $f(\vx)\in \cC$ be a polynomial, 
where the system $f(\vx)$ and $\baxioms$ 
is unsatisfiable. Suppose that   $g\not\in\cC$
for all $g\in\pring$ with 
\[
g(\vx)=\frac{1}{f(\vx)}, \quad \forall\vx\in\bits^n\,.
\]
By \Cref{def:single-functional-lower-bound-method}
this means that $f(\vx)$ and $\baxioms$ do not 
have $\mathcal{C}$-\lbIPS\ refutations, and 
moreover, if $\cC$ is a set of multilinear polynomials, then $f(\vx)$ and $\baxioms$
do not have  $\mathcal{C}$-\IPS\ refutations. 

\item \label{it:function-lower-bound:2}
\uline{$\cF$ is efficiently derivable from $f(\vx)$}:
Suppose that there is a $\cC$-{\rm \lbIPS}\ proof
of $\cF$ from $f(\vx)=0$ (and \baxioms) 
(it is possible that $\cF$ is equal to $\{f(\vx)\}$). 
\end{enumerate}
\end{definition}
\iddolater{mix between f=0 and f as axioms. Comment that lower bound here is no proof in C-IPS.}

\begin{observation}[General Functional Lower Bound Method]
\label{obs:general-functional-lower-bound-method}
Under the assumptions  in \Cref{def:general-functional-lower-bound-method},
if $\cF$ admits a functional lower bound against $\cC$-\lbIPS, 
then there are no $\cC$-{\rm \lbIPS}\ refutations of $\cF$.
\end{observation}

This observation follows, because by
\Cref{it:function-lower-bound:1} there is no 
$\cC$-{\rm \lbIPS}\ refutations of $\cF$ 
(otherwise, starting from $f(\vx)=0$, by 
\Cref{it:function-lower-bound:2} we can derive $\cF$ and
refute $f(\vx)=0$).

Let $P$ be a proof system for the language
of polynomial equations in $\cC\subseteq\pring$, that is not 
necessarily equal to $\cC$-\lbIPS.
Let $\cF:=\{f_i(\vx)=0\}_i$ be polynomial equations 
in $\cC$ and suppose 
we wish to establish a lower bound for
$\cF$ against $P$ using the Functional Lower Bound Method. For this purpose, we take a proof 
system $\cC'$-\lbIPS\ that simulates $P$
 and such that $\cC' \supseteq \cC$, 
and  prove a lower bound for $f(\vx)=0$ against
$\cC'$-\lbIPS, with $\cF $ semantically implied by $f(\vx)=0$
and \baxioms\ (for $f(\vx)\in\cC'$).

We say that a circuit class $\cC'\subseteq\pring$ is  \emph{closed under polynomial many sum of products} 
 if $h_i, g_i\in\cC'\cap \F[x_1,\dots,x_n]$, 
for $i\in I$ and $|I|=\poly(n)$, then
$\sum_{i\in I} h_i\cd g_i \in \cC'$.

An \emph{arithmetization scheme} denoted $tr$ is any translation between Boolean formulas (or circuits) to polynomials, that maintains their functional behaviour over the Boolean cube. In other words, an arithmetization scheme is a mapping $\Gamma:\textsf{ Boolean formulas }\rightarrow \F[\vx]$, such that for all $\valpha\in\bits^n$, $\Gamma(\phi)(\valpha) = 0 $ iff $\phi(\valpha)=\textsc{true}$ (we can also flip ``=0" to ``=1" in $\Gamma(\phi)(\valpha) = 0 $ or use other pair of values in $\F$ to represent \textsc{true} and \textsc{false}). The standard one is the following: define $tr$  inductively on the formula structure 
by $tr(x_i):=1-x_i$, $t(A\land B)=1-(1-tr(A))\cd (1-tr(B))$,
and $tr(A\lor B)=tr(A)\cd tr(B)$ and $tr(MOD_p(A_1,\dots,A_r))
 = (\sum_{i=1}^r tr(A_i))^{p-1}$).

\begin{definition}[Sufficiently strong proof system]
\label{def:suf-strong-ps}
Let $P$ be a proof system for a set of polynomial 
equations in $\cC$, for some $\cC\in\pring$,
that can be simulated by some $\cC'$-\lbIPS\ with 
$\cC'\supseteq \cC$ that is closed under partial assignments and sum of products.
We say that $P$ is a \demph{sufficiently strong proof system}
if there exists an arithmetization scheme $tr(\cd)$ such that 
for every set of Boolean polynomial equations $\{f_i(\vx)=0\}_i$,
there is a $P$-proof (equivalently, a $P$-derivation) of the arithmetization of  $\bigwedge_i f_i$
of size $\poly(\sum_i|f_i|)$ (where $|f_i|$ denotes the 
size of the circuit computing $f_i$ in $\cC$).
\end{definition}


Examples of proof systems that are sufficiently 
strong (under the arithmetization $tr$ shown above) are \ACZ$[p]$-Frege, \TCZ-Frege and 
constant-depth \IPS. 
First note that \ACZ$[p]$-Frege, for a prime $p$, 
can be viewed as a proof system for sets of polynomial 
equations in $\cT\in\pring$, where $\cT$ consists of all
standard arithmetizations of constant-depth Boolean 
formulas. And similarly, 
$\TCZ$-Frege can be considered as operating with the set 
of constant depth polynomials over the integers. 
The fact 
that \ACZ$[p]$-Frege, \TCZ-Frege are sufficiently strong
is immediate from the $\land$-introduction rules. 
The fact that constant-depth \lIPS\ is sufficiently strong (under the arithmetization $tr$)
stems from the fact that the $\land$-introduction rule
can be easily simulated in constant-depth 
\lIPS~(cf.~\cite{GP18,PT16}) (using the arithmetization scheme $tr$).

\iddolater{Check \lIPS\ simulation; we only know that \IPS\ bit not necessarily \lIPS\ is  simulated by the sources above.}

\begin{theorem}[Main barrier for Boolean instances]
\label{thm:barrier}
The functional lower bound method
(\Cref{def:general-functional-lower-bound-method}) cannot establish
lower bounds for any \textit{\emph{Boolean}} instance against sufficiently strong proof 
systems (\Cref{def:suf-strong-ps}).\footnotemark~
In particular, it cannot establish any lower bounds  against
\ACZ$[p]$-Frege, \TCZ-Frege (and constant-depth \lbIPS\ when the hard instances are Boolean).
\end{theorem}

\footnotetext{Formally, we also require that the functional
lower bound method is used on $P$, as in the proof, by lower bounding
$\cC$-\lbIPS, where $\cC\in\pring$  is closed under
partial assignments and \emph{polynomial many sum of products}.}

\begin{proof}
Let $P$ be a sufficiently strong proof system for 
the set of polynomial equations $\cC\in\pring$, and 
suppose that $\cC'$-\lbIPS\ simulates $P$, for some
$\cC\subseteq\cC'\subseteq\pring$, that 
is closed under partial assignments and 
sum of products in the following sense:
 if $h_i, g_i\in\cC'\cap \F[x_1,\dots,x_n]$, 
for $i\in I$ and $|I|=\poly(n)$, then
$\sum_{i\in I} h_i\cd g_i \in \cC'$.
Let $\cF:=\{f_i(\vx)=0\}_i$ be a collection of polynomial
equations in $\cC$, where all $f_i(\vx)$ are Boolean.

We show that the following  cannot all hold:
\begin{enumerate}
\item \label{it:1:3307}
$g\not\in\cC'$
for all $g\in\pring$ such that 
$
g(\vx)=\frac{1}{f(\vx)}, ~ \forall\vx\in\bits^n\,.
$

\item \label{it:2:3315}
There is a $\cC'$-{\rm \lbIPS}-proof of
$\cF$ from $f(\vx)=0$ and \baxioms.

\item \label{it:3:3318}
There is no $\cC'$-{\rm \lbIPS}-refutations of $\cF$.
\end{enumerate}

We show that if \Cref{it:2:3315} and \Cref{it:3:3318} above hold then \Cref{it:1:3307} does not. 

Inside $\cC'$-\lbIPS, we start with $f(\vx)$ and 
derive $\cF$ by assumption that \Cref{it:2:3315} holds. Now, derive
from $\cF$ the polynomial $1-\prod_i(1-f_i(\vx))$, 
by the assumption that $P$ is sufficiently strong and that $\cC'$-\lbIPS\ simulates $P$.
But this polynomial is always $1$ over the Boolean cube,
which implies that there is a polynomial-size circuit
$g\in\cC'$ such that 
$g(\vx)=\frac{1}{f(\vx)}, ~ \forall\vx\in\bits^n\,$,
contradicting \Cref{it:1:3307} above.
More formally, we have the following.

By \Cref{it:2:3315}, 
\begin{align}
\text{ there exist } g_i\in \cC' \text{ such that }
g_i\cd f(\vx) = f_i(\vx), \forall i.
\end{align}
By assumption  that $P$ is sufficiently strong and that $\cC'$-\lbIPS\ simulates $P$,
\begin{align}
\text{ there exist } h_i\in \cC' \text{ such that }
\sum_i h_i\cd f_i(\vx) = 1-\prod_i(1-f_i(\vx)).
\end{align}
Hence, since there is no Boolean assignment that satisfies all the Boolean axioms $f_i(\vx)$'s,  
\begin{equation}\label{eq:1397:31}
\left(\sum_i h_i\cd g_i(\vx)\right)\cd f(\vx)  = 
    1-\prod_i(1-f_i(\vx)) = 
        1 \mod \baxioms.
\end{equation}
By the assumption that $\cC'$ is closed under polynomial many
sum of products we know that $\sum_i h_i\cd g_i(\vx)\in\cC'$,
and so \Cref{eq:1397:31} contradicts \Cref{it:1:3307}. 
\end{proof}

\iddolater{GIVE EXAMPLES and talk about finite fields! Show that roABP are not closed under polynomial many 
sum of products so there may be hope there for CNF
lower bounds?}

\subsection{Conclusion}

This work wraps up to some extent research on IPS lower bounds via the functional lower bound method, showing how far it can be pushed, and where it cannot be applied. It generalises and improves previous work on IPS lower bounds obtained via the functional lower bound method in \cite{FSTW21,GHT22}. 
We established size lower bounds for symmetric instances, and hard instances qualitatively different from previously known hard instances. This allows us also to show lower bounds over finite fields, which were open. 
As a corollary, we also show a new finite field functional lower bound for \emph{multilinear formulas} which may be of independent interest.
We then demonstrate how to incorporate recent developments on constant-depth algebraic circuit lower bounds \cite{AGK0T23} in the setting of proof complexity. This enables us to improve the constant-depth IPS lower bounds in \cite{GHT22} to stronger fragments, namely IPS refutations of constant depth and $O(\log\log n)$-individual degrees. 

As for the barrier we uncovered, it is now evident that the functional lower bound method \emph{alone} cannot be used to settle the long-standing open problems about the proof complexity of  constant-depth propositional proofs with counting gates. This does not rule out however the ability of IPS lower bounds, and the IPS ``paradigm''  in general, to progress on these open problems, since  other relevant methods may be found helpful (the meta-complexity method established in \cite{ST21}, the lower bounds for multiples method \cite{FSTW21,AF22}, and the noncommutative reduction; see summary of methods at the end of \Cref{sec:first-section-in-intro}). Moreover, our barrier only shows that we cannot hope to use
a single non-Boolean unsatisfiable axiom $f(\vx)=0$ and 
consider the function $\nicefrac{1}{f(\vx)}$ over the Boolean cube
to obtain a CNF IPS lower bound (whenever the CNF is semantically
implied from $f(\vx)=0$ over the Boolean cube). However, it does not rule out in
general the use of a reduction to matrix rank, which is the backbone
of many algebraic circuit lower bounds (as well as the functional lower bound method), and should potentially be
helpful in proof complexity as well.

A very interesting  problem that remains open is to prove CNF lower bounds using the functional method against
fragments of IPS that sit below the reach of the barrier, namely fragments
that cannot derive efficiently the conjunction of arbitrarily 
many polynomials (that is, systems that are not ``sufficiently strong''
in the above terminology).  Two such proof systems are IPS operating with roABPs and multilinear formulas, respectively. 

\section{Acknowledgments} We wish to thank the anonymous reviewers of  STOC'24 for very helpful comments that improved the exposition. 

\small
\bibliographystyle{alphaurl}
\bibliography{metabib,PrfCmplx-Bakoma}


\end{document}